\newif\ifdraft
\newif\ifrealdraft
\newif\iffullversion
\newif\iftikzavailable

\tikzavailabletrue
\fullversiontrue
\realdraftfalse

\InputIfFileExists{localflags}{%
   }

\iffullversion
\documentclass[\ifrealdraft{draft}\fi]{article}
\newcommand{\ffootnote}[1]{\footnote{#1}}
\else
\documentclass[\ifrealdraft{draft}\fi]{sig-alternate-05-2015} %CCS
\newcommand{\ffootnote}[1]{}

\fi

\usepackage{cite}

\iffullversion
\author{Michael Backes\\CISPA, Saarland University \& MPI-SWS\\Saarland Informatics Campus\\~ \and Robert K\"unnemann\\CISPA, Saarland University\\Saarland Informatics Campus\\~ \and Esfandiar Mohammadi\\ETH Zurich
}
\else
\numberofauthors{3}
\author{%
    \alignauthor
    Michael Backes
    \affaddr{CISPA, Saarland University \& MPI-SWS}\\
    \affaddr{Saarland Informatics Campus}\\
    \email{backes@cs.uni-saarland.de}
    \alignauthor
    Robert K\"unnemann
    \affaddr{CISPA, Saarland University}\\
    \affaddr{Saarland Informatics Campus}\\
    \email{robert.kuennemann\\@uni-saarland.de}
    \alignauthor
    Esfandiar Mohammadi
    \affaddr{ETH Zurich}\\
    \email{mohammadi@inf.ethz.ch}
}
\fi

\usepackage{url}
\usepackage[noend]{myalgorithmic}
\usepackage{framed}
\usepackage{xspace}
\usepackage{multicol}
\usepackage{tikz}
\usepackage{pgfplots}
\usepackage{pmboxdraw}
\usetikzlibrary{calc,positioning,shapes,automata,arrows,fit}
\tikzset{state/.style={draw,rounded corners}}
\pgfplotsset{ignore legend/.style={every axis legend/.code={\renewcommand\addlegendentry[2][]{}}}}

\usepackage{calc} % compute text size for parboxes
\usepackage{color}

\iffullversion
\usepackage{amsthm,amssymb}
\usepackage[cmex10]{amsmath} % [cmex10] is required by IEEE
\usepackage[margin=1.2in]{geometry}
\fi

\usepackage{multirow}
\usepackage{graphicx}
\usepackage{bm}
\usepackage[inference]{semantic}
\usepackage{relsize}
\usepackage[utf8]{inputenc}
\usepackage[T1]{fontenc}
\usepackage{hyperref}
\usepackage{listings} %TODO can be removed, only used for java listing.
\usepackage{paralist}

\newtheorem{corollary}{Corollary}
\newtheorem{lemma}{Lemma}
\newtheorem{theorem}{Theorem}
\newtheorem{definition}{Definition}
\definecolor{darkgreen}{rgb}{0,.5,0}
\definecolor{darkpurple}{rgb}{.3,0,.6}
\IfFileExists{./whoami.txt}{
\usepackage[users=MER]{uchanges}
	\ChangesSetUserColor{M}{red} % Michael
	\ChangesSetUserColor{E}{darkgreen} % Esfandiar
	\ChangesSetUserColor{R}{blue} % Robert
}{
	
	\newcommand{\TODOE}[1]{\ifdraft[{\small\color{darkgreen} \textsc{TODO:} ##1}]\fi}
	\newcommand{\TODOR}[1]{\ifdraft[{\small\color{blue} \textsc{TODO:} ##1}]\fi}

	\newcommand{\REMOVEE}[1]{}
	\newcommand{\REMOVEEfor}[3]{}
	\newcommand{\NOTEE}[1]{\ifdraft[{\small\color{darkgreen} \textsc{NOTE:} ##1}]\fi}
}

\newcommand{\mathcmd}[1]{{\normalfont\ensuremath{#1}}\xspace}

\newcommand{\mathstring}[1]{\mathcmd{\text{\texttt{\smaller{}#1}}}}
\newcommand{\mathname}[1]{\mathcmd{\text{\textrm{#1}}}}

\newcommand{\mathfun}[1]{\mathcmd{\mathit{#1}}}

\newcommand{\mathset}[1]{\mathname{#1}}
\newcommand{\mathsymb}[1]{\mathfun{#1}}

\newcommand{\progcmdraw}[1]{\mathcmd{\mathsf{#1}}}
\newcommand{\progcmd}[1]{\mathcmd{\progcmdraw{#1}~}}
\newcommand{\textop}[1]{\relax\ifmmode\mathop{\text{#1}}\else\text{#1}\fi}

\newcommand{\ie}{i.e.}

\newcommand{\eg}{e.g.}

\newcommand{\etAl}{et al.\xspace}

\newcommand{\invokestatic}{\progcmd{invoke\text{-}static}}
\newcommand{\invokedirectr}{\progcmd{invoke\text{-}direct\text{-}range}}

\newcommand{\Mal}{\ensuremath{\mathit{Mal}}\xspace}

\newcommand{\attacker}{\mathcmd{\mathcal{A}}}
\newcommand{\adv}{\attacker}
\newcommand{\advE}{\attacker}
\newcommand{\advrespDV}{\mathcal{ADVR}}
\newcommand{\advss}{\mathcal{Q}}
\newcommand{\advM}{\mathcal{ADV}}

\newcommand{\set}[1]{\mathcmd{\left\{ #1 \right\}}}
\newcommand{\subst}[1]{\set{#1}}
\newcommand\bit{\lbrace0,1\rbrace}
\newcommand\bits[1]{\bit^{#1}}

\newcommand{\interact}[2]{\ensuremath{\bm{\langle} #1 \| #2 \bm{\rangle}}}

\newcommand{\false}{\mathsf{false}}

\newcommand{\tic}[1][]{\mathrel{\approx^{#1}_\textit{tic}}}

\makeatletter
\DeclareRobustCommand{\defeq}{\mathrel{\rlap{%
  \raisebox{0.3ex}{$\m@th\cdot$}}%
  \raisebox{-0.3ex}{$\m@th\cdot$}}%
  =}
\DeclareRobustCommand{\eqdef}{=\mathrel{\rlap{%
  \raisebox{0.3ex}{$\m@th\cdot$}}%
  \raisebox{-0.3ex}{$\m@th\cdot$}}%
  }
\makeatother

\mathchardef\mhyphen="2D

\newcommand{\n}{{\nu}}

\newcommand{\termsof}{\mathit{Out}}
\newcommand{\inputof}{\mathit{In}}

\newcommand{\outmetaof}{\mathit{Out\mhyphen Meta}}
\newcommand{\testsof}{\mathit{K}}

\newcommand{\secpar}{\mathcmd{\eta}}

\newcommand{\voidDV}{\mathcmd{\mathrm{void}}}
\newcommand{\resultUpperDV}{\mathcmd{\mathfun{res}_{\mathit{up}}}}
\newcommand{\resultLowerDV}{\mathcmd{\mathfun{res}_{\mathit{lo}}}}
\newcommand{\lowerDV}{\mathfun{lo}}
\newcommand{\upperDV}{\mathfun{up}}
\newcommand{\dom}{\mathfun{dom}}
\newcommand{\class}{\mathfun{cls}}

\newcommand{\restrict}[2]{\mathcmd{\left.{#1}_{{}}\right|_{{#2}}}}

\newcommand{\embedding}{e}
\newcommand{\AlgoNameString}{E}
\newcommand{\AlgoName}[1][T_H]{\mathcmd{\AlgoNameString_{{#1}(s)}}}
\newcommand{\postext}{pos}
\newcommand{\pos}{\mathcmd{\mathtt{\postext}}}

\newcommand{\reduction}{\rightarrow}%{\rightsquigarrow}
\newcommand{\predRel}[2]{\mathrel{\mathcmd{\xrightarrow{#2}_{#1}}}}
\newcommand{\redRel}[1]{\mathcmd{\predRel{#1}{~}}}
\newcommand{\redRelSt}{\mathrel{\redRel{~}}}

\newcommand{\symtrans}[2][\proto]{\ensuremath{\xrightarrow{#2}_{Sym}}\xspace}
\newcommand{\cospsym}{\rightsquigarrow_\mathit{CoSPs}}
\newcommand{\cospcomp}{\rightsquigarrow_\mathit{CoSPc}}

\newcommand{\advs}{\mathit{as}}
\newcommand{\advsd}{\mathit{as}_\mathit{dmy}}
\newcommand{\stateDV}[1][m\cdot ml,h,pp\cdot ppl,r\cdot rl,\advs]{\left\langle #1 \right\rangle}
\newcommand{\stateDVSt}[1]{\stateDV[m{#1}\cdot ml,h{#1},pp{#1}\cdot ppl,r{#1}\cdot rl,\advs{#1}]}
\newcommand{\stateDVs}{\stateDV[{ml,h,ppl,rl}]}
\newcommand{\stateDVsN}[1]{\stateDV[{ml{#1},h{#1},ppl{#1},rl{#1}}]}
\newcommand{\stateDVsSt}[1]{\stateDV[m{#1}\cdot ml{#1},h{#1},pp{#1}\cdot ppl{#1},r{#1}\cdot rl{#1}]}
\newcommand{\wait}{\mathtt{wait}}

\newcommand{\ADL}{\mathcmd{\text{ADL}}}
\newcommand{\ADLo}{\mathcmd{\text{ADLo}}^\mathit{SS}}
\newcommand{\ADLs}{\mathcmd{\text{ADLs}}}
\newcommand{\ADLss}{\mathcmd{\text{ADL}}^\mathit{SS}}
\newcommand{\ADLsss}{\mathcmd{\text{ADLs}}^\mathit{SS}}
\newcommand{\CoSP}{\mathcmd{\text{CoSP}}}

\newcommand{\executionCategory}{\mathcmd{\text{SS}}}
\newcommand{\equivalent}[1]{\mathcmd{\approx^{#1}_s}}
\newcommand{\notEquivalent}[1]{\mathcmd{{\not\approx}^{#1}_s}}
\newcommand{\indistinguishable}[1]{\mathcmd{\approx^{#1}_c}}
\newcommand{\notIndistinguishable}[1]{\mathcmd{{\not\approx}^{#1}_c}}
\newcommand{\CoSPEquivalent}{\equivalent{\CoSP}}
\newcommand{\notCoSPEquivalent}{\notEquivalent{\CoSP}}
\newcommand{\CoSPIndistinguishable}{\indistinguishable{\CoSP}}
\newcommand{\notCoSPIndistinguishable}{\notIndistinguishable{\CoSP}}
\newcommand{\execEquivalent}{\equivalent{\executionCategory}}
\newcommand{\notExecEquivalent}{\notEquivalent{\executionCategory}}
\newcommand{\execIndistinguishable}{\indistinguishable{\executionCategory}}
\newcommand{\notExecIndistinguishable}{\notIndistinguishable{\executionCategory}}
\newcommand{\ADLEquivalent}{\equivalent{\ADL}}
\newcommand{\notADLEquivalent}{\notEquivalent{\ADL}}
\newcommand{\ADLIndistinguishable}{\indistinguishable{\ADL}}

\newcommand{\SSIndistinguishable}[1]{\indistinguishable{#1}}
\newcommand{\OutN}{\ensuremath{\mathit{OutN}}}
\newcommand{\InN}{\ensuremath{\mathit{InN}}}
\newcommand{\CompN}{\ensuremath{\mathit{CompN}}}

\newcommand{\implementation}{\mathname{Impl}}

\newcommand{\proto}{\mathcmd{\Pi}} % protocol
\newcommand{\protosss}[3][\secpar]{\mathcmd{\proto^{#1}#2\hspace{-0.4ex}\stateDV[{#3}]}}
\newcommand{\protoss}[2][\secpar]{\mathcmd{\protosss[#1]{_#2}{s_{#2}}}}
\newcommand{\protos}[1]{\mathcmd{\protoss[~]{#1}}}

\newcommand{\initState}[1][]{\mathcmd{s}#1}
\newcommand{\initH}[1]{\mathcmd{\initState[_{H,#1}^0]}}
\newcommand{\initL}{\mathcmd{\initState[_L]}}

\newcommand{\stringInput}{\mathstring{in}}
\newcommand{\stringOutput}{\mathstring{out}}
\newcommand{\stringControl}{\mathstring{control}}

\newcommand{\transitionAHIdentifierRaw}{in}
\newcommand{\transitionHAIdentifierRaw}{out}
\newcommand{\transitionHLIdentifierRaw}{LC}
\newcommand{\transitionLHIdentifierRaw}{LR}
\newcommand{\transitionAHIdentifier}{\mathstring{\transitionAHIdentifierRaw}}
\newcommand{\transitionHAIdentifier}{\mathstring{\transitionHAIdentifierRaw}}
\newcommand{\transitionHLIdentifier}{\mathstring{\transitionHLIdentifierRaw}}
\newcommand{\transitionLHIdentifier}{\mathstring{\transitionLHIdentifierRaw}}
\newcommand{\tiAH}{\transitionAHIdentifier}
\newcommand{\tiHA}{\transitionHAIdentifier}
\newcommand{\tiHL}{\transitionHLIdentifier}
\newcommand{\tiLH}{\transitionLHIdentifier}
\newcommand{\tifinal}{\mathstring{final}}
\newcommand{\tifinalCall}{\mathstring{finalCall}}

\newcommand{\setN}{\mathbb{N}}

\newcommand{\states}{\mathset{S}}
\newcommand{\initstates}{\mathset{S}^0}
\newcommand{\actions}{\mathset{A}}

\newcommand{\Const}{\mathcmd{\mathbf{C}}}
\newcommand{\C}{\Const}
\newcommand{\N}{\mathbf{N}}
\newcommand{\NE}{\mathbf{N_E}}
\newcommand{\NP}{\mathbf{N_P}}
\newcommand{\Dest}{\mathcmd{\mathbf{D}}}
\newcommand{\D}{\Dest}
\newcommand{\Terms}{\mathcmd{\mathbf{T}}}
\newcommand{\T}{\Terms}
\newcommand{\Model}{\mathcmd{\mathbf{M}}}
\newcommand{\M}{\Model}
\newcommand{\Nonce}{\mathcmd{\mathbf{N}}}
\newcommand{\nonce}{\Nonce}
\newcommand{\setSymbolicOperations}{\mathset{SO}}
\newcommand{\SO}{\setSymbolicOperations}

\newcommand{\Views}{\mathset{SViews}}
\newcommand{\traces}{\ensuremath{\mathit{traces}}\xspace}
\newcommand{\Traces}{\ensuremath{\mathit{Traces}}\xspace}
\newcommand{\supp}{\ensuremath{\mathit{supp}}\xspace}
\newcommand\protoClass{\mathrm{P}}
\newcommand{\Eventin}{\ensuremath{\mathit{Event}_\mathit{in}}\xspace}
\newcommand{\Eventout}{\ensuremath{\mathit{Event}_\mathit{out}}\xspace}
\newcommand{\Eventctl}{\ensuremath{\mathit{Event}_\mathit{ctl}}\xspace}
\newcommand{\attstrat}{\ensuremath{\mathit{AttS}}\xspace}

\newcommand{\succs}{\ensuremath{\mathit{succ}}\xspace}

\newcommand{\implementations}{\mathname{\bf Impl}}

\newcommand{\fieldsDV}{\mathcal{F}}
\newcommand{\locationsDV}{\mathcal{L}}
\newcommand{\valuesDV}{\mathcal{V}}
\newcommand{\objectsDV}{\mathcal{O}}
\newcommand{\heapsDV}{\mathcal{H}}
\newcommand{\arraysDV}{\mathcal{AR}}
\newcommand{\methodsDV}{\mathcal{M}}
\newcommand{\midDV}{\ensuremath{\mathcal{MID}}\xspace}

\newcommand{\calC}{\ensuremath{\mathcal{C}}\xspace}
\newcommand{\calD}{\ensuremath{\mathcal{D}}\xspace}

\newcommand{\calN}{\ensuremath{\mathcal{N}}\xspace}
\newcommand\calP{\ensuremath{\mathcal{P}}\xspace}
\newcommand\calR{\ensuremath{\mathcal{R}}\xspace}
\newcommand\calS{\ensuremath{\mathcal{S}}\xspace}

\newcommand\calV{\ensuremath{\mathcal{V}}\xspace}

\newcommand{\rReturnVoidF}{\textsc{rReturnVF}\xspace}
\newcommand{\rReturnVoid}{\textsc{rReturnV}\xspace}
\newcommand{\rReturnF}{\textsc{rReturnF}\xspace}
\newcommand{\rReturn}{\textsc{rReturn}\xspace}
\newcommand{\rIfTestT}{\textsc{IfTestT}\xspace}
\newcommand{\rIfTestF}{\textsc{IfFestF}\xspace}
\newcommand{\rIDR}{\textsc{rIDR}\xspace}
\newcommand{\Prob}{\textsc{Prob}\xspace}

\newcommand{\AdvInv}{\textsc{AdvInv}\xspace}
\newcommand{\AdvRet}{\textsc{AdvRet}\xspace}
\newcommand{\AdvFin}{\textsc{AdvFin}\xspace}

\newcommand{\LibCall}{\textsc{LibCall}\xspace}
\newcommand{\LibResponse}{\textsc{LibResponse}\xspace}
\newcommand{\LeakMsg}{\textsc{LeakMsg}\xspace}
\newcommand{\ReceiveMsg}{\textsc{ReceiveMsg}\xspace}
\newcommand{\FinalCall}{\textsc{FinalCall}\xspace}

\newcommand{\LLibCall}{\textsc{LLibCall}\xspace}
\newcommand{\LLibRetVoid}{\textsc{LLibRetVoid}\xspace}
\newcommand{\LLibRet}{\textsc{LLibRet}\xspace}

\newcommand{\ALeakMsg}{\textsc{ALeakMsg}\xspace}
\newcommand{\AReceiveMsg}{\textsc{AReceiveMsg}\xspace}
\newcommand{\AFinal}{\textsc{AFinal}\xspace}

\newcommand{\tl}{\mathfun{tl}}

\newcommand{\eval}{\mathfun{eval}}

\newcommand{\lookupVirtualDV}{\mathfun{lookup\text{-}virtual}}
\newcommand{\lookupStaticDV}{\mathfun{lookup\text{-}static}}
\newcommand{\lookupDirectDV}{\mathfun{lookup\text{-}direct}}
\newcommand{\lookupSuperDV}{\mathfun{lookup\text{-}super}}
\newcommand{\defaultRegistersDV}{\mathfun{defReg}}
\newcommand{\transitionAH}{\mathset{in}}
\newcommand{\transitionHA}{\mathset{out}}
\newcommand{\transitionHL}{\mathset{libCall}}
\newcommand{\transitionLH}{\mathset{libResp}}
\newcommand{\tAH}{\transitionAH}
\newcommand{\tHA}{\transitionHA}
\newcommand{\tHL}{\transitionHL}
\newcommand{\tLH}{\transitionLH}
\newcommand{\tfinal}{\mathset{final}}

\newcommand{\libSpec}{\ensuremath{\mathfun{libSpec}}\xspace}
\newcommand{\sliceof}[2][h]{\ensuremath{{#1}|_{#2}}\xspace}
\newcommand{\lreachabe}[1][h]{\ensuremath{\mathfun{lreachable}_{#1}}\xspace}

\newcommand{\equals}{\mathsymb{equals}}
\newcommand{\iszero}{\mathsymb{iszero}}
\newcommand{\payloadString}{\mathsymb{string}}
\newcommand{\payloadUnstring}{\mathsymb{unstring}}
\newcommand{\payloadEmpty}{\mathsymb{emp}}
\newcommand{\pair}{\mathsymb{pair}}
\newcommand{\fst}{\mathsymb{fst}}
\newcommand{\snd}{\mathsymb{snd}}
\newcommand{\enc}{\mathsymb{enc}}

\newcommand{\Exec}{\mathname{Exec}}

\newcommand{\pparagraphraw}[1]{\par \smallskip\noindent\textbf{#1}}
\newcommand{\pparagraph}[1]{\pparagraphraw{#1.}\;}
\iffullversion
\newcommand{\pparagraphorsubsection}[1]{\subsection{#1}}
\newcommand{\fullversionfootnote}[1]{\footnote{#1}}
\else
\newcommand{\pparagraphorsubsection}[1]{\pparagraph{#1}}
\newcommand{\fullversionfootnote}[1]{}
\fi

\newenvironment{framedFigure}[2]
{
\begin{figure*}[t]
\def\tmpCaption{#1}
\def\tmpLabel{#2}
\begin{framed}
}
{
\vspace{-0.5em}
\end{framed}
\vspace{-1em}
\caption{\tmpCaption}\label{\tmpLabel}
\end{figure*}
}

\newcommand{\codesize}{\footnotesize}

\newenvironment{framedFigureSingle}[2]
{
\begin{figure}[t]
\def\tmpCaption{#1}
\def\tmpLabel{#2}
\begin{framed}
}
{
\vspace{-0.5em}
\end{framed}
\vspace{-1em}
\caption{\tmpCaption}\label{\tmpLabel}
\end{figure}
}

\newenvironment{imitate}[1]
{

\medskip
\noindent\textbf{#1.}
\bgroup
\it
}
{
\egroup
}

\newenvironment{Claim}[1]
{
\begin{imitate}{Claim~{#1}}
}
{
\end{imitate}
}

\newenvironment{claimproof}[1]
{

\medskip
\noindent\textit{Proof of Claim~#1.}
\bgroup
}{
\hfill $\diamond$
\egroup

\par
}

\newcounter{myExampleCounter}
\refstepcounter{myExampleCounter}

\newenvironment{example}[1][]
{
  \refstepcounter{myExampleCounter}\smallskip\noindent \emph{Example \themyExampleCounter\ifthenelse{\equal{#1}{}}{}{: #1}.}% \textbf{Example \themyExampleCounter: #1.}
}
{\hfill$\diamond$% \smallskip
}

\newenvironment{inferenceRules}
{
\codesize
\begin{math}
\begin{array}{c}
}
{
\end{array}
\end{math}
}

\begin{document}

% \pagestyle{plain}

%\title{Canonical Embedding into CoSP}
% \title{Computational Sound Libraries Taken To the Next Level: Computationally Sound Symbolic Operational Semantics}
\title{Computational Soundness for Dalvik Bytecode}
 % \\ {\Large Computationally Sound Symbolic Operational Semantics}

\date{}

\maketitle

\begin{abstract}
Automatically analyzing information flow within Android applications that rely on cryptographic operations with their computational security guarantees
 imposes formidable challenges that existing approaches for understanding an app's
behavior  struggle to meet. These approaches do not distinguish cryptographic and non-cryptographic operations, and hence
do not account for cryptographic protections:
$f(m)$ is considered sensitive for a sensitive message $m$ irrespective of potential secrecy properties offered by a cryptographic operation $f$.
These approaches consequently provide a safe approximation of the app's behavior, but they  mistakenly classify a large fraction of apps as potentially insecure and consequently yield overly pessimistic results.

In this paper, we show how cryptographic operations can be faithfully included into existing approaches for automated app analysis.
To this end, we first show how cryptographic operations can be expressed as symbolic abstractions within the comprehensive Dalvik bytecode language. These abstractions
are accessible to automated analysis and can be conveniently added to existing app analysis tools using minor changes in their
semantics. 
Second, we show that our abstractions are faithful by providing the first computational soundness result for Dalvik bytecode, i.e., the absence of attacks against our symbolically abstracted program entails the absence of any attacks against a suitable cryptographic program
realization. We cast our computational soundness result in the CoSP framework, which makes the result modular and composable.
\end{abstract}
\iffullversion\else
\keywords{Android; Computational Soundness; Secure Information Flow}
\fi

\iffullversion
\clearpage
\tableofcontents
\clearpage
\fi

%%%%%%%%%%%%%%%%%%%%%%%%%%%%%%%%%%%%%%%%%%
%%%%%%%%%%%%%%%%%%%%%%%%%%%%%%%%%%%
%%% The main body begins here 
%%%%%%%%%%%%%%%%%%%%%%%%%%%%%%%%%%%

%!TEX root = main.tex

\section{Introduction}

Android constitutes an open-source project not only in terms of source code but also in terms of the whole ecosystem, allowing practically everyone
to program new apps and make them publicly available in Google Play.  
This open nature of Android has facilitated a rapid pace of innovation, but it has also led to the creation and widespread deployment of malicious apps~\cite{int-sec-threat-rep:symantec15,mobile-malware:gdata15}. 
Such apps often cause  privacy  violations that  leak  sensitive  information such as location or the user's address book, either
as an intended functionality or as a result of uninformed programming. In some cases such apps can even extract sensitive information from honest apps. 

A comprehensive line of research has, hence, strived to rigorously analyze how apps are accessing and processing sensitive information. These approaches typically
employ the concept of  information flow control (IFC), i.e., certain information sources such as GPS position and address book are declared to be sensitive, and certain
information sinks are declared to be adversarially observable. An IFC-based analysis then traces the propagation of sensitive information through the program, i.e.,  if sensitive data $m$ is 
input to a function $f$, then the result $f(m)$ is considered sensitive as well. IFC-based analyses thereby determine if information from sensitive sources can ever reach an observable sink, and in that
case report a privacy violation.
 
A considerable number of apps rely on cryptographic operations, e.g., for encrypting sensitive information before it is sent over the Internet. However, 
analyzing information flow within Android apps that rely on such cryptographic operations with their computational security guarantees imposes formidable challenges that all existing 
approaches for automated app analysis 
struggle to meet, e.g.,~\cite{EnGiChCoJuDaSh_14:taintdroid,LoMaStBaScWe_14:cassandra,BaBuDeGeHa_15:r-droid}. Roughly, these approaches  do not distinguish cryptographic operations from other, non-cryptographic functions. Thus, the standard information-tracing mechanism for arbitrary functions applies:
$f(m)$ is considered sensitive for a sensitive message $m$ irrespective of potential secrecy properties offered by a cryptographic function $f$, e.g., 
the encryption of a sensitive message $m$ is still considered sensitive such that sending this encryption over the Internet is considered a privacy breach.
These approaches consequently provide a safe approximation of the app's behavior, but they mistakenly classify a large fraction of apps as potentially insecure and consequently yield \emph{overly pessimistic results}. While approaches based on manual declassification have successfully managed to treat cryptographic
operations and their protective properties more accurately, see the section on related work for more details, no concept for an accurate cryptographic treatment is
known for automated analysis of Android apps.

\iffullversion\else\clearpage\fi
\subsection{Our Contributions}
In this paper, we show how cryptographic operations can be faithfully included into existing approaches for automated app analysis on Android, in the presence of malicious apps or network parties aiming to extract sensitive information from honest parties.
Our paper makes two main tangible contributions to this field: (i)  we show how cryptographic operations can be expressed as symbolic abstractions within Dalvik bytecode, so
that  existing automated analysis tools can adopt them with only minor changes in their
semantics; and (ii) we show that our abstractions are faithful by providing the first computational soundness result for the comprehensive Dalvik bytecode
language, i.e., the absence of attacks against our symbolically abstracted program entails the absence of any attacks against a suitable cryptographic program
realization. 

\pparagraph{Symbolic abstractions in Dalvik bytecode}
We first show how cryptographic operations  can be expressed as symbolic abstractions within Dalvik bytecode.  These symbolic abstractions -- often  also 
referred to as perfect cryptography or Dolev-Yao models --  constitute idealizations of cryptographic operations as free algebras that lend themselves towards automated analysis. 
Deriving such abstractions within the comprehensive Dalvik bytecode language constitutes a challenging task, since existing formalizations of Dalvik do not offer
a distinction between honest and adversarially controlled components, which is crucial for defining the  rules that the symbolic adversary has to adhere to. To this end, we
develop a novel semantic characterization of Dalvik bytecode that we
call \emph{split-state semantics} that provides a clear separation between honest program parts and cryptographic API calls with their corresponding augmented adversarial symbolic capabilities.
%as well as the existing non-cryptographic part of a given program. We consider this representation of independent interest.
%This representation in particular prevents mistakes that could be made in the definition of a symbolic version of a programming language, or induced
%by inappropriate manual declassifications. \textbf{ToDo-MB: Drop the last sentence?} 
Moreover, this split-state form is key to our proof of computational soundness, see below.
Existing tools for automated app analysis can conveniently include our abstractions using minor changes in their underlying semantics, and thereby reason more accurately about cryptographic operations.
% \textbf{ToDo-MB: Do we have one good senteERE QUICKLY WHAT NEEDS TO BE CHANGES IN 1-2 SENTENCES.}

%
%
%\textbf{MB: Where to put the following information? Drop if the respective section goes away!}
%Second, we show how our result gives rise to a computationally sound declassification policy that is parametric in the symbolic model (proven sound in CoSP). The policy declassifies all terms that are equivalent in the sense of the computationally sound symbolic model. Our policy coincides with encryption-based declassification together with a key-release declassification policy, for a computationally sound symbolic model that includes encryption and enables the leakage of the key (e.g.,~\cite{BaMaUn_12:without-restrictions-cosp,BaMoRu_14:equivalence}). As a result, generalizes previous work on gradual release~\cite{AsSa_07:gradual-release} in a computationally sound manner.

\pparagraph{Computational soundness for Dalvik bytecode}
We show that our symbolic abstractions can be securely instantiated using suitable cryptographic primitives, and thereby 
provide the first computational soundness result for Dalvik bytecode. More specifically,
our result is grounded in the Abstract Dalvik Language (ADL)~\cite{LoMaStBaScWe_14:cassandra}, which constitutes the currently  most detailed and comprehensive operational
semantics for Dalvik in the literature. To this end, we first extended ADL by probabilistic choices, as it otherwise would be inappropriate to express cryptographic operations.

We cast our computational soundness result in CoSP, a framework for establishing computational soundness results that decouples the process of embedding programming languages into CoSP from the computational soundness proofs itself. 
In particular, by casting our soundness results in CoSP, a rich body of computational soundness results for individual cryptographic primitves~\cite{BaHoUn_09:cosp,Me_10:ind-cpa-cosp,BaMaUn_12:without-restrictions-cosp,BaBeUn_13:cosp-zk,BaMoRu_14:equivalence,BaBeMaMoPe_15:malleable-zk} is immediately valid for Dalvik bytecode
without any additional work.

Establishing computational soundness results for Dalvik bytecode imposed a series of technical challenges that many prior computational soundness
works did not have to cope with. We highlight one such challenge. Computational soundness results struggle
when confronted with situations in which binary operations are applied to outputs of cryptographic operations, e.g., if the parity of a ciphertext should be checked,
 since such an operation would be undefined in the symbolic setting.
    Prior computational soundness results simply excluded programs with such illegitimate operations; this exclusion introduced an additional proof obligation for the automated analysis
    tool. While excluding such programs simplifies the soundness result, integrating these additional proof obligations
     in existing automated app analysis tools constitutes a tedious task, since the tools would need to check upfront whether any symbolically undefined operations will be performed on symbolic terms, in any execution branch.  As a consequence, we hence decided to establish a computational soundness result that
	over-approximates such scenarios by sending information of such illegitimate operations to the adversary and
    letting the adversary decide the result of such operations.

Finally, our proof reveals an additional result that we consider of independent interest: we show that any small-step semantics $S$ in split-state form entails a canonical small-step semantics $S^*$ for a symbolic model that is computationally sound with respect to $S$. Hence, for establishing a computationally sound symbolic abstraction for \emph{any} given programming language, it suffices to show that the interaction with the attacker and the cryptographic API can be expressed by means of our concept of split-state semantics.

\subsection{Summary of Our Techniques}
This section summarizes the techniques that we use to achieve these results. 
\iffullversion{We believe that this makes the paper better accessible.}
\else
{We hope that this summary makes the paper better accessible, given that it was distilled from a technical report spanning over
roughly 50 pages~\cite{ourfullversion}.}
\fi
Finally, we discuss how our results can be used to extend information flow tools.
\pparagraph{The CoSP framework% (Section \ref{section:cosp})
}
% Our characterization of the attacker's knowledge uses so-called symbolic abstractions for cryptographic operations. Symbolic abstractions of cryptographic operations represent cryptographic messages (i.e., the results of cryptographic operations) as terms: an encryption of a message $m$ with a key $k$ would, e.g., be represented as $\mathstring{enc(k,m)}$. Crucially, a symbolic abstraction typically comprises a concise characterization of the attacker's capabilities by a small set of rules: the fact that information about an honestly generated ciphertext can only be learned if the key is known is represented as a rule $\forall \mathstring{k},\mathstring{m}. \mathfun{dec}(\mathstring{k}, \mathstring{enc(k,m)}) = \mathstring{m}$, which defines a symbolic decryption function $\mathfun{dec}$. Naturally, a symbolic abstraction only makes sense if it captures all cryptographic attacks; we say that a such symbolic abstraction is computationally sound. Due to the concise characterization of the attacker's capabilities, symbolic abstractions make it possible to automatically predict a cryptographic attacker's worst case actions, thereby enable automated security verification techniques~\cite{proverif, tamarin, avispa}.
% The present work uses the CoSP framework as a unified platform for computationally sound symbolic abstractions.
% The guiding theme of this work is to
A central idea of our work is to reduce computational soundness of Dalvik bytecode  to computational soundness in the CoSP framework~\cite{BaHoUn_09:cosp,BaMoRu_14:equivalence}.
 All definitions
in CoSP are cast relative to a \emph{symbolic model} that specifies a set of constructors and destructors that
symbolically represent the cryptographic operations and are also used for characterizing the terms that the attacker can derive (called \emph{symbolic attacker knowledge}), and a \emph{computational implementation} that specifies cryptographic algorithms for these constructors and destructors. In CoSP, a \emph{protocol} is represented
by an infinite tree that describes the protocol as a labeled
transition system. Such a CoSP protocol contains actions
for performing abstract computations (applying constructors
and destructors to messages) and for communicating
with an adversary. 
%All other steps, e.g., internal computations that are not related to the cryptographic library, are represented by an poly-time algorithm that computes the subsequent node for every prefix %of nodes that starts at the root node.
A CoSP protocol is equipped with two
different semantics: (i) a \emph{symbolic CoSP execution}, in which messages are represented  by terms;
 and (ii) a \emph{computational CoSP execution}, in which messages are
bitstrings, and the computational implementation is used
instead of applying constructors and destructors.  
A computational implementation is said to be \emph{computationally sound}
for a class of security properties if any CoSP protocol that satisfies these properties in the symbolic execution
also satisfies these properties in the computational execution.
The advantage of expressing computational soundness results in CoSP is that  the protocol
model in CoSP is very general so that the semantics of other languages can be embedded therein, thereby
transferring  the various established soundness results from CoSP into these languages~\cite{BaHoUn_09:cosp,Me_10:ind-cpa-cosp, BaMoRu_14:equivalence,BaMaUn_12:without-restrictions-cosp, BaBeUn_13:cosp-zk,BaBeMaMoPe_15:malleable-zk}. 

\iftikzavailable
\begin{framedFigure}{Overview of the main technical lemmas, where $e$ is the embedding into CoSP}{fig:overview-figure}
\begin{center}
	\hspace{-3em}
	\small
\begin{tikzpicture}[node distance=5.8cm,>=stealth,bend angle=45,auto]
	\tikzstyle{edge} = [draw,thick,->]
  \begin{scope}
    \node	(ADL-symb)
			{\parbox{3.8cm}{\centering ADL\\ Semantics\\ $\protos 1 \ADLEquivalent \protos 2$}};
    \node	(ADL-symbolic-execution) [right of= ADL-symb]
			    % \node 	(ADL-symbolic-execution)		[right of= ADL-symb]
			{\parbox{3.5cm}{\centering Symbolic\\ Split-State\\ $\protos 1 \execEquivalent \protos 2$}};
    \node	(ADL-symbolic-cosp) 			[right of= ADL-symbolic-execution]
			{\parbox{3.5cm}{\centering Symbolic CoSP\\ Execution\\ $e(\protos 1) \CoSPEquivalent e(\protos 2)$}};
    \node 	(ADL)							[below of= ADL-symb, node distance=2.8cm]
			{\parbox{3.5cm}{\centering ADL\\ Semantics\\ $\protos 1 \ADLIndistinguishable \protos 2$}};
    \node	(ADL-computational-execution)	[right of= ADL]
			{\parbox{3.5cm}{\centering Computational\\ Split-State\\ $\protos 1 \execIndistinguishable \protos 2$}};
    \node	(ADL-computational-cosp) 		[right of= ADL-computational-execution]
			{\parbox{3.5cm}{\centering Computational CoSP\\ Execution\\ $e(\protos 1) \CoSPIndistinguishable e(\protos 2)$}};
	
	\draw[thick,double,->] ($(ADL-symb)+(2,0)$) -- node {Lemma~\ref{lemma:symbolic-adl-symbolic-execution}} ($(ADL-symbolic-execution)-(1.5,0)$);
	\draw[thick,double,->] ($(ADL-symbolic-execution)+(1.7,0)$) -- node {Lemma~\ref{lemma:embedding-soundness-symbolic}} ($(ADL-symbolic-cosp)-(1.8,0)$);
        \draw[thick,double,->] ($(ADL-symbolic-cosp)+(0,-0.8)$) -- node [left] {\parbox{\widthof{Soundness in CoSP}}{\centering Computational Soundness in CoSP}} ($(ADL-computational-cosp)+(0,0.8)$);
	\draw[thick,double,->] ($(ADL-computational-cosp)-(1.8,0)$) -- node {Lemma~\ref{lemma:embedding-soundness-computational}} ($(ADL-computational-execution)+(1.7,0)$);
	\draw[thick,double,->] ($(ADL-computational-execution)-(1.5,0)$) -- node {Lemma~\ref{lemma:computational-execution-adl}} ($(ADL)+(2,0)$);
	\draw[dotted,double,->] ($(ADL-symb)+(0,-0.8)$) -- node {Theorem~\ref{theorem:symbolic-adl-computational-adl}} ($(ADL)+(0,0.8)$);
	
  \end{scope}
\end{tikzpicture}
\end{center}
\TODOE{Why do we write down the initial configuration as $\Pi\langle s\rangle$? Why do we for the moment not consider two transition systems that only differ in the initial configurations as different transition systems? We could then later, when we need it (in the comparisons to information flow methods), define the notation that $\Pi\langle s\rangle$ to replace the initial configuration with $s$?}
\end{framedFigure}
\fi % if tikzavailable

\pparagraph{Symbolic ADL and probabilistic choices}
ADL as defined in \cite{LoMaStBaScWe_14:cassandra} does not support probabilistic choices, and hence no generation of cryptographic keys
and no  executions of cryptographic functions. We thus extended ADL with a rule that uniformly samples random values (more concretely: a register value from
the set of numerical values). 
%Since we have to operate on asymptotic security definition, we require that cryptographic functions  belong to a family of functions that is parametrized by the security parameter. This %requirement matches the history of cryptographic implementations that were updated over time while the security parameter also increased over time.

We consider attackers that are external to the app, e.g., malicious parties or network parties with the goal of extracting secrets from an honest app. We characterize the interaction points of the attacker with our extended version of ADL by a set of so-called \emph{malicious functions} that communicate with the attacker. The attacker itself is modeled as a probabilistic polynomial-time machine. We introduce an additional semantic rule that is applied whenever a malicious function is called. The rule invokes the attacker with the arguments of the function call and stores the response of the attacker as a return value. For defining the indistinguishability of two ADL programs, we additionally require that the adversary can also send a single bit $b$ as a final guess, similar to other indistinguishability definitions. This entails the notions of \emph{symbolic equivalence} ($\ADLEquivalent$, in the symbolic setting) and of \emph{computational indistinguishability} ($\ADLIndistinguishable$, in the computational
setting) of two ADL programs.

\pparagraph{Split-state semantics for symbolic ADL}
Establishing a computational soundness proof for ADL requires  a clear separation between honest program parts and cryptographic API calls with their corresponding augmented adversarial symbolic capabilities. To achieve this, we characterize this partitioning by introducing the concept of a \emph{split-state form} of an operational semantics. The split-state form 
partitions the original semantics into three components, parallelly executed asynchronously: $(i)$ all steps that belong to computing cryptographic operations (called the \emph{crypto-API semantics}), $(ii)$ all steps that belong to computing the malicious functions (called the \emph{attacker semantics}), and $(iii)$ all steps that belong to the rest of the program (called the \emph{honest-program semantics}). Moreover, we define explicit transitions between each of these components, which gives rise to a precise message-passing interface for 
 cryptographic operations and for communicating with the attacker.

Our strategy for showing computational soundness is to use this split-state form for phrasing the symbolic variant as a small-step semantics by replacing the crypto-API semantics with the symbolic constructors and destructors from the symbolic model, and by replacing the attacker semantics by the symbolic characterization of the attacker. 
With this symbolic semantics at hand, we define split-state symbolic equivalence ($\execEquivalent$) as equivalence of (sets of) traces.
However, as
explained before, we first have to resolve the problem that computational soundness results struggle to deal with situations in which binary operations are applied to  outputs of cryptographic operations. We decided not to exclude programs that exhibit such behaviors, but to perform an over-approximation instead by letting the adversary determine the
outcome of such operations. This makes our abstractions conveniently accessible for existing tools, but it also complicates the computational soundness proof since we
have to consider the operations of constructors and destructors on non-symbolic terms as well. To this end, we encode them as bitstrings and interpret these bitstrings
symbolically again. Fortunately, symbolic bitstring interpretations can be seamlessly combined with all previous CoSP results.

% Once such a separation is in place, we can define the symbolic semantics of two ADL programs $\Pi_1, \Pi_2$. The symbolic execution relation is the honest program semantics. We abstract the results of the cryptographic library with symbolic constructors and destructors from the computationally sound symbolic model in CoSP (\ie, registers can contain symbolic terms). As there maybe malicious components on the system, the adversary is assumed to control a set of \emph{malicious} functions, which, when called, may reveal the arguments to the adversary, and return any term that can be derived from the adversary knowledge at that point.
% Adding this attacker knowledge to the ADL semantic, we define
% a labelled semantics as a light-weight extension to the ADL semantics,
% which captures the interaction between the honest ADL code, and the
% (under specified) adversary.
% With this symbolic semantics at hand, we define symbolic equivalence $\ADLEquivalent$ as labelled bisimilarity, where the labels are the parameters in the calls of malicious functions. It might be of independent interest, that this symbolic semantics of ADL can be solely defined on the split state semantics and is independent of ADL.

We finally define computational  indistinguishability of two honest program semantics in the split-state computational execution of ADL ($\execIndistinguishable$). As usual, the adversary
we consider is a probabilistic polynomial-time machine, and all semantics constitute families of semantics that are indexed by a security parameter.

It might be of independent interest that our symbolic variant of the semantics and the computational indistinguishability can be defined on the split-state form  independently of ADL. We show that ADL can be brought into such a split-state form, then prove later that symbolic equivalence in ADL implies symbolic equivalence in the split-state form, and conclude
by proving that computational indistinguishability in the split-state form implies computational indistinguishability in ADL. Hence, for every two ADL programs $\proto_1, \proto_2$ and initial configurations $s_1, s_2$ ($\protos{i}$ denoting $\proto_i$ with initial configuration $s_i$)
we have
\begin{align*}
    \protos 1 \ADLEquivalent \protos 2 &\implies \protos 1 \execEquivalent \protos 2, 
\text{ and }\\
\protos 1 \execIndistinguishable \protos 2 &\implies \protos 1 \ADLIndistinguishable \protos 2.
\end{align*}

\pparagraph{Computational soundness proof} %
We first construct an injective embedding $\embedding$ that maps every ADL program to a CoSP protocol. We stress that within CoSP, the same CoSP protocol is used for the computational and the symbolic execution and that CoSP requires a separation of the attacker and the cryptographic operations from the rest of the program. Our split-state form precisely satisfies these requirements. 
The embedding $\embedding$ uses the honest-program semantics to iteratively construct a CoSP protocol:  a transition to the crypto-API semantics corresponds to a computation node; a transition to the attacker semantics corresponds to an output node followed by an input node; and whenever several possibilities exist, a control node is selected to
let the adversary decide which possibility (which node) to take.

We prove this embedding sound in the symbolic model,  and we prove it complete with respect to the range of $e$ in the computational model, i.e., for every two ADL programs $\proto_1, \proto_2$ and initial configurations $s_1, s_2$ 
we have
$$\protos 1 \execEquivalent \protos 2 \implies e(\protos 1) \CoSPEquivalent e(\protos 2)$$
and
$$e(\protos 1) \CoSPIndistinguishable e(\protos 2) \implies \protos 1
 \execIndistinguishable \protos 2,$$
%$$e(\sigma_M(\protos 1)) \CoSPIndistinguishable e(\sigma_M(\protos 2)) \implies \sigma_M(\protos 1)
% \execIndistinguishable \sigma_M(\protos 2),$$
 where  $\CoSPEquivalent$ and $\CoSPIndistinguishable$ denote symbolic equivalence and computational indistinguishability in CoSP.
%
% 
%\begin{imitate}{Lemma~\ref{lemma:embedding-soundness-symbolic}}
%Let $\protos 1$ and $\protos 2$ be two programs in the language ADL, and $e$ the injective embedding into CoSP. Then,  
%$$\Pi_1 \execEquivalent \Pi_2 \implies e(\Pi_1) \CoSPEquivalent e(\Pi_2)$$
%\end{imitate}
%% and
%\begin{imitate}{Lemma~\ref{lemma:embedding-soundness-computational}}
%Let $\protos 1$ and $\protos 2$ be two programs in the language ADL, and $e$ the injective embedding into CoSP. Then,  $$e(\sigma_M(\protos 1)) \CoSPIndistinguishable e(\sigma_M(\protos 2)) \implies \sigma_M(\protos 1) \execIndistinguishable \sigma_M(\protos 2)$$
%\end{imitate}

Figure \ref{fig:overview-figure} finally shows how all pieces are put together:
\begin{imitate}{Theorem~\ref{theorem:symbolic-adl-computational-adl} (computational soundness of Dalvik -- simplified)}
Let $\proto_1, \proto_2$ be two ADL programs that use the same crypto-API and $s_1, s_2$ be two initial configurations. Then we have
$$\protos 1 \ADLEquivalent \protos 2 \implies \protos 1 \ADLIndistinguishable \protos 2.$$
\end{imitate}

%
% \end{tabular}
% \smallskip
%
% The embedding directly maps steps in the symbolic ADL semantics to  is constructed in a natural way
% The embedding maps calls to the cryptographic library as computation nodes in CoSP, calls of malicious function to output nodes, and responses from malicious functions as input nodes. All other instructions are encoded into the algorithm of the CoSP protocol that computes the next node. The reduction is constructed as an interactive poly-time machine that acts as an attacker against the computational CoSP execution by transforming the actions of an attacker against the computational execution and that acts as computational execution towards the attacker (using the messages from the computational execution in CoSP).
% As the split-state semantics already separates cryptographic library calls from messages to and from the attacker and from the honest program, the construction of the reduction maps the distinct global transitions for the communication with the cryptographic library and the attacker.
%  With this reduction, we show that every attacker that can break computational indistinguishability against the computational execution can already break computational indistinguishability against two CoSP protocols.
%
% For the symbolic execution, we use the same reduction to prove a dual statement. For all CoSP programs that are in the range of the reduction, we can show that whenever two CoSP programs are symbolic distinguishable the corresponding programs in the original language are also symbolically distinguishable.

\pparagraph{Extension of information flow tools}
To put our work in perspective, we elaborate on a possible application of our result. We envision the extension of information flow (IF) methods with symbolic abstractions. On a high-level, we envision the following approach for extending IF-tools: whenever there is a potential information flow from High to Low and a cryptographic function is called, we extract a model of the app and query a symbolic prover to find our whether the attacker learns something about the High values. As in symbolic ADL only a few semantic rules are changed w.r.t. ADL, modifications to existing analyses are likely to be confined, thus simple to integrate.

This approach imposes the challenge of extracting a model of the app. Our embedding of an ADL program into CoSP already extracts a symbolic model for the program. However, shrinking this extracted model to a manageable size and querying the symbolic prover in a way such that it scales to complex apps is a task far from
simple that merits a paper on its own.

\subsection{Overview}
Section~\ref{section:cosp} reviews the CoSP framework for equivalence properties that we ground our computational soundness result on. 
Section~\ref{section:dalvik-bytecode-semantics} reviews the Abstract Dalvik Language (ADL). 
Sections~\ref{sec:securityframework} and~\ref{sec:dalvik-symbolic-semantics} define the probabilistic execution of ADL and introduce our symbolic variant of ADL,  including the symbolic abstractions of cryptographic
operations and the capabilities of the symbolic adversary. Section~\ref{sec:soundness}  defines the connections between ADL, symbolic ADL, and CoSP, and based on these connections proves the computational soundness result. 
%Section~\ref{section:applicability} introduces a computationally sound declassification policy that is canonically defined from the symbolic model.\textbf{MB: Check if we keep the last sentence.} 
Section \ref{section:related-work} discusses related work. 
We 
conclude in Section \ref{sec:conclusion} with a summary of our findings and outline directions for future research.

% Section~\ref{section:split-state-form} defines the split-state form for an operational semantics and shows that the operational semantics for Dalvik Bytecode, proposed by Lortz et al.~\cite{LoMaStBaScWe_14:cassandra}, can be brought into the split-state form. Section~\ref{section:interactive-game} shows how to construct an interactive cryptographic game out of a split-state operational semantics and how to define computational indistinguishability. Section~\ref{section:symbolic-operational-semantics} defines a symbolic version of the split-state semantics, from a computationally sound symbolic model in CoSP.
% Section~\ref{section:cosp-embedding} shows how to embed both the split-state semantics and the symbolic split-state semantics into CoSP, using the representation as an interactive game from Section~\ref{section:interactive-game}. Moreover Section~\ref{section:cosp-embedding} proves that equivalence properties are preserved by the embedding and that in turn a computationally symbolic model in CoSP induces a computational soundness for the symbolic version of the split-state semantics from Section~\ref{section:symbolic-operational-semantics}, which in particular holds for DEX Bytecode. Section~\ref{section:applicability} illustrates how a symbolic version of the operational semantics could be used for both static and dynamic verification methods.
 %4 pages

%!TEX root = main.tex
\section{Notation}\label{sec:notation}%
% sets and functions
Let $\setN$ be the set of natural numbers and assume that they begin at $0$. 
\iffullversion%
We abbreviate the statement that a set $A$ is a finite subset of a set
$B$ as $A \subseteq_{fin} B$.
\fi%
For indicating that  function $f$ from a set $A$ to a set $B$ is
a partial function, we write $f : A \rightharpoonup B$.
%
% things with [brackets]
We use squared brackets in two different ways: $(i)$ $m[pp]$ denotes the instruction with the number $pp$ for a set of instructions $m$, 
and $(ii)$ $r[v \mapsto \mathit{val}] := (r\setminus (v,r(v))) \cup
(v,val)$ is short-hand for the function mapping $v$ to $\mathit{val}$
and otherwise behaving like $r$. 
\iffullversion%
We lift this notation to functions,
\ie, $r[f](v)$ is equal to $f(v)$ if $f$ is defined on $v$, and equal to $r(v)$
otherwise.
\fi%
Throughout the paper, we use $\secpar$ as the security parameter. 
% 
% Sequence notation
We use $\varepsilon$ to denote the 
empty sequence, 
empty path, 
empty bitstring,
or empty action
depending on the context.
We write $\underline t$ for a sequence $t_1,\dotsc,t_n$ if $n$ is clear from the context.
For any sequence $l\in E^*$, we use 
$\cdot$ to denote concatenation $l_1 
\cdot l_2$, as well as the result of 
appending ($l \cdot e$) or 
prepending ($e \cdot l$) an element, as long as the difference is
clear from context.
\iffullversion%
We use ${(f{(l)})}_{e\in l}$ to the sequence resulting from applying the
(meta-language) operation $f$
to each element in a sequence $l$, \eg, ${(e^2)}_{e\in
\set{1,2,3}}$ is the sequence of the first three square numbers.
We use $l|_0^k$ to denote the $k$-prefix of $l$.
\fi%
We filter a sequence $l$ by a set $S$, denoted $l|_S$, by removing each element
that is not in $S$. 
\iffullversion%
We use similar notation for the projection of a sequence $l=s_1,\ldots,s_n \in S^*$:
given a partial function $\pi \colon S \rightharpoonup T$, 
    $l|_\pi = (s_1,\ldots,s_{n-1})|_\pi\cdot \pi(s_n)$
    or 
    $(s_1,\ldots,s_{n-1})|_\pi$, if $\pi(s_n)$ undefined.
\fi%
\iffullversion%
    $(e)^l$ denotes the sequence of length $l$ where each element equals $e$.
\fi%
As we represent the attacker as a transition system, we sometimes write $T_A$ and sometimes $\attacker$ for the attacker.

% vim: ft=tex

%!TEX root = main.tex
\section{CoSP Framework (Review)}\label{section:cosp}
\iffullversion
The computational soundness proof developed in this paper follows CoSP~\cite{BaHoUn_09:cosp,BaMoRu_14:equivalence}, a general framework for 
conducting computational soundness proofs of symbolic cryptography and for embedding these proofs into programming languages with their
given semantics. CoSP enables proving computational soundness results in a conceptually modular and generic way: every computational soundness
result for a cryptographic abstraction phrased in CoSP automatically holds for all suitably embedded languages, and the process of embedding is conceptually
decoupled from computational soundness proofs. Hence in this work, we will   suitably embed Dalvik Bytecode into CoSP, and thereby leverage existing computational soundness results 
of CoSP.

In Section~\ref{section:symbolic-model}, we review the \emph{symbolic model} of CoSP, which encompasses symbolic abstractions of cryptographic operations and the representation of programs (which are called protocols) in CoSP. In Section~\ref{section:symbolic-equivalence}, we present the notion of \emph{symbolic equivalence} in CoSP, which defines that two CoSP protocols are indistinguishable for an attacker that operates on symbolic abstractions. In Section~\ref{section:computational-indistinguishability}, we review the notion of \emph{computational indistinguishability} in CoSP, which defines
the execution of protocols using actual cryptographic algorithms. In Section~\ref{section:computational-soundness}, we finally review CoSP's notion of \emph{computational soundness}.
\else
The computational soundness proof developed in this paper follows CoSP~\cite{BaHoUn_09:cosp,BaMoRu_14:equivalence}, a general framework for 
conducting computational soundness proofs of symbolic cryptography and for embedding these proofs into programming languages with their
given semantics. This section reviews the CoSP framework.
\fi

\pparagraphorsubsection{Symbolic Model \& Execution}\label{section:symbolic-model}~
\iffullversion
CoSP provides a general symbolic model for expressing cryptographic abstractions. 
We start with some basic concepts such as constructors, destructors, nonces and
message types.

\begin{definition}[CoSP Terms]
  A \emph{constructor} $C$ is a symbol with a (possibly zero)
  arity.
  A \emph{nonce} $N$ is a symbol with zero arity. 
  We write $C/n\in \mathbf C$ to denote that $\mathbf C$ contains a constructor
  $C$ with arity $n$.  A \emph{message
      type $\mathbf T$ over $\mathbf C$ and $\mathbf
    N$} is a set of terms over constructors $\mathbf C$ and nonces $\mathbf N$.
     A \emph{destructor} $D$ of arity $n$, written $D/n$, over a
    message type $\mathbf
      T$ is a partial map $\mathbf T^n\to \mathbf T$. If $D$ is
  undefined on $\underline t$, we write $D(\underline t)=\bot$.
\end{definition}

In CoSP, symbolic abstractions of protocols and of the
attacker are formulated in a symbolic model, including
countably infinite set of nonces partitioned into protocol and
attacker nonces.

  \begin{definition}[Symbolic model]\label{def:symb.model}
  A \emph{symbolic model}
  $\mathbf M=(\mathbf C,\mathbf N,\mathbf T,\mathbf D)$
  consists of a set of constructors $\mathbf C$, a set of nonces $\mathbf N$, a message type
  $\mathbf T$ over $\mathbf C$ and $\mathbf N$ (with $\mathbf
  N\subseteq\mathbf T$), a
  set of destructors $\mathbf D$ over $\mathbf T$. We require that $\N = \NE \uplus \NP$ for countably infinite sets
  $\NP$ of protocol nonces and attacker nonces $\NE$.
\end{definition}

To unify notation for constructors, destructors, and nonces, we define the \emph{evaluation of terms} as the partial function $\eval_F:\mathbf T^n\to\mathbf T$ for every constructor or destructor $F/n \in \D \cup \C$ and every nonce $F\in\N$ as follows: (where $n = 0$ for a nonce).

\begin{definition}[Evaluation of terms]\label{def:eval-prime}
If $F$ is a
  constructor, define $\eval_F(\underline t) \defeq  F(\underline t)$ if
  $F(\underline t)\in\mathbf T$ and $\eval_F(\underline t) \defeq  \bot$
  otherwise.  If $F$ is a nonce, define $\eval_F() \defeq  F$. If $F$ is a destructor,
  define $\eval_F(\underline t) \defeq  F(\underline t)$ if $F(\underline t)\neq\bot$
  and $\eval_F(\underline t) \defeq  \bot$ otherwise.
\end{definition}

\else%
 In CoSP, symbolic abstractions of protocols and the
attacker are formulated in a \emph{symbolic model} $\M= (\C, \nonce, \T, \D)$: a set of free functions
$\C$, a countably infinite set $\N$ of nonces, a set $\T$ of
terms, and a set $\D$ of partial mappings from terms to terms (called destructors).
To unify notation, we introduce $\eval_F(t)$: if $F$ is a
  constructor, $\eval_F(t) \defeq  F(t)$ if
  $F(t)\in\mathbf T$ and $\eval_F(t) \defeq  \bot$
  otherwise.  If $F$ is a nonce, $\eval_F() \defeq  F$. If $F$ is a destructor,
  $\eval_F(t) \defeq  F(t)$ if $F(t)\neq\bot$
  and $\eval_F(t) \defeq  \bot$ otherwise.
\fi

\pparagraph{Protocols} In CoSP, protocols are represented as infinite trees with the following nodes: \emph{computation nodes} are used for drawing fresh nonces
and for applying constructors and destructors; \emph{input} and
\emph{output nodes} are used for send and receive operations;
\emph{control nodes} are used for allowing the attacker to
schedule the protocol. A computation node is annotated
with its arguments and has two outgoing edges: a yes-edge, used for the
application of constructors, for drawing a nonce, and for the successful
application of a destructor, and a no-edge, used if an application of a constructor or
destructor $F$ on a term $t$ fails, i.e., if $\eval_F(t)=\bot$.
Nodes have explicit \emph{references} to other nodes whose terms they use.
\iffullversion%
For example, a computation node that computes
$C(t)$ references the node that produced $t$, e.g., an input node or another computation node.\fi

\iffullversion \begin{definition}[CoSP protocol]
\label{def:cosp-protocol}
  A \emph{CoSP protocol} ${\proto}$ for a symbolic model $\mathbf M=(\mathbf C,\mathbf N,\mathbf T,\mathbf D)$ is a tree of infinite depth with a distinguished root and labels
  on both edges and nodes. Each node has a unique identifier $\n$ and
  one of the following types:\footnote{In contrast to the definition in the original CoSP framework~\cite{BaHoUn_09:cosp}, we do not consider
   the type ``non-deterministic node''. This type has been obsolete for all embeddings of symbolic calculi that have been established in the original framework so far 
   as well.}
\begin{itemize}
  \item
    \emph{Computation nodes} are annotated
    with a constructor, nonce or destructor $F/n \in \mathbf C \cup \NP \cup \mathbf D$
  together with the identifiers of $n$ (not necessarily distinct)
  nodes; we call these annotations \emph{references}, and we call the referenced nodes \emph{arguments}. Each computation node $\n$ has exactly two
  successors; $\mathit{yes}(\nu)$ and $\mathit{no}(\nu)$,
  the corresponding edges are labeled with
  $\mathtt{yes}$ and $\mathtt{no}$, respectively.

\item \emph{Input nodes} have no annotations. They have
  exactly one
  successor, $\succs(\nu)$.

 \item \emph{Output nodes} have a reference to exactly one node in their annotations.
     They have exactly one successor, $\succs(\nu)$.
\item \emph{Control nodes} are annotated with a bitstring $l$. They have at least one and up to countably many successors; the corresponding edges are labeled
  with distinct bitstrings $l'$. We call $l$
    the \emph{out-metadata} and $l'$ the \emph{in-metadata} of this node.
  \end{itemize}
  We assume that the annotations are part of the node identifier.
  A node $\n$ can only reference other nodes $\n'$
  on the path from the root to $\n$;
  in this case $\n'$ must be a computation
  node or input node. If
  $\n'$ is a computation node,  the path from
  $\n'$ to $\n$ has additionally to go through the outgoing edge of $\n'$ with label $\mathtt{yes}$.
\end{definition}
 \fi

% \pparagraph{Bi-protocols}  For the definition of symbolic equivalence, we represent pairs of protocols as so-called \emph{bi-protocols}, which are pairs of protocols that
% are encoded in one extended CoSP protocol tree and that only differ in the messages they operate on.
% \TODOR{Do we ever Bi-protocols in this work?}

\pparagraph{Symbolic operations} \iffullversion
As a next step, we model the capabilities of the symbolic attacker. 
We have to capture which protocol messages the attacker observes, in particular in
which order an attacker observe these messages. Moreover, we have to capture which tests an
attacker can perform in order to judge whether two protocols are distinguishable. These tests,
called \emph{symbolic operations}, capture the sequence of operations a symbolic attacker applies,
including the used protocol messages. This notion is sometimes
referred to as a \emph{recipe}.

\begin{definition}[Symbolic Operation]
  \label{def:symbolic-operation}
Let $\Model=(\C,\N,\Terms,\D)$ be a symbolic model.
A \emph{symbolic operation $O/n$ (of arity $n$) on $\M$} is a finite tree whose nodes are
labeled with constructors from $\C$, destructors from $\D$, nonces from $\N$, and formal parameters $x_i$ with
$i\in\set{1,\dots,n}$.
For constructors and destructors, the children of a node represent its arguments (if any). Formal
parameters $x_i$ and nonces do not have children. 
    The set of symbolic operations of arity $n$ for a model $\Model$
    is denoted as $\setSymbolicOperations(\M)_n$. The set
    $\setSymbolicOperations(\M) := \bigcup_{n\in\setN}\setSymbolicOperations(\M)_n$ is the set of
    all symbolic operations for $\Model$.
  We extend the evaluation function to symbolic operations.
  Given a list of terms $\underbar t\in\T^n$, the evaluation function $\eval_O: \T^n \to \T$
  recursively evaluates the tree $O$
  starting at the root as follows: The formal parameter $x_i$ evaluates
  to $t_i$. A node with $F\in \C \cup \N_E \cup \D$ evaluates according to $\eval_F$, where $\N_E\subseteq\N$ are attacker nonces.
  If there is a node that evaluates to $\bot$, the whole tree
  evaluates to $\bot$.
\end{definition}
Note that the identity function is included. It is the tree that contains only $x_1$ as node.
\else
To model the capabilities of the symbolic attacker, we explicitly list the tests and operations that the attacker can perform on protocol messages using so-called \emph{symbolic operations}.
A symbolic operation is (similar to a CoSP tree) a finite tree whose nodes are labeled with constructors, destructors, or nonces from the symbolic model $\Model$, or formal parameters $x_i$ denoting pointers to the $i$th protocol message. There is a natural evaluation function for a symbolic operation $O$ and a list $\underline t$ of terms that the attacker received so far (the \emph{view}).
\fi

\pparagraph{Symbolic execution}
A symbolic execution is a path through a protocol tree defined as defined below.
It induces a \emph{symbolic view}, which contains the communication
with the attacker. We, moreover, define an \emph{attacker strategy} as
the sequence of symbolic operations that the attacker performs in the
symbolic execution.
\iffullversion
It induces a \emph{view}, which contains the communication with the attacker. Together with the
symbolic execution, we define an \emph{attacker strategy} as the sequence of symbolic operations that the
attacker performs.
\fi

\begin{framedFigure}{Rules for defining the smallest relation for the symbolic execution}{figure:symbolic-execution}
\footnotesize
\vspace{-0.5em}
\begin{align*}
    (V,\nu,f) & \cospsym (V,\mathit{yes}(\nu),f[\nu\mapsto m])
    &&
    \nu \text{ computation n. with $F\in\C\cup\D\cup\N$,
    $m\defeq \eval_F (\underline{\tilde t})\neq\bot$}
    \\
    (V,\nu,f) & \cospsym (V,\mathit{no}(\nu),f)
    && 
    \text{$\nu$ computation node with $F\in\C\cup\D\cup\N$,
    $\eval_F (\underline{\tilde t})=\bot$}
    \\
    (V,\nu,f) & \cospsym (V\cdot(\texttt{in},
    (t,O)),\mathit{succ}(\nu),f[\nu\mapsto t])
    &&
    \text{$\nu$ input node,
    $t\in\T$,
    $O\in\SO(\M)$,
$\eval_O(\termsof(V))=t$}
    \\
    (V,\nu,f) & \cospsym (V\cdot(\texttt{out}, \tilde t_1),\mathit{succ}(\nu),f)
    &&
    {\nu~\text{output node}}
    \\
    (V,\nu,f) & \cospsym ( V\cdot(\texttt{control},(l,l')),\nu',f) 
    && 
    \nu~\text{control n., out-metadata $l$,
    successor  $\nu'$ has  in-metadata $l'$}
\end{align*}
\vspace{-0.5em}
\end{framedFigure}

\begin{definition}[Symbolic executions]\label{def:symbolic-execution}
  Let a symbolic model $\Model=(\C,\N,\Terms,\D)$ and a CoSP protocol
  ${\proto}$ for $\Model$ be given. Let $\mathit{Views} = 
   (\Eventin \cup \Eventout \cup \Eventctl)^*$, with 
$\Eventin := \set{\texttt{in}} \times \T \times \setSymbolicOperations(\M),
        \Eventout := \set{\texttt{out}} \times \T^*,
       \Eventctl := \set{\texttt{control}} \times \bits* \times \bits*$.
  We define 
\begin{align*}
\cospsym \subseteq &
      \mathit{Views}^*\times \mathit{Nodes} \times (\mathit{Nodes} \to \Terms)
      \to\\
      &\mathit{Views}^*\times \mathit{Nodes} \times (\mathit{Nodes} \to \Terms)
\end{align*}
  as the smallest relation s.t.\ the rules from Figure~\ref{figure:symbolic-execution} hold,
where 
     $\tilde{\underline t}
     % \defeq \tilde t_{1}, \dots, \tilde t_{|\tilde \nu|}
     \defeq f(\tilde \nu_{1}), \dots, f(\tilde
     \nu_{|\tilde{\underline{\nu}}|})$,
     for the nodes
     $\tilde \nu_1,\ldots,\tilde \nu_{|\tilde{\underline{\nu}}|}$
     referenced by $\nu$ (for computation nodes).
%
  % Let $\cospsym$ be the smallest relation such that the rules in
  % Figure~\ref{fig:cosp-sym} hold.
  The set of \emph{symbolic executions of
  \proto} is defined as
\begin{align*}
      \mathrm{SExec}(\proto) := & \{ 
      ((V_0,\nu_0,f_0), \dots, (V_n,\nu_n,f_n)) \mid n\in\setN\\
	  & \land \forall i. (V_i,\nu_i,f_i) \cospsym (V_{i+1},\nu_{i+1},f_{i+1})\},
\end{align*}
  where $(V_0,\nu_0,f_0) = (\epsilon,\mathit{root}(\proto),\emptyset)$.
  $V_i$ is called a \emph{symbolic view} for step $i$.
  The set of symbolic views of $\proto$ is defined as
  \[ \Views(\proto) := 
      \set{V_n | (V_0,\nu_0,f_0), \cdots, (V_n,\nu_n,f_n) \in \mathrm{SExec}(\proto)}
  \]
  Given a view $V$,
  $\termsof(V)$ denotes the list of terms $t$ contained in $(\stringOutput,t)\in V$.
  $\outmetaof(V)$ denotes the list of terms $l$ contained in elements of the form $(\stringControl, (l,l'))$ in the view $V$.
    $\inputof(V)$, called the \emph{attacker strategy}, denotes the list of terms that contains only entries of $V$ of the form $(\stringInput,(t,O))$ or $(\stringControl,(l,l'))$, where for $(\stringInput,(t,O))$ the entry $(\stringInput,O)$ is stored and for $(\stringControl,(l,l'))$ the entry $(\stringControl,l')$ is stored% , where the input term and the out-metadata first term has been masked with the symbol $\ast$
. $[\inputof(V)]_{\Views(\proto)}$ denotes the equivalence class of all views $U\in\Views(\proto)$ with $\inputof(U) = \inputof(V)$.
\end{definition}

\pparagraph{Symbolic knowledge}
The \emph{symbolic knowledge} of the attacker comprises the results of all symbolic tests the
attacker can perform on the messages output by the protocol.
\iffullversion%
The definition captures that the attacker knows exactly which symbolic operation
leads to which result.

\begin{definition}[Symbolic knowledge]
 \label{def:symbolic-knowledge}
 Let $\M$ be a symbolic model.
 Given a view $V$ with $|\termsof(V)| = n$, a
 \emph{symbolic knowledge function} $f_V:\setSymbolicOperations(\Model)_n \rightarrow
 \set{\top,\bot}$ is a partial function from symbolic operations (see
 Definition~\ref{def:symbolic-operation}) of arity $n$ to $\set{\top,\bot}$. The \emph{full symbolic
 knowledge function} is a total symbolic knowledge function $\testsof_V(O)$  defined by
 \[
 \testsof_V(O) \defeq
 \begin{cases}
 \bot \quad \text{if } \eval_O(\termsof(V)) = \bot \\
 \top \quad \text {otherwise.}
 \end{cases}
 \]
\end{definition}
\else%
Given a view $V$ with $|V_\termsof|=n$, we define the symbolic
knowledge $\testsof_V$ as a function from symbolic operations on $\M$
to $\set{\top,\bot}$, where $\top$ comprises all results of
$\eval_O(V_\termsof)$ that are different from $\bot$.
\fi%

\pparagraphorsubsection{Symbolic Equivalence} \label{section:symbolic-equivalence}
Two views are \emph{equivalent}, denoted as $V \sim V'$, if they $(i)$
have the same structure (i.e., the same order of $\stringOutput,
\stringInput, \stringControl$ entries), $(ii)$ have the same
out-metadata (i.e., $V_\outmetaof = V'_\outmetaof$), and $(iii)$ lead
to the same knowledge (i.e., $\testsof_V = \testsof_{V'}$).
Finally, two CoSP protocols %(e.g., the left and right variant of a bi-protocol) 
are \emph{symbolically equivalent} ($\CoSPEquivalent$) if its two variants lead to equivalent
views when run with the same attacker.

\iffullversion
\begin{definition}[Equivalent views]
  \label{def:static-equivalence}
%  Let $\Pi$ be a CoSP protocol, and
 Let two views $V, V'$ of the same length be given. We denote their $i$th entry by $V_i$ and $V'_i$, respectively. $V$ and $V'$ are \emph{equivalent} ($V\sim V'$), if the following three conditions hold:
  \begin{enumerate}
    \item\label{item:same-structure} (Same structure) $V_i$ is of the form $(s,\cdot)$ if and only if $V'_i$ is of the form $(s,\cdot)$ for some $s\in\{\stringOutput,\stringInput,\stringControl\}$.
    \item\label{item:same-out-metadata} (Same out-metadata) $\outmetaof(V) = \outmetaof(V')$.
    \item\label{item:same-symbolic-knowledge} (Same symbolic knowledge) $\testsof_V = \testsof_{V'}$.
  \end{enumerate}
\end{definition}

Finally, we define two protocols to be \emph{symbolically equivalent} if the two protocols lead to equivalent
views when faced with the same attacker strategy.
Thus, a definition of symbolic equivalence should compare the symbolic knowledge of two
protocol runs only if the attacker behaves identically in both runs.

Given the static equivalence relation from Definition~\ref{def:static-equivalence}, 
the notion of symbolic equivalence in CoSP is defined as trace equivalence.

\begin{definition}[Symbolic equivalence]
\label{def:symbolic-indistinguishable}
Let $\Model$ be a symbolic model and $\protoClass$ be a class of
protocols for $\Model$.
Let $\traces_A(\Pi)=\set{[A]_{\Views(\Pi)}}$
denote the list of traces for an attacker strategy $A$.
We lift the relation $\sim$ from Definition~\ref{def:static-equivalence} to sets, \ie, $A \sim B$ iff. 
$\forall a\in A \exists b\in B. a \sim b$ and vice versa.

Two protocols $\Pi_1, \Pi_2 \in \protoClass$ are \emph{trace
equivalent} ($\Pi_1 \CoSPEquivalent \Pi_2$), if, for all attacker
strategies $A$,
\[ \traces_A(\Pi_1)\sim\traces_A(\Pi_2). \]
\end{definition}

%!TEX root = main.tex

\begin{framedFigure}{Labelled transition}{def:symbolic-labelled-trans}
\footnotesize
\vspace{-0.5em}
\begin{align*}
    (V,v,f) 
    & \rightarrow (V,\mathit{yes}(v),f[\nu\mapsto m])
    && \text{$\nu$ computation node with $F\in\C\cup\D\cup\N$},
    m\defeq \eval_F (\underline{\tilde t})\neq\bot 
    \\
    (V,v,f) 
    & \rightarrow (V,\mathit{no}(v),f)
    && 
    \nu \text{ computation node with }F\in\C\cup\D\cup\N,
    \eval_F (\underline{\tilde t})=\bot
    \\
    (V,v,f) 
    & \xrightarrow{\mathit{in}(O)} 
    (V\dot(\texttt{in}, (t,O)),\mathit{succ}(v),f[\nu\mapsto t])
    &&
    \nu~\text{input node},
    t\in\T,
    O\in\SO(\M),
    \eval_O(\termsof(V))=t
    \\
    (V,v,f) 
    & \xrightarrow{\textit{out}}
    (V\cdot(\texttt{out}, \tilde t_1),\mathit{succ}(v),f)
    &&
    {\nu~\text{output node}}
    \\
    (V,v,f) 
    & \xrightarrow{\textit{control}(l)}
     V\cdot(\texttt{control},(l,l')),\nu',f) 
    && 
    \nu~\text{control n., out-metadata $l$,
    successor  $\nu'$ has  in-metadata $l'$}
\end{align*}
\vspace{-0.5em}
\end{framedFigure}

\begin{definition}[(Weak) bisimilarity]
    \label{def:bisimiliarity}
    A \emph{(weak) bisimulation relation} is a symmetric relation
    \[\calR: 
      \mathit{Views}^*\times \mathit{Nodes} \times (\mathit{Nodes} \to \Terms)
      \times
      \mathit{Views}^*\times \mathit{Nodes} \times (\mathit{Nodes} \to \Terms)
  \]
    such that, if
    $A \calR B$, then
    \begin{itemize}
        \item $K_{V_A} = K_{V_B}$ (where $V_A$,$V_B$ are the symbolic views in $A$ and $B$).
        \item if $A \rightarrow A'$ then there is $B'$ s.t.  $B \rightarrow^* B'$ and $A
            \calR B$
        \item if $A\xrightarrow{\alpha} A'$ then there is $B'$ such
            that $B \rightarrow^* \xrightarrow{\alpha} \rightarrow^*
            B'$ and $A \calR B$.
    \end{itemize}

    If there exist a weak bisimulation relation between $A$ and $B$,
    they are \emph{bisimilar}.
\end{definition}

Without non-deterministic nodes, more generally without internal non-determinism, the notion of trace equivalence (i.e., symbolic equivalence from Definition~\ref{def:symbolic-indistinguishable}) is equivalent to weak bisimilarity. 
% In general, trace equivalence does is not equivalent to weak bisimilarity, but in our model there is no internal
% non-determinism, which is generally the reason bisimulation exists.

\begin{lemma}
    For any model \M, any two processes $\Pi_1$ and $\Pi_2$ are trace
    equivalent if they are bisimilar.
\end{lemma}
\begin{proof}
We introduce the following function mapping a single stap of a trace,
\ie,
$e\in\mathit{Views}$
to  labels according to
Definition~\ref{def:symbolic-labelled-trans}. Let
\[ \beta (v)=
    \begin{cases}
        \texttt{in}(O) & \text{if }v=(\texttt{in},(t,O))\in\Eventin \\
        \texttt{out} & \text{if }v=(\texttt{out},\tilde t)\in\Eventout \\
        \texttt{control}(l) & \text{if }v=(\texttt{control}(l,l'))\in\Eventctl. \\
    \end{cases}
\]
% \textbf{Bisimilarity implies trace equivalence:}
Without loss of generality, we assume that each CoSP tree is transformed into a CoSP tree where each node uniquely describes its subtree completely.
\TODOR{I identify the protocol with the root of its tree, i.e.,
    each node describes its subtree completely. Make sure to mention
it around Definition~\ref{def:cosp-protocol}.}
Fix an arbitrary trace $t_1=(v_1,\ldots,v_n)\in\traces_A(\Pi_1)$.
By Definition~\ref{def:symbolic-execution}
and Figure~\ref{def:symbolic-labelled-trans}, there is a sequence of
transitions and triples $(V_i,\nu_i,f_i)$ such that:
\[ 
    (V_0,\nu_0,f_0) 
    \rightarrow^* \xrightarrow{\beta(v_1)} \rightarrow^* 
    \cdots
    \rightarrow^* \xrightarrow{\beta(v_n)} \rightarrow^* 
    (V_n,\nu_n,f_n),
\]
and $v_0=\Pi_1$.
As $\Pi_1$ and $\Pi_2$ are bisimilar, there exists a sequence of
transitions and triples
$(V_i',\nu_i',f_i')$ such that $v_0'=\Pi_2$, and
\[
    (V_0',\nu_0',f_0') 
    \rightarrow^* \xrightarrow{\beta(v_1)} \rightarrow^* 
    \cdots
    \rightarrow^* \xrightarrow{\beta(v_n)} \rightarrow^* 
    (V_n',\nu_n',f_n'),
\]
as well as $K_{V_n}=K_{V_{n'}}$. 
The latter point implies that the third condition of
Definition~\ref{def:static-equivalence} holds for
$(V_0,\ldots,V_n)$
and 
$(V_0',\ldots,V_n')$.
From the relation between the view and the labels in each transition
described in 
Definiton~\ref{def:symbolic-labelled-trans}, it follows that the first
and second condition holds. Hence for each
$(V_0,\ldots,V_n) \in \Views(\Pi_1)$
there exists an equivalent view
$(V_0',\ldots,V_n') \in \Views(\Pi_2)$.
The same argument can be made for $\Pi_2$ in place of $\Pi_1$ and vice
versa, thus $\Pi_1$ and $\Pi_2$ are trace equivalent.

\end{proof}

\fi

\pparagraphorsubsection{Computational execution}~\label{section:computational-indistinguishability}
 In the computational setting, symbolic constructors and destructors are realized with
 cryptographic algorithms.
A \emph{computational
implementation} is a family $\underline \implementation=(A_x)_{x \in \C \cup \D \cup \NP}$ of deterministic polynomial-time
algorithms $\implementation_F$ for each constructor or destructor $F \in \C\cup\D$ as well as a probabilistic
  polynomial-time (ppt) algorithm $A_N$ for drawing protocol nonces $N\in\N$. 
  \iffullversion%
The first argument of $\implementation_F$ and $A_N$ represents the security parameter. \fi

\iffullversion 

\begin{framedFigure}{Rules defining the smallest relation for the computational execution, where
     $\mathit{yes}(\nu)$ and $\mathit{no}(\nu)$ are the yes- and
     no-successor of $\nu$ (for computation nodes), 
     $\mathit{succ}$ is the successor node (for input or output nodes),
     and 
     $\tilde{\underline t}\defeq \tilde t_{j}\defeq f(\tilde
     \nu_{j})$ with $\underline{\tilde \nu}$ are the nodes referenced by $\nu$ (for computation nodes).
    }{fig:computational-execution}
\footnotesize
\vspace{-0.5em}
\begin{align*}
    (\nu,f,n) & \cospcomp (\mathit{yes}(\nu),f[\nu\mapsto n(N)],n)
    && \text{$\nu$ computation node with $N\in\N$ and $n(N)\neq\bot$}
    \\
    (\nu,f,n) & \cospcomp (\mathit{yes}(\nu),f[\nu\mapsto
m],n[N\mapsto m)
    && \text{$\nu$ computation node with $N\in\N$ and
$n(N)=\bot$}\\
    &&& \text{and  $m$ drawn according to $\implementation_N(\secpar)$.}
    \\
    (\nu,f, n) & \cospcomp (\mathit{yes}(\nu),f[\nu\mapsto m],n)
    && 
    \nu \text{~computation node with $F\in\C\cup\D\cup\N$, and}
    A_F (\secpar,\underline{\tilde t})=m\neq\bot
    \\
    (\nu,f, n) & \cospcomp (\mathit{no}(\nu),f,n)
    && 
    \nu \text{~computation node with $F\in\C\cup\D\cup\N$, and}
    A_F (\secpar,\underline{\tilde t})=\bot
    \\
    (\nu,f, n) & \cospcomp (\succs(\nu),f[\nu\mapsto m],n)
    && \text{$\nu$ input node, a request bitstring and receive $m$}
    \\
    (\nu,f, n) & \cospcomp (\succs(\nu),f,n)
    && \text{$\nu$ output note, send ${\tilde {m}_1}$ to the communication partner}
    \\
    (\nu,f, n) & \cospcomp (\nu',f,n)
    && \text{\parbox{8.3cm}{
    $\nu$ control node with out-metadata $l$,
     send $l$ to communication partner and receive in-metadata $l'$. If there is an edge with label $l'$, let $\nu'$ be the successor of $\nu$ along the edge
labeled $l'$; otherwise let $\nu'$ be the successor with the lexicographically smallest edge label.}}
\end{align*}
\vspace{-0.5em}
\end{framedFigure}

\begin{definition}[Computational implementation]\label{def:comp.impl}
  Let $\mathbf M$ $=(\mathbf C,\mathbf N,\mathbf T,\mathbf D)$ be a symbolic model. A
  \emph{computational implementation of $\mathbf M$} is a family of
  functions $\implementation = (A_x)_{x \in \mathbf C \cup \mathbf
    D\cup\mathbf N}$ such that
  $A_{F}$ for $F/n\in\mathbf C\cup\mathbf D$ is a partial
  deterministic function $\setN \times (\set{0,1}^*)^n\to\set{0,1}^*$,   and
  $A_N$ for $N \in \Nonce$ is a total probabilistic function with domain $\setN$ and range $\{0,1\}^*$.
  All functions $A_F$  have to be computable in
  deterministic polynomial time, and all $A_N$ have to be computable in
  probabilistic polynomial time (ppt).

  We extend the computational implementation to symbolic operations.
  The computational implementation of a symbolic operation $O\in\SO$ is defined  $A_O: (\bits*)^n \to \bits*$,
  recursively evaluating the tree $O$
  starting at the root as follows: The formal parameter $x_i$ computed
  to $i$th input value. A node with $F\in \C \cup \N_E \cup \D$ is computed according to $A_F$. If the function is undefined for any node, then the $A_O$ is undefined on these inputs.
  \end{definition}
\fi

\pparagraph{Computational execution} The \emph{computational execution} of a protocol is a randomized
interactive machine, called the \emph{computational challenger}, that traverses the protocol tree and interacts with 
a ppt attacker $\advE$: at
a computation node the corresponding algorithm is run and depending on
whether the algorithm succeeds or outputs $\bot$, 
 either the yes-branch or the
no-branch is taken; at an output node, the message is sent to the
attacker, and at an input node a message is received by the attacker;
at a control node
the attacker sends a command that specifies which branch to take.
The transcript of the execution contains the computational counterparts of a
symbolic view.\fullversionfootnote{%
    We stress that we not have to incorporate an explicit run-time polynomial to ensure overall
termination in polynomial time, since computational soundness in CoSP is based on tic-indistinguishability, which only requires indistinguishability for 
polynomially-long prefixes, see Section~\ref{section:computational-soundness}.}

\iffullversion

\begin{definition}[Computational execution] \label{def:computional-execution-challenger} \label{def:computional-execution}
  Let $\implementation$ be a computational implementation of the symbolic model $\Model=(\Const, \N, \Terms,\Dest)$ and $\proto$ be a
  CoSP protocol. For a security parameter $\secpar$, the \emph{computational challenger}
$\Exec_{\Model,\implementation,\proto}(\secpar)$
is an interactive Turing machine defined by repeatedly applying the
rules in Figure~\ref{fig:computational-execution}, starting from the
initial state $(\nu,f,n)$, where $\nu$ is the root of $\proto$, and
$f$ and $n$ are undefined partial functions from node identifiers to
bitstrings and from $\mathbf N$ to bitstrings, respectively.

We call the interaction between $\Exec_{\Model,\implementation,\proto}(\secpar)$ and an interactive ppt Turing machine $\attacker(\secpar)$ the \emph{computational execution}, and denote it as $\interact{\Exec_{\Model,\implementation,\proto}(\secpar)}{\attacker(\secpar)}$ using 
standard notation for interacting machines. The execution stops
whenever one of the two machines stops, and the output of $\interact{\Exec_{\Model,\implementation,\proto}(\secpar)}{\attacker(\secpar)}$ is the output of $\attacker(\secpar)$.
\end{definition}

By Definition~\ref{def:adl-semantical-domains}, there is an
efficiently computable injection from $\valuesDV\cup\heapsDV$ into
$\bits*$, hence we abuse notation by supplying values from domain in
lieu of bitstrings, and using the inverse of this injection to convert
bitstrings back into values from $\valuesDV\cup\heapsDV$.

% \pparagraph{Efficient CoSP protocols}
The CoSP execution is used in cryptographic reductions; however, CoSP protocols are by
definition infinite. Hence, as is standard in cryptography,
we have to ensure the execution of a CoSP protocols is computable in polynomial-time. To this
end, we require that it is possible to incrementally compute the node information, i.e., the node
identifier, for each path through the CoSP protocol in polynomially many steps (in the length
of the path from the root node to the current node). We call such (bi-)protocols efficient.

% is efficient in the following sense: $(i)$ the length of the
% identifier of a node in the (bi-)protocol is polynomially bounded (in the length of the path from
% the node to the root), $(ii)$ for each node the computation of the identifier of a successor node is
% polynomially bounded, and $(iii)$ 

\begin{definition}[Efficient Protocol]\label{def:cosp:efficient-protocol}
  We call a CoSP protocol
  \emph{efficient} if:
  \begin{itemize}
  \item There is a polynomial $p$ such that for any node $N$, the
    length of the identifier of $N$ is bounded by $p(m)$
    where $m$ is the length (including the total length of the
    edge-labels) of the path from the root to $N$.
  \item There is a deterministic polynomial-time algorithm that, given
    the identifiers of all nodes and the edge labels on the
    path to a node $N$, computes the identifier of $N$.
  \item There is a deterministic polynomial-time algorithm that, given the identifier of a control node $N$,
  the identifiers of all nodes and all edge labels on the path to $N$, computes the
  lexicographically smallest label of an edge (i.e., the in-metadata) of all edges that lead from $N$ to one of its successors.
  \end{itemize}
\end{definition}

% \paragraph{A note on the polynomial-time computability of the computational execution}
% For efficient protocols, a CoSP challenger can be computed by a polynomially time-bounded
%   machine. The deterministic polynomial-time algorithm from
%   Definition~\ref{def:cosp:efficient-protocol} is a compact description of the infinite CoSP
%   protocol. Since the identifier is poly-time computable from the path to the root and the data
%   sent by the attacker, the out-metadata is part of the identifier, all algorithms are poly-time, and
%   the lexicographically smallest edge can be found in poly-time, all messages that the CoSP challenger
%   sends to the attacker can be computed in poly-time.\TODOE{This is only a draft.}

\fi

\pparagraph{Computational indistinguishability}
\iffullversion%
For computational soundness results of equivalence properties in CoSP, the notion of computational indistinguishability that is prevalently used in
cryptography has to be refined to the notion of \emph{termination-insensitive computational indistinguishability} (tic-indistinguishability, in short)~\cite{Un11:tic-ind}.
In comparison to the standard notion of indistinguishability, tic-indistinguishability does not require the
interactive machines to be polynomial-time, but it instead only considers decisions that were made for polynomially-bounded prefixes of the interaction.
This excludes pathological cases in which programs of super-polynomial running time only differ in the response time (which an adversary might
not be able to react to in polynomial time) and in which cryptographic operations might be executed on inputs of super-polynomial length (for which
common cryptographic definitions do not give any guarantees).

\begin{definition}[Tic-indistinguishability]\label{def:tic-indistinguishability}
Given two machines $M,M'$ and a polynomial $p$, we write $\Pr[\interact{M}{M'}\downarrow_{p(\secpar)}x]$ for
the probability that the interaction between $M$ and $M'$ terminates within $p(\secpar)$ steps and $M'$ outputs $x$.

We call two machines $A$ and $B$ \emph{termination-insensitively computationally indistinguishable 
for a machine $\advE$ ($A\tic[\advE] B$)} if for all polynomials $p$, 
there is a negligible function $\mu$ such
that for all $z,a,b\in\bit$ with $a\ne b$, 
\begin{align*}
& \Pr[\interact{A(\secpar)}{\advE(\secpar,z)}\downarrow_{p(\secpar)} a]\\
& + \Pr[\interact{B(\secpar)}{\advE(\secpar,z)}\downarrow_{p(\secpar)} b]
 \le 1 + \mu(\secpar).
\end{align*}
Here, $z$ represents an auxiliary string. 
Additionally, we call $A$ and $B$ \emph{termination-insensitively computationally indistinguishable 
($A\tic B$)} if we have $A\tic[\advE] B$ for all polynomial-time machines $\advE$.

For two CoSP protocols $\proto_1,\proto_2$, a symbolic Model $\Model$, and implementation $\implementations$, if $\Model$ and $\implementations$ are clear from the context, we use for $\Exec_{\Model,\implementation,\proto_1}(\secpar) \tic \Exec_{\Model,\implementation,\proto_2}(\secpar)$ (for sufficiently large $k\in\setN$) the notation $\proto_1 \CoSPIndistinguishable \proto_2$.
\end{definition}

% \begin{definition}[Symbolic execution of a parameterized CoSP protocol]
%   \label{def:sym-execution-of-parameterized-protocol}
%   Let $\Pi$ be a parameterized CoSP protocol and $f_\textit{par}$ be a total function $\textit{par}(\PiS) \to \N$.
%   A \emph{full trace of $\PiS$ with parameters $f_\textit{par}$} is defined like a full trace of $\PiS$ (according to the standard symbolic CoSP execution) with the modification that parameter identifiers are treated like node identifiers and $f_1$ in the first list entry is $f_\textit{par}$ instead of the empty function.
%   \NOTEE{maybe too formal: A list of node identifiers $(\n_i)$ is a \emph{node trace of $\Pi$ with parameters $f_\textit{par}$} if there is a full trace without these node identifiers.}
% \end{definition}
%
% \begin{definition}[Computational execution of a parameterized CoSP protocol]
%   \label{def:comp-execution-of-parameterized-protocol}
%   Let $\proto$ be a parameterized CoSP protocol and $f_\textit{par}$ be a function $\textit{par}(\PiP) \to \set{0,1}^*$.
%   The computational execution of $\PiP$ with parameters $f_\textit{par}$ is the standard computational CoSP execution with the modification that parameter identifiers are treated like node identifiers and
%   in the initial state $f$ is set to $f_\textit{par}$ instead of the empty function.
% \end{definition}
\else
\fi
\iffullversion
\subsection{Computational Soundness}~\label{section:computational-soundness}

The previous notions culminate in the definition of computational soundness for equivalence
properties. It states that the symbolic equivalence of efficient protocols implies
their computational indistinguishability.

\begin{definition}[Computational soundness]\label{def:computational-soundness-equivalence}
Let a symbolic model $\M$ and a class $\protoClass$ of efficient protocols be given. A computational implementation $\implementations$ of $\M$ is \emph{computationally sound}
for $\Model$ if for every pair $\proto_1,\proto_2\in\protoClass$ symbolic equivalence $\proto_1 \CoSPEquivalent \proto_2$ implies $\proto_1 \CoSPIndistinguishable \proto_2$, i.e., tic-indistinguishability (see Definition~\ref{def:tic-indistinguishability}) of their computational execution $\Exec_{\Model,\implementation,\proto_1}(\secpar) \tic \Exec_{\Model,\implementation,\proto_2}(\secpar)$.
\end{definition}
\else%
We finally define \emph{computational soundness} by requiring that symbolic equivalence implies computational indistinguishability.

\begin{definition}[Computational Soundness]\label{def:computational-soundness-equivalence}
 An implementation $\underline \implementation$ of a given symbolic
model $\M$ is \emph{computationally sound}
for $\Model$ and a class $\protoClass$ of efficient pairs of protocols if for every $\Pi \in \protoClass$, we have that $\Pi$ is computationally indistinguishable whenever $\Pi$ is
symbolically equivalent.
\end{definition}
\fi

%%% Local Variables: 
%%% mode: latex
%%% TeX-master: "main"
%%% End:

%!TEX root = main.tex
\section{Dalvik Bytecode (Review)}\label{section:dalvik-bytecode-semantics}

The Abstract Dalvik Language (ADL)~\cite{LoMaStBaScWe_14:cassandra} constitutes the currently  most detailed and comprehensive operational
semantics for Dalvik in the literature, even though it does not encompass
concurrency and exceptions (as e.g.
in~\cite{EnGiChCoJuDaSh_14:taintdroid}). We refer to
Section~\ref{section:related-work} for further discussion on different Dalvik semantics. 
\iffullversion{We only provide a compact review of ADL and refer to~\cite{LoMaStBaScWe_14:cassandra} for more details.}
\else{Due to lack of space, we only provide a very compact review of ADL and refer to~\cite{LoMaStBaScWe_14:cassandra} for more details.}
\fi

%!TEX root = main.tex

\begin{framedFigure}{%
        Selection of inference rules that define $\redRelSt$ with $uop
        \in \mathcal{UNOP}, bop\in \mathcal{BINOP},
        rop\in\mathcal{RELOP}$, with some $n\in\setN$ counter for
        method calls, some ADL program $\proto$, and some method $m$.
        We write $v_a,\dots,v_e$ for $v_a,v_b,v_c,v_d,v_e$, and
        if $r=\defaultRegistersDV(u_1,\ldots,u_l)$, then 
        $r(v_i) = u_i$ for $i\in\set{0,\ldots,l}$ and $\voidDV$
        otherwise.
    }{fig:selected-rules}
\footnotesize
\begin{align*}
        \textsc{rConst}\colon &
        {s \redRelSt s\subst{\mathit{pp}+1, r[v_a \mapsto n]}} &&
        \text{for }
m[pp] = \progcmd{const} v_a,n 
\\
\textsc{rMove}\colon & 
s \redRelSt s\subst{pp+1,r[v_a \mapsto r(v_b)]} &&
\text{for~}
 m[pp] = \progcmd{move} v_a,v_b 
 \\
 \iffullversion%
 \textsc{rUnop}\colon & 
 s \redRelSt s\subst{pp+1,r[v_a \mapsto u]}&&
 \text{for~}
 m[pp] = \progcmd{unop} v_a,v_b,uop \text{~and~} u = \underline{uop}(r(v_b)) 
 \\
 \fi
 \textsc{rBinop}\colon & 
 s \redRelSt s\subst{pp+1,r[v_a \mapsto u]}&&
 \text{for~}
  m[pp] = \progcmd{binop} v_a,v_b,v_c,bop \text{~and~} u = r(v_b) ~\underline{bop}~ r(v_c) 
 \\
 \rIfTestT\colon & 
 s \redRelSt s\subst{pp+n}&&
 \text{for~}
  m[pp] = \progcmd{if\text{-}test} v_a,v_b,n,rop \text{~and~} r(v_a) ~\underline{rop}~ r(v_b) 
 \\
 \rIfTestF\colon & 
 s \redRelSt s\subst{pp+1}&&
 \text{for~}
  m[pp] = \progcmd{if\text{-}test} v_a,v_b,n,rop \text{~and~} \lnot(r(v_a) ~\underline{rop}~ r(v_b)) 
\\
\textsc{rMoveR}\colon &
s \redRelSt s\subst{pp+1,r[v_a \mapsto r(\resultLowerDV)]} &&
\text{for~}
m[pp] = \progcmd{move\text{-}result} v_a 
\\
\textsc{rISt}\colon &
     s
\redRelSt
s\left\{ 
\begin{aligned}
    & m'\cdot m \cdot ml,0\cdot pp \cdot ppl,\\
             & \defaultRegistersDV([r(v_a),\dots,r(v_e)])\cdot r \cdot rl
\end{aligned}
\right\}
&&
\begin{aligned}
\text{for~}
    m[pp] &= \progcmd{invoke\text{-}static} v_a,\dots,v_e,mid \\
    m' & = \lookupStaticDV_\proto(mid)
% mid &\in \dom(\lookupStaticDV_\proto) \\
\end{aligned}
\\
\textsc{rIDR}\colon &
     s
\redRelSt
 s
 \left\{
     \begin{aligned}
         & m' \cdot m \cdot ml, 0 \cdot pp \cdot ppl,\\
 & \defaultRegistersDV([r(v_k),\dots,r(v_{k+n-1})]) \cdot r \cdot rl
     \end{aligned}
 \right\}
 &&
 \begin{aligned}
\text{for~}
m[pp] &= \invokedirectr v_k,n,mid \\
m' & = \lookupDirectDV_\proto(\mathit{mid},h(r(v_k)).\mathsf{class})
% (\mathit{mid},h(r(v_k)).\mathsf{class}) &\in \dom(\lookupDirectDV_\proto) \\
 \end{aligned}
\\
\iffullversion
\rReturnVoid\colon
&
\stateDV[m\cdot ml, h, pp\cdot pp'\cdot ppl,r \cdot r' \cdot rl, as]
    \redRelSt  
&&
    \text{for $m[pp] = \progcmd{return\text{-}void}$ and $ml'\neq
    ()$}
    \\
    & \langle ml,h,(pp'+1)\cdot ppl,\\
& r'[\resultLowerDV \mapsto \lowerDV(\voidDV), \resultUpperDV \mapsto \upperDV(\voidDV)] \cdot rl, as \rangle
\\
\fi
\rReturn\colon &
\begin{aligned}
    & \stateDV[{m\cdot ml,h,pp\cdot pp'\cdot ppl,r\cdot r'\cdot rl,as}]
    \redRelSt \\
    & \langle ml,h,(pp'+1)\cdot ppl,r'[r''] \cdot rl, as \rangle
\end{aligned}
    && 
    \begin{aligned}
    \text{for~}
    m[pp] & = \progcmd{return} v_a\text{~and~} ml'\neq () \\
    r'' & = \resultLowerDV \mapsto \lowerDV(r(v_a)), \resultUpperDV \mapsto \upperDV(r(v_a))
    \end{aligned}
\\
\iffullversion%
\rReturnVoidF\colon &
    \stateDV[(m),h,(pp),(r),as]
    \redRelSt 
    \stateDV[\voidDV,h,as]
 &&
    \text{for~}
m[pp] = \progcmd{return\text{-}void}
\\
\fi
\rReturnF\colon &
     \stateDV[(m),h,(pp),(r),as]
    \redRelSt 
    \stateDV[r(v_a),h,as]
&&
    \text{for~}
m[pp] = \progcmd{return} v_a
    \end{align*}
\iffullversion
for $v_a,v_b,v_c,v_d,v_e \in \mathcal{X}$, $n \in \mathbb{N}$, $mid \in \mathcal{MID}$,  $uop \in \mathcal{UNOP}$, $bop \in \mathcal{BINOP}$, $rop \in \mathcal{RELOP}$ and $u\in\calN$.
\fi
\end{framedFigure}

\iffullversion%
%!TEX root = main.tex
\subsection{Syntax of ADL}

%!TEX root = main.tex
%
\begin{framedFigure}{The set $\mathcal{INSTR}$ of ADL Instructions (excluding all \progcmdraw{\text{-}wide} and \progcmdraw{\text{-}range} variants of commands)}{fig:dalvik-instructions}
	\codesize
\setlength{\columnsep}{0cm}
\begin{multicols}{3}
~\vspace{-2em}
\begin{align*}
&\text{Arithmetic instructions}\\
&\progcmd{move} v_a, v_b\\
% \progcmd{move\text{-}wide} v_a, v_b\\
&\progcmd{const} v_a, n\\
&\progcmd{cmp} v_a, v_b, v_c\\
&\progcmd{unop} v_a,v_b,uop\\
&\progcmd{binop} v_a,v_b,v_c,bop\\
&~\\
&\text{Array-related instructions}\\
&\progcmd{array\text{-}length} v_a, v_b\\
&\progcmd{new\text{-}array} v_a, v_b\\
&\progcmd{filled\text{-}array} v_a, v_b, v_c, v_d, v_e, n\\
% &\progcmd{filled\text{-}new\text{-}array\text{-}range} v_a, n\\
&\progcmd{fill\text{-}array\text{-}data} v_a, u_0, \dots, u_n\\
&\progcmd{aget} v_a, v_b, v_c\\
&\progcmd{aput} v_a, v_b, v_c\\
\end{align*}
\begin{align*}
% &~\\
&\text{Control flow instructions}\\
&\progcmd{nop}\\
&\progcmd{goto} n\\
&\progcmd{if\text{-}test} v_a, v_b, n, rop\\
&\progcmd{if\text{-}testz} v_a, n, rop\\
% ~\\
&\text{Object-related instructions}\\
&\progcmd{instance\text{-}of} v_a, v_b, cl\\
&\progcmd{new\text{-}instance} v_a, cl\\
&\progcmd{const\text{-}string} v_a, s\\
&\progcmd{const\text{-}class} v_a, cl\\
&\progcmd{iget} v_a, v_b, fid\\
&\progcmd{iput} v_a, v_b, fid\\
&\progcmd{sget} v_a, fid\\
&\progcmd{sput} v_a, fid\\
\end{align*}
\begin{align*}
% &~\\
&\text{Method-related instructions}\\
&\progcmd{invoke\text{-}virtual} v_a, v_b, v_c, v_d, v_e, n, mid\\
&\progcmd{invoke\text{-}super} v_a, v_b, v_c, v_d, v_e, n, mid\\
&\progcmd{invoke\text{-}direct} v_a, v_b, v_c, v_d, v_e, n, mid\\
&\progcmd{invoke\text{-}interface} v_a, \ldots,  v_e, n, mid\\
&\progcmd{invoke\text{-}static} v_a, v_b, v_c, v_d, v_e, n, mid\\
&\progcmd{move\text{-}result} v_a\\
&\progcmd{return\text{-}void}\\
&\progcmd{return} v_a
\end{align*}
\end{multicols}
for $v_a,v_b,v_c,v_d,v_e \in \mathcal{X}$, $n \in \mathbb{N}$, $mid \in \mathcal{MID}$, $fid \in \mathcal{FID}$, $cl \in \mathcal{CID}$, $uop \in \mathcal{UNOP}$, $bop \in \mathcal{BINOP}$, $rop \in \mathcal{RELOP}$, and $u_i \in \mathcal{N} \cup \mathcal{L}_c$ (for $i\in\mathcal{N}_0$).
\end{framedFigure}
% \end{definition}
%

ADL uses as syntactical domains seven underspecified sets and several sets of unary, binary, and relation operations.

% : a set of \emph{class names} $\mathcal{CID}_\proto$, a set of of \emph{field names} $\mathcal{FID}_\proto$, a set of \emph{method names} $\mathcal{MID}_\proto$, a set of \emph{fields names} $\mathcal{F}$, a set of string symbols

\begin{definition}[Syntactical domains of ADL]\label{def:syntactical-domains}

The following sets constitute the \emph{syntactical domains of
ADL}.\medskip

\noindent{
\codesize
\begin{tabular}{rlrl}
% \multicolumn{2}{l}{registers}\\
$\mathcal{CONV}$ &: set of symbols for type-casting operations
                 & \multicolumn{2}{l}{underspecified sets}\\
$\mathcal{UNOP}$ &:= $\set{-,\lnot} \cup \mathcal{CONV}$ \text{unary operations}
              & $\mathcal{CID}$ &: set of class names\\
$\mathcal{BINOP}$ &:= $\set{+,-,*,/,\%, \land, \lor, \oplus, <<, >>, >>>}$
              & $\mathcal{FID}$ &: set of field names\\
& \phantom{:= } \text{binary operations}
              & $\mathcal{MID}$ &: set of method names\\
$\mathcal{RELOP}$ &:= $\set{=,\neq,<,>,\le,\ge}$  \text{relations}
              & $\mathcal{S}$ &: set of string symbols\\
$\mathcal{X}$ &:= \set{v_i \mid i \in \setN} \text{register names} 
              & $\mathcal{N}$ &: set of numerical symbols\\
              &
              & &\phantom{: } (e.g., integers and floating point numbers)\\
              &
              & $\mathcal{L}_c$ &: set of constant memory locations\\
              &
              & $\mathcal{F}$ &: set of fields
\end{tabular}}\medskip

\noindent
We require that $\mathcal{MID}$,$\mathcal{CID}$, and $\mathcal{FID}$ are mutually disjoint. A given program typically only uses a subset of the underspecified sets. In these cases, we use the notation $K_y$ to denote that $K_y$ is some subset of $K$ ($K_y \subseteq K$) specific to $y$, for all $K\in\set{\mathcal{CID},\mathcal{FID},\mathcal{MID},\mathcal{S}, \mathcal{N}, \mathcal{F}}$.
\end{definition}

Many ADL instructions have a $\progcmdraw{\text{-}wide}$ variant for 2-register-values and a $\progcmdraw{\text{-}range}$ variant where instead of 5 registers only the starting register and a range $n$ is given and the arguments are read of the $n$ subsequent registers, beginning from the starting register. Figure~\ref{fig:dalvik-instructions} depicts the set of ADL instructions but omits all $\progcmd{\text{-}wide}$ and $\progcmd{\text{-}range}$ variants for the sake of readability, as they are treated analogously. The complete set of instructions can be found in~\cite{LoMaStBaScWe_14:cassandra}.

With the syntactical domains (Definition~\ref{def:syntactical-domains}) and the set of instructions (Figure~\ref{fig:dalvik-instructions}) at hand, we now define the syntax of ADL methods and of ADL programs.

\begin{definition}[ADL methods]\label{def:dalvik-method}
The set $\methodsDV$ of \emph{ADL methods} is defined by $\methodsDV := \mathcal{INSTR}^* \setminus []$, where $\mathcal{INSTR}$ is the set of instructions (Figure~\ref{fig:dalvik-instructions}).
\end{definition}

ADL uses five partial \emph{lookup functions}. These lookup functions refer to the method names that return the field ($\mathfun{lookup\text{-}field}_{\proto}$) and the instruction set with respect to the method's name: for static methods ($\lookupStaticDV_{\proto}$), normal methods ($\mathfun{lookup\text{-}direct}_\proto$), inherited methods ($\mathfun{lookup\text{-}super}$), and virtual methods ($\lookupVirtualDV_{\proto}$).

\begin{definition}[ADL programs]\label{def:adl-program}
An \emph{ADL program} $\proto$ is a tuple 
\begin{center}
\begin{align*}
\proto := ~& (\mathcal{CID}_\proto, \mathcal{FID}_\proto, \mathcal{MID}_\proto, \mathcal{M}_\proto,\mathcal{F}_\proto,\\
& \mathfun{lookup\text{-}field}_{\proto}, \lookupStaticDV_{\proto}, \mathfun{lookup\text{-}direct}_\proto, \mathfun{lookup\text{-}super}_\proto, \lookupVirtualDV_{\proto}), \text{where}
\end{align*}
$\mathcal{CID}_\proto \subseteq_{fin} \mathcal{CID}, \mathcal{FID}_\proto \subseteq_{fin} \mathcal{FID}, \mathcal{MID}_\proto  \subseteq_{fin} \mathcal{MID}, \mathcal{M}_\proto  \subseteq_{fin} \mathcal{M},\mathcal{F}_\proto \subseteq_{fin} \mathcal{F}$,
\begin{align*}
\mathfun{lookup\text{-}field}_{\proto} & : \mathcal{FID}_\proto \rightarrow \mathcal{F}_\proto,\\
\lookupStaticDV_{\proto} & : \mathcal{MID}_\proto \rightharpoonup \mathcal{M}_\proto,\\
\mathfun{lookup\text{-}direct}_\proto & : \mathcal{MID}_\proto \times \mathcal{CID}_\proto \rightharpoonup \mathcal{MID}_{\proto},\\
\mathfun{lookup\text{-}super}_\proto  & : \mathcal{MID}_\proto \times \mathcal{CID}_\proto \rightharpoonup \mathcal{MID}_{\proto},\\
\lookupVirtualDV_{\proto}             & : \mathcal{MID}_\proto \times \mathcal{CID}_\proto \rightharpoonup \mathcal{MID}_{\proto}
\end{align*}
\end{center}
\end{definition}

% An ADL program $\proto$ is formally a tuple $(\mathcal{CID}_\proto$, $\mathcal{FID}_\proto$, $\mathcal{MID}_\proto$, $\mathcal{M}_\proto,\mathcal{F}$, $\mathfun{lookup\text{-}field}_{\proto}$, $\lookupStaticDV_{\proto}$, $\mathfun{lookup\text{-}direct}_\proto$, $\mathfun{lookup\text{-}super}_\proto$, $\lookupVirtualDV_{\proto})$, where $\mathcal{M}_\proto$ denotes a set of \emph{instructions}, $\mathcal{CID}_\proto$ a set of \emph{class names}, $\mathcal{FID}_\proto$ a set of of \emph{field names}, $\mathcal{MID}_\proto$ a set of \emph{method names},  $\mathcal{F}$ an underspecified set of \emph{fields}, and  $\mathfun{lookup\text{-}field}_{\proto}$, $\lookupStaticDV_{\proto}$, $\mathfun{lookup\text{-}direct}_\proto$, $\mathfun{lookup\text{-}super}_\proto$, $\lookupVirtualDV_{\proto}$ five partial \emph{lookup functions}. These lookup
% functions refer to the method names that return the field ($\mathfun{lookup\text{-}field}_{\proto}$) and the instruction set with respect to the method's name: for static methods
%  ($\lookupStaticDV_{\proto}$), normal methods ($\mathfun{lookup\text{-}direct}_\proto$), inherited methods ($\mathfun{lookup\text{-}super}$), and virtual methods
%  ($\lookupVirtualDV_{\proto}$).
%
\else% 
\pparagraph{Syntax}
ADL methods are non-empty sequences of instructions and together constitute a
set $\methodsDV$. A method name $\mathit{mid}\in\midDV$ and a class
name $c\in\mathcal{CID}$ identify the method to be called by means
of lookup tables, \ie, partial functions mapping $(c,\mathit{mid})$ to
normal methods ($\mathfun{lookup\text{-}direct}_\proto$), inherited
methods ($\mathfun{lookup\text{-}super}$), or virtual
methods($\lookupVirtualDV_{\proto}$). In contrast, static methods
($\lookupStaticDV_{\proto}$) are identified by method name alone.
Similarly, fields ($\fieldsDV$) are referenced via $\mathfun{lookup\text{-}field}_{\proto}$.
\fi

\iffullversion%
%!TEX root = ../main.tex
\subsection{The Semantics of ADL}

In this section, we review ADL's operational semantics as far as
needed to understand and state our results; in particular, we only
provide a partial list of ADL's instructions. We refer to the original
ADL paper for further information~\cite{LoMaStBaScWe_14:cassandra}. 

The execution relation defines the operational semantics of ADL. Before we define the execution relation, we introduce the semantical domains of ADL,
A \emph{state}, also called an \emph{intermediate state}, in ADL
consists of a heap $h$, a program counter $pp$, and a set of register
values $r_1, \dots, r_n$, and is extended here by some adversarial
state $\mathit{as}$.
 Formally defining an ADL state and an ADL program requires several additional semantical domains sets and functions.
First, we define the set of \emph{registers}. 
 Let $\locationsDV = \locationsDV_v \cup \locationsDV_c$ denote
 the set of \emph{memory locations}, where  $\locationsDV_v$  is an underspecified set of \emph{variable locations} and $\locationsDV_c$ an underspecified set of \emph{constant memory locations}, with $\locationsDV_v \cap \locationsDV_c = \emptyset$. 
For the purpose of this work, we assume a total order on $\locationsDV$.
 Let $\mathcal{V} = \mathcal{N} \cup \locationsDV \cup \set{\voidDV}$ denote the set
 of \emph{values}, where $\mathcal{N}$ is an underspecified, but finite, set of \emph{numerical values} and $\voidDV$ a distinguished return-value for void-methods. Let $\mathcal{R}=\setN \cup \set{\resultLowerDV,\resultUpperDV} \rightarrow \mathcal{V}$ be the set of \emph{registers}, where $\resultLowerDV$ and $\resultUpperDV$ are reserved registers for return values (split into an upper and a lower part) of method calls. To define the set of heap states, let $\mathcal{AR} = \setN \times (\setN \rightharpoonup \mathcal{V})$. Then the set $\mathcal{H}$ of \emph{heaps} is defined as $\mathcal{H} := (\locationsDV \rightharpoonup (\mathcal{O} \cup \mathcal{AR}))$. Altogether, these notions are referred to as the \emph{semantical domains} of ADL.

Throughout the paper, we distinguish between configurations and states. A configuration describes the state of the program without the attacker. Later, we also exclude the state of the library from the configurations. In ADL, configurations are elements from the set ${\cal C}'$. In particular, we distinguish initial states and initial configurations.

\begin{definition}[Semantical domains]
    \label{def:adl-semantical-domains}
The semantical domains of ADL programs are defined by
\begin{align*}
	\locationsDV & = \locationsDV_c \cup \locationsDV_v & \text{locations} \\
	\valuesDV & = \mathcal{N} \cup \locationsDV \cup \set{\voidDV} & \text{values} \\
	\objectsDV & = \mathcal{CID} \times (\mathcal{F} \rightharpoonup \valuesDV) & \text{objects} \\
	\arraysDV & = \setN \times (\setN \rightharpoonup \valuesDV) & \text{arrays} \\
	\mathcal{X}_{\text{res}} & = \set{\resultLowerDV, \resultUpperDV} & \text{reserved registers} \\
	\mathcal{R} & = (\mathcal{X} \cup \mathcal{X}_{res}) \rightarrow \valuesDV & \text{register states} \\
	\mathcal{H} & = (\locationsDV \rightharpoonup (\mathcal{O} \cup \mathcal{AR})) & \text{heaps} \\
	\mathcal{C}' & = 
    \methodsDV \times \heapsDV \times \setN \times \mathcal{R} \times \advss
    & \text{intermediate configuration} \\
	\mathcal{C} & = 
    \methodsDV \times \heapsDV \times \setN \times \mathcal{R} \times \advss
    & \text{intermediate state} \\
	\mathcal{C}_{\text{final}} & = \valuesDV \times \heapsDV \times \advss 
    \cup \advrespDV
    & \text{final state}
\end{align*}
where $\locationsDV_v$ with $\locationsDV_v \cap \locationsDV_c = \emptyset$ is the set of variable locations, $\voidDV$ is a special value such that $\voidDV \not\in (\mathcal{N} \cup \locationsDV)$, $\resultLowerDV$ and $\resultUpperDV$ are special registers such that $\resultLowerDV,\resultUpperDV \not\in \mathcal{X}$, and $\cal Q$ is the state space of the adversary (see Section~\ref{section:threat-model}). Moreover, there is an efficiently computable injection from $\valuesDV\cup\heapsDV$ into bitstrings, which can be efficiently inverted on its range.
\end{definition}

\else% 
\pparagraph{Semantics}
The semantical domains of ADL programs are defined by a set of
locations $\locationsDV$, numerical values $\mathcal{N}$ and a special
value $\voidDV$, which together form the set of values $\valuesDV$.
Objects carry their class name and a partial function from fields to
values forming the set $\objectsDV$; similarly, arrays carry their
length and a partial function from indices in $\setN$ to values.  
Given a set of registers $\mathcal{X}$ (including two reserved
registers from a set $\mathcal{X}_{res}$, see below), the register state is defined as
$\mathcal{R} = (\mathcal{X} \cup \mathcal{X}_{res}) \rightarrow \valuesDV$. The space of heaps is
$\mathcal{H} = (\locationsDV \rightharpoonup (\mathcal{O} \cup \mathcal{AR}))$ where $\mathcal{AR} = \setN \times (\setN \rightharpoonup \mathcal{V})$. A configuration describes the state of the program without the attacker: $\mathcal{C}'  := 
\methodsDV \times \heapsDV \times \setN \times \mathcal{R} \times \advss$. The full intermediate states additionally contains the attacker states $\mathcal{Q}$ (see Section~\ref{sec:securityframework}): $\mathcal{C}  = 
\methodsDV \times \heapsDV \times \setN \times \mathcal{R} \times \advss$.
Final states are defined as:
($ \mathcal{C}_{\text{final}} = \valuesDV \times \heapsDV \times \advss 
\cup \advrespDV$), where $\advrespDV$ denotes the set of
malicious functions (see Section~\ref{sec:securityframework}).

Throughout the paper, we explicitly distinguish between configurations and states, in particular initial states and initial configurations.
\fi%

\pparagraph{The execution relation $\reduction$}
The operational semantics is defined in terms of an execution relation
$\reduction$ (for an ADL program $\proto$, which we assume
fixed in this section). 
%The complete
%definition of $\reduction_\proto$ is lengthy and hence postponed to the
%full version of this paper \cite{ourfullversion}.
%
For the sake of illustration, Figure~\ref{fig:selected-rules} contains
a representative selection of the rules defining $\reduction$.
For the full set of rules, we refer to the work of Lortz et al.~\cite{LoMaStBaScWe_14:cassandra}.

We use the following notation to shorten presentation and highlight the
modifications applied to the state.
For a state $s = \stateDV\in\mathcal{C}$, we use 
$s\subst{pp+1}$ to denote 
$\stateDV[m\cdot ml,h,pp+1\cdot ppl,r\cdot rl,as]$. Similarly 
$s\set{r[v\mapsto a]}$ denotes 
$\stateDV[{m\cdot ml,h,pp\cdot ppl,r[v\mapsto a]\cdot rl,as}]$, 
$s\set{h[l\mapsto a]}$ denotes $\stateDV[{m\cdot ml,h[l\mapsto a],pp\cdot ppl,r\cdot rl,as}]$, 
$s\set{m'}$ denotes $\stateDV[{m'\cdot ml,h,pp\cdot ppl,r\cdot rl,as}]$,
and $s\set{as'}$ denotes $\stateDV[{m\cdot ml,h,pp\cdot ppl,r\cdot rl,as'}]$.

The relation defines
constant assignment (\textsc{rConst}),
copying of register values (\textsc{rMove}),
% unary operations (\textsc{rUnop}), 
binary operations (\textsc{rBinop}), 
conditional branching (\textsc{rIfTestTrue} and
\textsc{rIfTestFalse}).
Moreover, we depict rules for static method evaluation (\textsc{rISt}) and
evaluation of final methods (\textsc{rIStR}).  
Return values  are stored in distinct
result register $\resultLowerDV, \resultUpperDV \in \mathcal{X}_{res}$
\ffootnote{Although ADL includes two
distinct result registers $\resultLowerDV$ and $\resultUpperDV$ exist, the
upper register $\resultUpperDV$ is only used for the
$\progcmd{move\text{-}result\text{-}wide}$ instruction, where the
return value is copied in two parts.}%
(see \textsc{rReturn}).

We slightly diverge from the
characterization of method calls in \cite{LoMaStBaScWe_14:cassandra} to capture the total
number of computation steps in a run: each transition corresponds to
one computation step, and each rule in Figure~\ref{fig:selected-rules} is
annotated accordingly ($\redRelSt$  instead of $\reduction_\proto$). 

%% \begin{definition}[Execution relation $\redRel{\cdot}{\proto}{m}$]
%Given an ADL program $\proto$ (see Definition~\ref{def:dalvik-program}) and a method $m \in \mathcal{M_\proto}$ (see Definition~\ref{def:dalvik-method}). The \emph{execution relation of ADL} is a relation $\redRel{\cdot}{\proto}{m} \subseteq \states \times (\setN \times \states \cup \stateSpace_{final})$. The relation is defined by inference rules, as in the figures~\ref{fig:selected-rules},\ref{figure:control-flow-instructions},\ref{figure:object-related-instructions},\ref{figure:array-related-instructions},\ref{figure:conversion-sixty-four-instructions}. For the full definition of the execution relation, we refer to the work of Lortz et al.~\cite{LoMaStBaScWe_14:cassandra}.
%% \end{definition}

% vim: ft=tex
 % 1 page

%!TEX root = main.tex
\section{Security framework}
\label{sec:securityframework}
We extend ADL with prbabilistic choices and with a probabilistic polynomial-time attacker that is invoked whenever specific functions are invoked. The modifications to the ADL-semantics are depicted in Figure~\ref{fig:adl-prob-adv}.

\pparagraphorsubsection{Execution and communication model}
ADL as defined in \cite{LoMaStBaScWe_14:cassandra} does not support probabilistic
choice and hence no generation of cryptographic keys
and no  executions of cryptographic functions. We thus first extended ADL with a
rule that uniformly samples a register value from
the set of numerical values $\calN$ (see the \Prob-rule in Figure~\ref{fig:adl-prob-adv}).

\iffullversion
%!TEX root = ../main.tex
To simplify presentation, we
interpret the ADL's semantics as a (generative) probabilistic transition system,
assuming probability $1$ and number of computation steps $1$
for all transitions in
Figure~\ref{fig:selected-rules} and Appendix~\ref{app:adl}.
The symbolic variant of ADL, presented in the next section, will
simplify the adversary by means of a deduction relation, i.e.,
non-deterministic choice over all message deducible by the adversary.
To be able to capture the non-determinism in the symbolic variant, as
well as the probabilism in the computational variant, we chose a model
of a probabilistic transition system similar to the model introduced
by Vardi under the name \emph{concurrent Markov
chains}~\cite{Vardi:1985:AVP:1382438.1382865}, but recast in terms of
a transition labelled system (as opposed to state labelled), and with
the restriction that probabilistic choices are always unlabelled,
which simplifies the definition of parallel composition, which we will
later use to decouple, and substitute attacker as well as library.
Hence, the definition of the distribution of traces below applies to
any probabilistic transition system,
including the split-state compositions from Section~\ref{sec:split-state}.
We furthermore annotate both probabilistic and non-deterministic steps with
the number of computation steps in $\setN$, to be able to argue about
the runtime of a system. 

\begin{definition}[Probabilistic transition system]\label{def:probabilistic-transition-system} %{{{ 
    A \emph{probabilistic transition system} is a quadruple
    $(\states,s_0,\actions,\delta)$ consisting of
    \begin{itemize}
        \item a set of states $\states$,
        % \item a set of initial states $\initstates\subseteq \states$,
        \item an initial state $s_0 \in \states$,
        \item a set of actions $\actions$, and
        \item a transition function $\delta \colon S \to 
            \calD(S \times \setN )  
            \uplus \calP(A\times S \times \setN)$.
    \end{itemize}
    Given $s\in \states$ and $\delta(s) = \mu$,
    we write $ s \predRel{n}{[p]} s'$
    if  $\mu(s',n)=p$.
    If $(a,s',n) \in \delta(s)$, we
    write $s \predRel{n}{a} s'$. In the first case, we speak of
    a probabilistic transition (and a probabilistic state $s$), in the
    second, we speak of a non-deterministic transition (and
    a non-deterministic state $s$).
    If a state $s$ is non-deterministic and
    $\delta(s)=\emptyset$, we also call this state \emph{final}. 
    If a state $s$ is non-deterministic and
    $\delta(s)$ a singleton set or empty set,
    or if $s$ is probabilistic and
    $\delta(s)$ is the Dirac distribution, 
    we also call $s$ \emph{deterministic}. 
    A probabilistic transition system
    is fully probabilistic if all of its states are either
    probabilistic or deterministic.
% definition probabilistic-transition-system (end)
\end{definition}%}}}

Here, $\calD(\Omega)$ denotes the set of all discrete probability
distributions on $\Omega$. A discrete probability distribution on
$\Omega$ is a function $\mu\colon \Omega \to [0,1]$ such that $\set{x
\in \Omega \mid \mu(x)>0}$ (also called the support of $\mu$, denoted
$\supp(\mu)$) is finite or countably infinite, and 
$\sum_{x\in \Omega} \mu(x)  = 1$. We use $\mu(X)$ as short-hand for
$\sum_{x\in X}\mu(x)$ for $X\subseteq \Omega$. For $x\in\Omega$, 
let $\mu^1_x$ denote the Dirac distribution at $x$, i.e., the
distribution with $\mu^1_x(x)=1$.

Having this model in place,
given an ADL program $\proto$,
we interpret $\proto$ as a transition system with
a transition 
$\stateDV \predRel{n}{[p]} \stateDVsSt{'}$
wherever 
$\stateDV \predRel{\Pi,n,p}{~} \stateDVSt{'}$.

\begin{example}[ADL transition system]\label{ex:adl}
    Given an ADL program
    $\Pi$,
    an initial configuration $\stateDVs$,
    and 
    an attacker system $(S_A,s^0_A,\emptyset,\delta_A)$
    in
    $\set{\adv^\secpar}_{\secpar\in\setN}\in\advM$,
    let $\ADL_{\Pi,\adv,\stateDVs}$ 
    be the probabilistic transition system 
    $(\states,\stateDVs,\emptyset,\delta)$,
    where 
    \begin{itemize}
        \item $\states = 
            \calC
            \uplus
            \calC_{\text{final}}
            \uplus
            \advss_\mathit{final}$,
        \item  $s_0 = \stateDV[mk,h,ppl,rl,s^A_0]$,
        \item  
            $\delta(s) = \mu^1_{s',1}$ for $s$ and $s'$,
            according to the rules described in
            Figure~\ref{fig:selected-rules} and Appendix~\ref{app:adl} 
            (observe that at most one rule applies to each state, and that
            each rule uniquely determines the follow-up state),
        \item 
            $\delta(s) = \mu_{s,v_a,\calN}$ where 
            $\mu_{s,v_a,\calN}(s\{r[v_a \mapsto n\},1)=\frac{1}{\calN}$
                for all $n\in\calN$, if
                $s$  matches the pre-condition of \Prob{} in
            Figure~\ref{fig:adl-prob-adv},
        \item 
            $\delta(s) = \mu_{s,\adv}$, where
            $\mu_{s,\adv}
            ( s\subst{pp+1,r[\resultLowerDV \mapsto \lowerDV(\mathit{res}), \resultUpperDV \mapsto
            \upperDV(\mathit{res})],as'},n+1)
            = 
  \Pr [
       \Exec(T) = \langle (\mathit{mid},r(v_a), \dots, r(v_e)),\mathit{as},\epsilon \rangle
            \rightarrow_A q_1 \rightarrow_A \cdots  \rightarrow_A
            q_{n-1}
        \rightarrow_A  \langle i,\mathit{as'},\mathit{res} \rangle ]$
        and
            $\mu_{s,\adv}
            ( \langle \mathit{res} \rangle)
            =  
  \Pr [
       \Exec(T) = \langle (\mathit{mid},r(v_a), \dots, r(v_e)),\mathit{as},\epsilon \rangle
            \rightarrow_A q_1 \rightarrow_A \cdots  \rightarrow_A
            q_{n-1}
        \rightarrow_A  \langle i,\mathit{as'},\mathit{res} \rangle ] $
            for $s$ according to the preconditions of 
            \AdvInv and \AdvRet in the same figure. (Note that
            $s\in \calC$.)

        \item $\delta(\stateDV[u,h,\mathit{as}]) = \mu_\mathit{as}$
            where $\mu_\mathit{as} = 
            \Pr [
                \Exec(T) = \langle (),\mathit{as},\epsilon \rangle
                \rightarrow_A q_1 \rightarrow_A \cdots  \rightarrow_A
                q_{n-1}
                \rightarrow_A \langle i,\mathit{as'},\mathit{res} \rangle
            ]$ and
            $s$ according to the preconditions of \AdvFin.
            (Note that
            $s\in \advss_\mathit{final}$.)
    \end{itemize}
\end{example}

\begin{definition}[Trace distribution (probabilistic)]%
    \label{def:traces-dist}
    Given a fully probabilistic transition system
    $T=(\states,s_0,A,\delta)$, 
    where $\states$ is finite or countably infinite
    and $\epsilon \notin A$,
    we define
    $
        \Pr [ s \xrightarrow{\alpha}_n s' ] = p 
        $
        iff. $s \predRel{n}{[p]} s' \land \alpha=\epsilon$
        or $s \predRel{n}{\alpha} s' \land p=1$.
    We define the outcome probability of an execution as follows:
    \[
        \Pr [ 
        \Exec(T)  =
         s_0 \xrightarrow{\alpha_1}_{n_1} \cdots \xrightarrow{\alpha_n}_{n_n} s_n ]
        =
        \prod^{n-1}_{i=0} \Pr[ s_i \xrightarrow{\alpha_{i+1}} s_{i+1} ].
        \]
    The probability to reach a certain state within $n$ steps is
    defined
    \[ 
        \Pr[T \downarrow_n s_m ] = 
        \sum 
        \Pr \left[ \Exec(T) = s_0 \xrightarrow{\alpha_1}_{n_1} \cdots \xrightarrow{\alpha_m}_{n_m} s_m 
    \text{ and  
    $\sum_{i=1}^{m} n_i \le n$
    }
    \right],
    \]
    % the probability of reaching it \emph{within} $n$ steps is 
    % \[
    %     \Pr[T \downarrow_n s_m ] = 
    % \sum_{n'\le n} \Pr[T \downarrow_{=n'} s_m ], \]
    and the probability of a trace $(a_1,\ldots,a_n)$ as
    \[
        \Pr [ \Traces(T)  = (a_1,\ldots,a_n) ]
        =
        % \sum_{ (s_0 \xrightarrow{\alpha_1}_{n_1} \cdots s_n)\in\supp(\Exec(T)) }
        \sum_{ (\alpha_1,\ldots,\alpha_m)|_A = (a_1,\ldots,a_n)}
        \Pr [ \Exec(T) = s_0 \xrightarrow{\alpha_1}_{m_1} \cdots \xrightarrow{\alpha_m}_{m_m} s_m ].
    \]
\end{definition}

\paragraph{Notational conventions.}
\label{par:notational_conventions_traces}

Within $\Pr[\cdot]$, we abbreviate $s \xrightarrow{\alpha}_n s' $ with
$s \xrightarrow{\alpha} s'$ if $n=1$. We furthermore use $\alpha$
to indicate a transition that might be unlabelled, i.e.,
starts from a probabilistic state, in which case
$\alpha=\epsilon$. We use $a$ instead of $\alpha$, to indicate that the parting
state is indeed non-deterministic.

\else
To simplify presentation, we then  interpret ADL's semantics as
a probabilistic transition system in the following manner.
    A probabilistic transition system
    $T=(\states,s_0,A,\delta)$,
    defines probabilistic and non-deterministic transitions
    that are annotated with a number of
    computation steps, i.e.,
    $\delta \colon S \to 
            \calD(S \times \setN )  
            \uplus \calP(A\times S \times \setN)$.
    We write $ s \predRel{n}{[p]} s'$
    if  $\delta(s) = \mu$ and $\mu(s',n)=p$,
    and
    $s \predRel{n}{a} s'$ if
    $(a,s',n) \in \delta(s)$.
    If $p=1$ or $n=1$, we omit $[p]$ or $n$, respectively.
   If $S$ is finite or countably infinite,
   and 
   $T$ is \emph{fully probabilistic}, i.e., 
   there is no non-deterministic transition with more than one
   follow-up state, we 
   introduce the following random
   variables:
    $
        \Pr [ s \xrightarrow{\alpha}_n s' ] = p 
        $
        iff. $s \predRel{n}{[p]} s' \land \alpha=\epsilon$
        or $s \predRel{n}{\alpha} s' \land p=1$.
    Given any initial state $s_0\in\states^0$, we define the \emph{outcome
    probability} of an execution
    $
        \Pr [ 
        \Exec(T)  =
        s_0 \xrightarrow{\alpha_1}_{n_1} \cdots \xrightarrow{\alpha_n}_{n_n} s_n ]
        =
        \prod^{n-1}_{i=0} \Pr[ s_i \xrightarrow{\alpha_{i+1}} s_{i+1} ]
        $
    and the \emph{probability of a trace} $(\alpha_1,\ldots,\alpha_n)$ as
    $
        \Pr [ \Traces(T)  = (\alpha_1,\ldots,\alpha_n) ]
        =
        \sum
        \Pr [ \Exec(T) = s_0 \xrightarrow{\alpha_1}_{n_1} \cdots \xrightarrow{\alpha_n}_{n_n} s_n ].
        $

\fi

\pparagraphorsubsection{Threat model}\label{section:threat-model}
In this work, we consider adversaries that are network parties or malicious apps that try to retrieve sensitive information from an honest app.
We model the adversary as an external entity that cannot run any code
within the program and that does not have access to the program's
heap, but can read input to and control output from a given set of
\emph{malicious functions} $\advrespDV$. The attacker is a probabilistic polynomial-time
algorithm $\adv$. 

We represent the attacker $\adv$ in ADL as an unlabelled probabilistic transition system. We assume that each state is a triple, where
the first element solely contains the inputs
and the last element contains the outputs. We use the relation $\rightarrow_A$ to denote transitions in the attacker's transition system in Figure~\ref{fig:adl-prob-adv}.
Many computation models can be expressed this way, including, but not
restricted to, Turing machines.

Whenever a malicious function is invoked, $\adv$ is executed with its previous state $as$ and the arguments
$arg$ of the malicious function. The output of $\adv$ is the new
state $as'$ and a response-message 
$\mathit{res}$ in  $\valuesDV$, or in the set $\advrespDV$, which is
distinct from $\valuesDV$.
If $\mathit{res}\in\valuesDV$, it is interpreted as the function's output
values, otherwise, \ie, if the response message is in $\advrespDV$,
the execution terminates with the adversarial
output $\mathit{res}\in\advrespDV$.
Figure~\ref{fig:adl-prob-adv} precisely defines this behavior in
the rule rInvoke-Adv.
%

%!TEX root = ../main.tex

\begin{framedFigure}{Inference rules
    extending ADL with probabilistic semantics and adversarial
interaction.}
   {fig:adl-prob-adv}
% \begin{multicols}{2}
\begin{center}
\begin{inferenceRules}
\inference[\Prob]
{m[pp] = \progcmd{rand} v_a & n'\in \calN }
{s \predRel{~}{\left[\frac{1}{|\calN|}\right]} s\subst{r[v_a \mapsto n']}}\\
\end{inferenceRules}\\[2ex]

\begin{inferenceRules}
\inference[\AdvInv]%
{%
m[pp] = \invokestatic~ v_a,\dots,v_e,\mathit{mid} 
& \mathit{mid} \in \Mal_\mathit{static}
 & \mathit{res}\in\valuesDV
\\
 % \Pr[(as',res) = x :x \leftarrow \adv(as,(r(v_a), \dots, r(v_e)))] = p
  \Pr [
       \Exec(T) = \langle (\mathit{mid},r(v_a), \dots, r(v_e)),\mathit{as},\epsilon \rangle
            \rightarrow_A q_1 \rightarrow_A \cdots  \rightarrow_A
            q_{n-1}
        \rightarrow_A  \langle i,\mathit{as'},\mathit{res} \rangle 
    ] = p
}%
{%
    s 
\predRel{n+2}{[p]}~
% \xrightarrow{(\mathstring{out},\mathit{mid},(r(v_a), \dots, r(v_e)))}_{\proto,n+2,p} 
s\subst{pp+1,r[\resultLowerDV \mapsto \lowerDV(\mathit{res}), \resultUpperDV \mapsto
\upperDV(\mathit{res})],as'}
}
\\[2em]
\inference[\AdvRet]%
{%
m[pp] = \invokestatic~ v_a,\dots,v_e,\mathit{mid}
& \mathit{mid} \in \Mal_\mathit{static}
 & \mathit{res}\in\advrespDV
\\
 % \Pr[(as',res) = x :x \leftarrow \adv(as,(r(v_a), \dots, r(v_e)))] = p
  \Pr [
       \Exec(T) = \langle (\mathit{mid},r(v_a), \dots, r(v_e)),\mathit{as},\epsilon \rangle
            \rightarrow_A q_1 \rightarrow_A \cdots  \rightarrow_A
            q_{n-1}
        \rightarrow_A  \langle i,\mathit{as'},\mathit{res} \rangle
    ] = p
}%
{%
\parbox{8cm}{$s 
% \xrightarrow{(\tifinal,\mathit{res})}_{\proto,n+2,p} 
\predRel{n+1}{[p]}~
\langle \mathit{res} \rangle$}%
}
\\[4ex]
\inference[\AdvFin]%
{%
  \Pr [
      \Exec(T) = \langle (),\mathit{as},\epsilon \rangle
            \rightarrow_A q_1 \rightarrow_A \cdots  \rightarrow_A
            q_{n-1}
        \rightarrow_A \langle i,\mathit{as'},\mathit{res} \rangle
    ] = p
 & \mathit{res} \in\advrespDV
}%
{%
    \stateDV[u,h,\mathit{as}]
    % \xrightarrow{(\tifinal,\mathit{res})}_{\proto,n,p} 
\predRel{n+2}{[p]}
    \stateDV[\mathit{res}]
}
\end{inferenceRules}
\end{center}
\iffullversion
\scriptsize
for $v_a,v_b,v_c,v_d,v_e \in \mathcal{X}$.
\fi
% \end{multicols}
\end{framedFigure}

\iffullversion
\TODOR{We use terminology that is only introduced later.. should we
move the third section before this one?} 
\begin{definition}[Attacker]\label{def:attacker}
    An attacker $\adv$ is a family of 
    fully probabilistic 
    transition systems (cf.
    Definition~\ref{def:probabilistic-transition-system}), 
    indexed by a security parameter $\secpar\in\setN$.
    For each of theses probabilistic transition systems
    $(\states,s_0,\emptyset,\delta)$, the following holds:
    $\states=
    (F_\mathit{mal}, \valuesDV^* )
    \times \advss 
    \times (\valuesDV \uplus \advrespDV \uplus \set{\epsilon})$,
    where
    $F_\mathit{mal}$ is a set of malicious functions,
    $\advss$ is the state-space of the adversary, and
    $\advrespDV$ are the adversarial outputs,
    and each state $s$ in
    $\advss_\mathit{final}\defeq \set{ (\langle
        i,q,\mathit{res}\rangle \mid
        \mathit{res}\in\advrespDV}$ is final, \ie,
        $\delta(s)=\emptyset$.
        Furthermore (for simplicity), we assume that
        every step has computation time 1, i.e.,
        $\delta(s)(s',n)=0$ for all $n\neq 1$.
\end{definition}

We define the adversarial computation model as a set of attackers,
which we call $\advM$.

\fi

\iffullversion
\begin{definition}
    There is a subset of $\midDV_\proto$ which the adversary controls,
    called \Mal. No element of $\Mal$ is in the range of any
    lookup-tables \lookupVirtualDV, \lookupStaticDV, \lookupDirectDV,
    \lookupSuperDV. \Mal is partitioned into the sets 
    $\Mal_\mathit{virtual}$,
    $\Mal_\mathit{static}$,
    $\Mal_\mathit{super}$, and
    $\Mal_\mathit{direct}$.
\end{definition}
\fi

\iffullversion%
We obtain the following definition for executing ADL in the presence of an adversary.
\else%
We can cast the ADL semantics in terms of a probabilistic transition
system $\ADL_{\proto,\adv,\stateDVs}$ that is parametric in the program
executed, the attacker (see below) and the initial configuration. 
We obtain the following definition for executing ADL in the presence
of an adversary.
\fi%
\begin{definition}[ADL Execution with Adv]
    Given an  ADL program $\proto$, an  adversary $A$,
    and an initial configuration $\stateDVs$,
    the probability that the interaction between $\proto$ on
    $\stateDVs$ and $\adv$ 
    terminates within $n$ steps
    and
    results in $x$
    is defined as 
    \iffullversion $$ \else \begin{multline*}\fi
    \Pr[
        \interact{ \proto\stateDVs} {\adv} \downarrow_{n}x] \defeq 
        \iffullversion\else\\\fi
        \Pr[\ADL_{\proto,\adv,\stateDVs}\downarrow_n \langle x \rangle ].
	\iffullversion $$ \else \end{multline*}\fi
	where we write (given $\adv$'s initial state $as$) $\proto\stateDVs$ for the program $\proto$ with initial state $\stateDV[{ml,h,ppl,rl,as}]$.
\end{definition}

\pparagraphorsubsection{Calls to crypto APIs} %
\iffullversion%
The most direct way of modelling the crypto-API would 
take the symbolic model as a starting point,
defining
static functions without any side-effect, one of which can be used to
generate nonces or keys, while the others apply constructors or
destructors. However, as our results shall be applicable to Android
Apps, we take existing crypto-APIs as the starting point, since 
the symbolic model is meant to abstract
existing crypto-APIs, rather than demand a ``matching implementation''.
To this end, it is worth 
having a look at how 
Java's Standard Cryptography Interface is typically used.

\begin{lstlisting}[language=java,
        caption={Excerpt from 
        \url{https://www.owasp.org/index.php/Using_the_Java_Cryptographic_Extensions}},
    label=lst:javax2]
KeyGenerator keyGen = KeyGenerator.getInstance("AES");
keyGen.init(128);
SecretKey secretKey = keyGen.generateKey();
byte[] iv = new byte[128 / 8];
SecureRandom prng = new SecureRandom();
prng.nextBytes(iv);
Cipher c = Cipher.getInstance("AES/CBC/PKCS7PADDING");
c.init(Cipher.ENCRYPT_MODE,
       secretKey, 
       new IvParameterSpec(iv));
byte[] byteCipherText = c.doFinal(byteDataToEncrypt);
\end{lstlisting}

We observe that here, \texttt{c} is a function object used for encryption and
decryption, which is initialised with the cypher to be used, mode of
operation and additional parameters. Consequently, when the actual
computations are performed in the last line, the fields of this object
carry information important to its operation, not only the arguments
supplied. 
Hence, the cryptographic abstraction of this library
depends not only on the method called (\texttt{doFinal} in this case),
but also on the object in the heap. Moreover, the question
\emph{which} cryptographic abstraction is chosen, e.g., encryption,
decryption, authenticated encryption, etc., depends on object in the
heap.

Furthermore, these functions are
typically \texttt{final}, meaning that they cannot be overwritten by
sub classes. This is (most likely) compiled to a direct call of said
function, implemented in the instruction $\invokedirectr$ 
(see Figure~\ref{fig:dalvik-instructions}). From the point of view of
this paper, invokevirtual, invokedirect,
\progcmd{invoke\text{-}super} and their 
\progcmd{\text{-}range} counterparts behave very similarly, which is
why we concentrate on the case relevant for Java's standard
cryptographic interface, \ie, \invokedirectr. But it is worth noting,
that static function calls (\invokestatic) are not used.

We can also observe, that
the randomness used for encryption (\texttt{iv}) is
generated using a PRNG, but user-supplied, while the
\texttt{generateKey} method choses a key without exposing the
randomness used. But even for keys, the randomness can be supplied by
the user, as demonstrated by the following (insecure) example.

\begin{lstlisting}[language=java,caption={Excerpt from 
        \url{https://gist.github.com/bricef/2436364}},
        label=lst:javax1
        ]
static String IV = "AAAAAAAAAAAAAAAA";
static String plaintext = "test text 123\0\0\0"; /*Note null padding*/
static String encryptionKey = "0123456789abcdef";

Cipher cipher = Cipher.getInstance("AES/CBC/NoPadding", "SunJCE");
SecretKeySpec key = new SecretKeySpec(encryptionKey.getBytes("UTF-8"), "AES");
cipher.init(Cipher.ENCRYPT_MODE, key,new IvParameterSpec(IV.getBytes("UTF-8")));
return cipher.doFinal(plainText.getBytes("UTF-8"));
\end{lstlisting}

On the other hand, it is also possible to let the \texttt{Cipher}
object pick the initialisation vector itself (and request it via
a method call in order to, e.g., attach it to the ciphertext).
This is implemented in Java via overloading, i.e., the method id is
known at compile time based on the number and types of arguments.

To summarize these observations:
\begin{itemize}
    \item Cryptographic output may depend on values in the heap.
    \item Which cryptographic abstraction is appropriate may depend on
        values in the heap, too.
    \item Cryptographic calls are calls to non-static functions.
    \item Cryptographic APIs provide interfaces for random number
        generation and key-generation with or without explicit
        randomness.
\end{itemize}

We can validate these observations on
the open source library
bouncycastle\footnote{http://www.bouncycastle.org}, which is very popular
in mobile applications, as it is fairly lightweight. The interface
it provides also permits user-supplied randomness, while at the same
time retaining state within the \texttt{cipher} object.
\begin{lstlisting}[language=java,
    caption={Excerpt from
    \url{https://www.bouncycastle.org/specifications.html}},
    label=lst:bouncy1]
BlockCipher engine = new DESEngine();
BufferedBlockCipher cipher = new PaddedBlockCipher(new CBCCipher(engine));
byte[] key = keyString.getBytes();
byte[] input = inputString.getBytes();
cipher.init(true, new KeyParameter(key));
byte[] cipherText = new byte[cipher.getOutputSize(input.length)];
int outputLen = cipher.processBytes(input, 0, input.length, cipherText, 0);
cipher.doFinal(cipherText, outputLen);
\end{lstlisting}

Note that it is also possible to use bouncycastle as a security
provider for Java's standard cryptographic interface.

Since randomness might be implicit, or explicit, depending on how
a method is called, 
rather than mapping methods to (the evaluation of) a constructor or
destructor on its arguments, we chose to generalize this to arbitrary
combinations of constructors, destructors and randomness generation.
Symbolic operations (see
Definition~\ref{def:symbolic-operation}) nicely capture the concept. Hence, a library
specification maps a method and a predicate on the object's current
state (\ie, a subset of $\objectsDV$) to a symbolic operation. The input
to the symbolic operation contains the method's arguments, as well as
a defined part of the heap (cf. Definition~\ref{def:pre-compliant}).
\else%
%
%\TODOR{This might need to go, but I want to offer it -- Robert}
%Consider the following example of how Java's Standard Cryptography
%Interface is used:
%\begin{lstlisting}[language=java,
%        % caption={Excerpt from
%        % \url{https://www.owasp.org/index.php/Using_the_Java_Cryptographic_Extensions}},
%    label=lst:javax2]
%KeyGenerator keyGen = KeyGenerator.getInstance("AES");
%keyGen.init(128);
%SecretKey secretKey = keyGen.generateKey();
%byte[] iv = new byte[128 / 8];
%SecureRandom prng = new SecureRandom();
%prng.nextBytes(iv);
%Cipher c = Cipher.getInstance("AES/CBC/PKCS7PADDING");
%aesCipherForEncryption.init(Cipher.ENCRYPT_MODE, secretKey, 
%                new IvParameterSpec(iv));
%byte[] byteCipherText = c.doFinal(byteDataToEncrypt);
%\end{lstlisting}
%
%Here \texttt{c }is a function object used for encryption and
%decryption, which is initialised with the cypher to be used, mode of
%operation and additional parameters, \ie, the fields of this object
%carry information important to its operation. 
%Hence, the cryptographic abstraction of this library
%depends not only on the method called (\texttt{doFinal} in this case),
%but also on the object in the heap. Furthermore, these functions are
%typically \texttt{final}, meaning that they cannot be overwritten by
%sub classes. This is (most likely) compiled to a direct call of said
%function, implemented in the instruction $\invokedirectr$ 
%(see Figure~\ref{fig:dalvik-instructions}). Observe further, that
%the randomness used for encryption (\texttt{iv}) is
%user-supplied, while the \texttt{generateKey} method choses a key
%without exposing the randomness used.
%
Libraries that offer cryptographic operations leave the generation of randomness either implicit or explicit. Our result encompasses both settings. Moreover, most classes for cryptographic operations, such as the commonly used \texttt{javax.crypto.Cipher}, are used with a method call that lacks some arguments, e.g., the encryption key when using \texttt{javax.crypto.Cipher.doFinal()} to account for situations in which the missing argument has been stored earlier in a member variable, e.g., using \texttt{javax.crypto.Cipher.init()}. We address this by considering the class instance (from the heap) as an additional argument to the method call that has been marked, e.g., \texttt{javax.crypto.Cipher.doFinal()}. Since the class instance contains more information than the arguments of the cryptographic operation of the computationally sound symbolic model, all relevant information such as the encryption key, plaintext, and randomness has to be extracted from the class instance. The derived functions that we explicitly added in our formalization take care of this extraction.
%
%Technically, we show in Corollary~\ref{corollary:derived-destructors} in the full version that 
%derived destructors, defined by symbolic operations, can be added as destructors to
%the symbolic model in a computationally sound way.
\fi

\begin{definition}[Library specification]
    A library specification is an efficiently computable partial function
    $\libSpec: \midDV \times \objectsDV
    \cup \mathcal{UNOP}  
    \cup \mathcal{BINOP}  
    \cup \mathcal{RELOP}  
    \rightharpoonup \SO 
    $
    defined at least on
    $\mathcal{UNOP}  \cup \mathcal{BINOP}  \cup \mathcal{RELOP}$.
\end{definition}

\iffullversion
Note that operations on bitstrings are also specified by $\libSpec$.  
In Corollary~\ref{corollary:derived-destructors} we show that via
symbolic operations, destructors derived from these operations
can also be added as destructors (together with potentially some fresh nonces)
to the symbolic model in a computationally sound way. \TODOE{Replace derived destructors with symbolic operations.}
As we will see later, this allows us to retain more precision, as
bitstrings obtained from the crypto-API, e.g., via decryption, can
still be treated as bitstrings.
\fi

In order to be able to use an asymptotic security definition, we
define uniform families of ADL programs as programs generated from
a security parameter.

\begin{definition}[Uniform families of \iffullversion ADL \fi programs]%
    \label{def:adl-family}
    Let $\proto$ be an algorithm that,
    given a security parameter $\secpar$,
    outputs an ADL program. We denote the output of this program
    $\proto^\secpar$ and call
    the set of outputs of this program a
    \emph{uniform family of ADL programs}.
\end{definition}

Next, we define initial configurations for families (indexed by a security parameter) of transition systems as states that are valid initial configuration for all security parameters.
\begin{definition}[Initial configuration]\label{def:initial-state}
Given a family $(T_\secpar)_{\secpar}$ of transition systems. We say that a state $s$ is an \emph{initial configuration} for $(T_\secpar)_{\secpar}$ if it is a valid initial configuration for all $T_\secpar$ in this family. Analogously, we say that a configuration $s$ is an initial configuration for a uniform family of ADL programs if for all $\secpar$ this $s$ is a valid initial configuration for $\proto^\secpar$.
\end{definition}

\iffullversion
Real-world cryptographic libraries require each
cipher before use to be initialised with the key-length, \eg,
the class methods 
\texttt{getInstance} in \texttt{javax.crypto.Cipher} is called with
a string specifying cipher, mode of operation, key-length and possibly
more (see, e.g., Listing~\ref{lst:javax1}).
Even though the key-length cannot be arbitrarily large (in fact, the
choice is quite limited), we think this comes reasonably close to how
asymptotic security is achieved in real life: when standardisation
bodies or security experts advise developers to chose larger keys,
some constant in the existing source code is adapted to change the
parameters to the security library, \eg, in a preprocessor step. 
We are aware that this requirement is not met by most real-world
cryptographic libraries. 
% We intentionally leave at this point a gap
% between our theoretical result and the actually deployed program.
Still, we consider this gap between the actually deployed programs and
the theoretical result to be significantly smaller than in previous
results for actual programming
languages~\cite{BaHoUn_09:cosp,BaMaUn_10:rcf,AiGoJu_11:extracting-c}.
\fi

\pparagraphorsubsection{Indistinguishability of two ADL programs}
Similar to CoSP, we define indistinguishability of ADL programs using tic-indistinguishability~\cite{Un11:tic-ind}.
%
% \iffullversion
% Lortz et al.~\cite{LoMaStBaScWe_14:cassandra} define the
% indistinguishability of two ADL programs in a very strict sense: two
% programs are indistinguishable if they have the same heap and the same
% registers. We could extend this notion with common notions of
% computational non-interference~\cite{La_02:computational-ni}.
\begin{definition}[Indistinguishability (ADL)]%
    \label{def:ADL-indistinguishability}
We call two uniform families of ADL programs 
$\proto_1= \set{\proto^\secpar_1}_{\secpar\in\setN}$
and 
$\proto_2= \set{\proto^\secpar_2}_{\secpar\in\setN}$
with initial configuration 
$s_1=\stateDVsN{_1}$ 
and
$s_2=\stateDVsN{_2}$ 
\emph{computationally indistinguishable}
for a (not necessarily uniform) family of attackers 
$\adv = \set{\adv^\secpar}_{\secpar\in\setN}$
$\advE$ (written $\protos 1\ADLIndistinguishable^\adv \protos 2$)
if for all polynomials $p$, 
there is a negligible function $\mu$ such
that for all $a,b\in\bit$ with $a\ne b$, 
\[
    \Pr[ \interact{\protoss 1} {\adv^\secpar}
\downarrow_{p(\secpar)} a]
   + \Pr[
   \interact{\protoss 2}{\adv^\secpar}
  \downarrow_{p(\secpar)} b]
 \le 1 + \mu(\secpar).
 \]
We call $\protos 1$ and $\protos 2$ \emph{computationally indistinguishable 
($\protos 1 \ADLIndistinguishable \protos 2$)} if we have $\protos 1 \ADLIndistinguishable^\advE \protos 2$ for all machines $\advE$.

\TODOE{Should we not use a different $p$ for $\protoss 2$, i.e., something like ``for all $p$ there is a $p'$''?}
\end{definition}

This notion indistinguishability gives rise to 
notion of non-interference, when $\proto_1=\proto_2$. 
\TODOR{What can we say about this notion of non-interference?}
% \else
% We call two ADL uniform families of ADL programs $\proto_1$ and $\proto_2$
% \emph{computationally indistinguishable} ($\ADLIndistinguishable$), if $\proto_1$ and $\proto_2$ are tic-indistinguishable for all attackers $\advE$.
% \fi
 % 1 page

%!TEX root = main.tex

\section{Symbolic Dalvik Bytecode}\label{sec:dalvik-symbolic-semantics}

In this section, we define a symbolic variant of ADL, which uses symbolic terms 
instead of cryptographic values. We sometimes abbreviate this symbolic variant as \ADLs.
We will show in the next section that it suffices to analyze a (symbolic) ADLs program in order to prove the corresponding (cryptographic) ADL program secure,
provided that the cryptographic operations used in that program are computationally sound, i.e., provided that they have a computationally sound symbolic
model in CoSP%
\iffullversion%
in the sense of Definition~\ref{def:computational-soundness-equivalence}).
\else.
\fi%
ADLs' semantics precisely corresponds to the semantics of ADL, except for the treatment of cryptographic operations.
As a consequence, existing automated analysis tools can be conveniently extended to ADLs and, thus, accurately cope with
cryptographic operations, since this extension only requires semantic adaptations precisely for those cases where cryptographic 
behavior needs to be captured.
ADLs is parametric in the symbolic model of the considered cryptographic operations and, hence, benefits
from the rich set of cryptographic primitives that are already supported
by computational soundness results in CoSP, such as encryption and signatures~\cite{BaHoUn_09:cosp,Me_10:ind-cpa-cosp,BaMoRu_14:equivalence,BaMaUn_12:without-restrictions-cosp} or zero-knowledge proofs \cite{BaBeUn_13:cosp-zk,BaBeMaMoPe_15:malleable-zk}.

\subsection{Embeddable symbolic CoSP models}

\iffullversion We first define sufficient conditions under which a given CoSP model can be embedded into ADL\@.
To this end, we require that values in $\calN$ can be embedded into CoSP, \ie, there
needs to be an injective function from $\calN$ into \T, so that, e.g., 
a register value can be encrypted.
As in previous
embeddings~\cite{BaHoUn_09:cosp,BaMaUn_10:rcf}, we require that the
symbolic model includes an equality operation $\equals$. 

\begin{definition}[ADL-embeddable symbolic model]\label{def:adl-embeddable-symbolic-model}
    A symbolic model $\M=(\C,\N,\T,\D)$ (see Definition~\ref{def:symb.model}) is \emph{ADL-embeddable} if
\begin{itemize}

\item for $S=
    \valuesDV  \cup 
    \heapsDV\cup
    \Mal_\mathit{static}\cup
    \mathcal{UNOP}\cup
    \mathcal{BINOP}\cup
    \mathcal{RELOP}\cup
    \mathcal{CONV}\cup
    \set{0,1}$,
    there exists an injective function 
    $\iota: S \to \T$
    from $S$ to terms consisting only of constructors,
    such that all $n\in S$ can be distinguished using deconstructors, \ie,
    for each $n\in S$,
    there are deconstructors $D_1,\ldots,D_m\in\D$ such that
    for any $n'\in S$,
    \[ 
        D_1 \circ \cdots \circ D_m \circ \iota(n')=
    \begin{cases}
        \iota(n') & \text{if }n'=n\\
        \bot & \text{if }n'\neq n.
    \end{cases}
\]
We assume the inverse $\iota^{-1}$ on the range of $\iota$
    to be efficiently computable.

\item The destructor distinguishing $0\in S$ is called \iszero.

\item There is a destructor $\equals/2 \in \D$ such that for all $x,y\in\T$, $x\neq y$, $\equals(x,x) = x$ and $\equals(x,y) = \bot$.

\item There are $\pair/2\in\C$ and $\fst/1,\snd/1\in\D$ such that for all
    $,y\in\T$,  
    $\fst(\pair(x,y)) = x$ and 
    $\snd(\pair(x,y)) = y$.

\item For each $n$-ary operation $\mathit{op}\in\mathcal{UNOP}\cup\mathcal{BINOP}\cup\mathcal{RELOP}$, there is a destructor 
    $\mathit{op}_s$ such that for all $a_1, \dots, a_n\in\valuesDV$,
    $\mathit{op}_s(\iota(a_1),\dots,\iota(a_n))=\iota(\underline{\mathit{op}}(\iota^{-1}(a_1),\dots,\iota^{-1}(a_n)))$.
\end{itemize}
We abbreviate $\widehat n:= \iota(n)$, and call $\widehat n$ the
symbolic representation of $n$. We lifts this notion to sets:
$\forall N\subseteq \calN. \widehat N = \set{ \widehat n \mid n \in N}$.
\end{definition}
 \else

We first define sufficient conditions under which a given CoSP model
can be embedded into ADL.  To this end, we require that, e.g., values in
$\valuesDV = \calN\cup\locationsDV\cup \set{\voidDV}$
 can be embedded into CoSP, \ie, there needs to be an
injective function from $\valuesDV$  into \T. 
Formally, we assume that there exists an injective function
$\iota$ such that all $v\in\valuesDV$ can be
distinguished using deconstructors, \ie, $ \forall
v,v'\in\valuesDV
\exists D_1,\ldots,D_m\in\D$, $\iota(v') = v$ if $v'=v$ and $\bot$
otherwise.
As in previous
embeddings~\cite{BaHoUn_09:cosp,BaMaUn_10:rcf}, we require that the
symbolic model includes an equality destructor $\equals$ and a pairing
constructor $\pair$. 
We define $\widehat v = \iota(v)$ if $n\in\valuesDV$ and the
identity if $n\in\T$, and 
lift this notion to sets:
$\forall V \subseteq \valuesDV. \widehat V = \set{ \widehat v \mid v \in
v}$. 
\fi

\iffullversion

With Corollary~\ref{corollary:bitstring-operation-extension}, all
recent CoSP-results from the
literature~\cite{BaHoUn_09:cosp,Me_10:ind-cpa-cosp,BaMaUn_12:without-restrictions-cosp,BaBeUn_13:cosp-zk,BaMoRu_14:equivalence,BaBeMaMoPe_15:malleable-zk}
satisfy these requirements.
We show that the first
condition of
Definition~\ref{def:adl-embeddable-symbolic-model} can easily be
satisfied for ADL values, if we instantiate the (underspecified)
set $\valuesDV$ with 
bitstrings. 

\begin{example}[ADL-embeddable symbolic model for bitstrings and heaps]%
    \label{ex:bitstring-representation}
    
    Let $\valuesDV=\bits*$, \ie, numerical values, locations (and
    $\voidDV$) are expressed via bitstrings, but are distinguishable,
    \ie, via a tagging convention. We show how an embedding of
    bitstring can be
    achieved with a
    symbolic model $\M=(\C,\N,\T,\D)$ which includes
    \begin{itemize}
        \item $\equals/2 \in \D$ such that for all $x,y\in\T$, $x\neq
            y$, $\equals(x,x) = x$ and $\equals(x,y) = \bot$, and
        \item $\payloadString_0/1,\payloadString_1/1,\payloadEmpty/0\in \C,\payloadUnstring_0/1,\payloadUnstring_1/1\in\D$ such that for all $x\in\T$ $$\payloadUnstring_0(\payloadString_0(x)) = x \text{ and } \payloadUnstring_1(\payloadString_1(x)) = x,$$ and in all other cases $\payloadUnstring_0(x) = \bot$ and $\payloadUnstring_1(x) = \bot$.
    \end{itemize}
    In this case, 
$\iota_{\valuesDV}: \bits*\to\T$
    can be defined as follows:
    \[
        \iota( b_l \ldots b_0) 
     := \payloadString_{b_{l}}(
        \payloadString_{b_{l-1}}(
        \dots(\payloadString_{b_0}(\payloadEmpty))\dots)).
\]
    
    The symbols in 
    $\Mal_\mathit{static}\cup
    \mathcal{UNOP}\cup
    \mathcal{BINOP}\cup
    \mathcal{RELOP}\cup
    \mathcal{CONV}\cup
    \set{0,1}$
    can be expressed using the same method via an embedding into bitstrings,
    or even simpler, by having one constructor and one destructor per
    element, as these sets are a-priori fixed.
    One could make the same argument for the heap, but as symbolic
    abstractions may produce updates to the heap, i.e.,
    ``partial'' heaps that overwrite the heap where they are defined
    (cf. Section~\ref{sec:adls-semantics}), 
    and symbolic verification tools do not support arbitrary
    destructors, we propose a modelling 
    as a list of pairs of argument
    and function value. As the heap can contain objects, as well as
    arrays, these need to be embedded as well.
    \TODOR{We could go into more detail here, even define this
        thoroughly, but I am not sure any one is interested. Discuss.
    (Notes in Latex comments.)}
    % Ok, so there is a pair constructor heappair/2, and heapcons/2
    % heapnil/0 for the list. A deconstructor heapat/2 searches for the value
    % recursively, i.e.,
    %   heapat( loc, nil) = void
    %   heapat( loc, heapcons(heappair(loc,val),l)) = val
    %   heapat( loc, heapcons(heappair(loc',val),l)) = heapat(loc,l)
    %           if loc≠loc'
    %
    %  IF val is in \valuesDV
    %
    %  This means heaps that are not well-formed, e.g. assign values
    %  twice, have a clear interpretation (constructor application is
    %  update).
    %
    %  We can do this for all partial function, so also for the
    %  assignment array indexes to values.
    %
    %  Objects have a class id and a partial mapping from fields to
    %  values. As the fields of an object are known a-priori, we have
    %  for every object of classid cid and field values
    %   a constructor
    %
    %  cid(f_1,..,f_n) where n is |F|
    %
    %  Then there is a deconstructor returning the class id, another
    %  one returning each field by name, and an update constructors
    %  for each field as follows. Say f_i has name fn
    %
    %  cid_update_fn( cid(f_1,..,f_i,..,f_n),v) = cid(f_1,..,v,..,f_n)
    %
    %  if v is in values.
\end{example}

 For illustration consider the bitstring $\mathstring{01101}$. Its
 \emph{symbolic representation} is the term 
 \[ 
     t:=
     \widehat{01101} =
    \payloadString_0(
    \payloadString_1(
    \payloadString_1(
    \payloadString_0(
    \payloadString_1(\payloadEmpty))))).
 \]
If a term $t$ represents the bitstring $\mathstring{01101}$,
the following equality test evaluates to $\payloadEmpty$:
{\iffullversion \scriptsize \else \tiny \fi
\begin{align*}
& \equals(\payloadUnstring_1(\payloadUnstring_0(\payloadUnstring_1(\payloadUnstring_1(\payloadUnstring_0(t))))),\payloadEmpty)\\
& = \equals(\payloadUnstring_1(\payloadUnstring_0(\payloadUnstring_1(\payloadUnstring_1(\payloadUnstring_0(\payloadString_0(\payloadString_1(\payloadString_1(\payloadString_0(\payloadString_1(\payloadEmpty)))))))))),\payloadEmpty)\\
& = \equals(\payloadUnstring_1(\payloadUnstring_0(\payloadUnstring_1(\payloadUnstring_1(\payloadString_1(\payloadString_1(\payloadString_0(\payloadString_1(\payloadEmpty)))))))),\payloadEmpty)\\
& =
\equals(\payloadUnstring_1(\payloadUnstring_0(\payloadUnstring_1(\payloadString_1(\payloadString_0(\payloadString_1(\payloadEmpty)))))),\payloadEmpty)\\
& =
\equals(\payloadUnstring_1(\payloadUnstring_0(\payloadString_0(\payloadString_1(\payloadEmpty)))),\payloadEmpty)\\
& =
\equals(\payloadUnstring_1(\payloadString_1(\payloadEmpty)),\payloadEmpty)\\
& =
\equals(\payloadEmpty,\payloadEmpty)\\
& =
\payloadEmpty
\end{align*}
}In Appendix~\ref{app:extendability}, we show that any computationally
sound symbolic model that contains symbolic bitstrings and equality
can be extended by destructors implementing any polynomial-time
computable function on bitstrings, allowing us to perform, \eg, binary
operations like XOR on the symbolic representation of bitstrings.
\else
\NOTEE{I do not see why the paper would significantly benefit from mentioning that we proved that any binary (deterministic) poly-time computable function on symbolic bitstrings as a destructor can be added to an existing symbolic model without loosing computational soundness.}
% We stress that any computationally
% sound symbolic model, which contains symbolic bitstrings and equality
% can be extended by destructors implementing any polynomial-time
% computable function on bitstrings, allowing us to perform, \eg, binary
% operations like XOR on the symbolic representation of
% bitstrings~\cite{ourfullversion}.
\fi

\subsection{Semantics of symbolic ADL}\label{sec:adls-semantics}
The semantical domains of the symbolic
variant of ADL coincide with the semantical domains of ADL except for the set
of values $\valuesDV$ and the set of (intermediate and final) states.
\iffullversion The set of values is the union of symbolic terms and
values as defined previously. \else The set of values is extended to
hold terms in addition to ADL values:
$\valuesDV = \mathcal{N} \cup \locationsDV \cup \set{\voidDV} \cup
\T$.
\fi

% \input{old/dalvik-modified-symbolic-inference-rules}
%!TEX root = main.tex

\begin{framedFigure}{Modified inference rules of the symbolic execution relation $\reduction\text{-}Sym$, where $s = \stateDV[{ml,h,pp \cdot ppl,r\cdot rl}] $.
        \iffullversion%
    (Variable ranges as in Figure~\ref{fig:dalvik-instructions}.)
    \fi%
}{fig:modified-inference-rules-symbolic}
\begin{center}
\codesize
\begin{align*}
\textsc{rIDR-s:}~&
     s
     \symtrans{%
     (\tiHL,\mathit{mid},r(\underline {v_k}),\sliceof{r(v_k)}),
     (\tiLH,(u_\mathit{lo},u_\mathit{up},h'))
     }
% & \stateDV[m'\cdot ml,h,0\cdot ppl,\defaultRegistersDV([r(v_k),\dots,r(v_{k+n-1})])\cdot rl,as]
&& \text{for~}
    m[pp] = \invokedirectr v_k,n,mid \\
    &    
    s\subst{h[h'],pp+1\cdot ppl,r[
            \resultLowerDV \mapsto u_\mathit{lo},
            \resultUpperDV \mapsto u_\mathit{up}
    ]\cdot rl}%
% (\mathit{mid},h(r(v_k)).\mathsf{class}) &\in \dom(\lookupDirectDV_\proto) \\
    % m' &= \lookupDirectDV_\proto(\mathit{mid},h(r(v_k)).\mathsf{class}) \\
&&\land  O = \libSpec(mid,h(r(v_k))) \\
&&&\land 
     \pair(t,h') = \eval_O(r(\underline {v_k}),\widehat h)
    \\
\textsc{ISt-m:}~&
s
\symtrans{(\tiHA,\mathit{mid}, r(v_a),\dots,r(v_e)),%
(\tiAH,(u_\mathit{lo},u_\mathit{up}))%
}
&& \text{for~}
m[pp] = \invokestatic v_a,\dots,v_e,mid\\
&
s\subst{pp+1\cdot ppl,r[\resultLowerDV \mapsto u_\mathit{lo}, \resultUpperDV \mapsto u_\mathit{hi}]\cdot rl}   
&&
\land mid \in \Mal_\mathit{static} \land V' \vdash u_\mathit{lo} \land V' \vdash u_\mathit{up} 
\\
\textsc{Unop-m:}~&
 s 
 \symtrans{(\tiHA,\mathit{uop},r(v_b)),(\tiAH,(u,u))}
&& \text{for~}
m[pp] = \progcmd{unop} v_a,v_b, \mathit{uop} \\
&
s\subst{pp+1\cdot ppl,r[v_a \mapsto u]\cdot rl}
&&\land r(v_b) \not\in \T_\calV \land V' \vdash u
\\
\textsc{Unop-l:}~&
 s \symtrans{(\tiHL,\mathit{uop},r(v_b)),(\tiLH,(u,u))}
&& \text{for~}
m[pp] = \progcmd{unop} v_a,v_b, \mathit{uop} \\
&
s\subst{pp+1\cdot ppl,r[v_a \mapsto u]\cdot rl}
&& \land r(v_b)\in\T_\calV \land u = d_{uop}(r(v_b))
\\
\iffullversion
\textsc{Binop-m:}~&
 s 
 \symtrans{(\tiHA,\mathit{bop},r(v_b),r(v_c)),(\tiAH,(u,u))}
&& \text{for~}
m[pp] = \progcmd{binop} v_a,v_b,v_c,\mathit{bop} \\ 
&
s\subst{pp+1\cdot ppl,r[v_a \mapsto u]\cdot rl}
&&
(\land r(v_b) \not\in \T_\calV 
\lor r(v_c) \not\in \T_\calV )
\land V' \vdash u
\\
\textsc{Binop-l:}~&
s \symtrans{(\tiHL,\mathit{bop},r(v_b),r(v_c)),(\tiLH,(u,u))}
&& \text{for~}
m[pp] = \progcmd{binop} v_a,v_b,v_c,\mathit{bop} \\ 
&
s\subst{pp+1\cdot ppl,r[v_a \mapsto u]\cdot rl}
&& \land r(v_b),r(v_c)\in\T_\calV \land u = d_{bop}(r(v_b))
\\
\fi
\textsc{IfTest-m:}~&
 s
 \symtrans{(\tiHA,\mathit{rop},r(v_a),r(v_b))}
s\subst{(pp+n')\cdot ppl,r\cdot rl}
&& \text{for~}
    m[pp] = \progcmd{if\text{-}test} v_a,v_b,n,rop \\
	&&&\land \lnot(r(v_a)\in\T_{\valuesDV} \land r(v_b)\in\T_{\valuesDV}) \land n' \in \set{n,1}
\end{align*}
\end{center}
\scriptsize
\iffullversion
for $v_a,v_b,v_c,v_d,v_e \in \mathcal{X}$, $mid \in \mathcal{MID}$,  $uop \in \mathcal{UNOP}$, $bop \in \mathcal{BINOP}$, $rop \in \mathcal{RELOP}$ and 
$u,u_\mathit{lo},u_\mathit{up}\in\T$.
\fi
\end{framedFigure}

We add to each state in the inference rules the list of terms output
so far, called \emph{the symbolic view}. This symbolic view
corresponds to the symbolic view in the symbolic execution of a CoSP
protocol (Definition~\ref{def:symbolic-execution}), and is used to
track the information the adversary can use to deduce messages.

\iffullversion
\begin{definition}[Symbolic semantical domains]\label{def:adls-dom}
The \emph{symbolic semantical domains} of ADL programs are defined as
in Definition~\ref{def:syntactical-domains}, with the exception
of

\begin{align*}
	\valuesDV^\T & = \mathcal{N} \cup \locationsDV \cup \set{\voidDV} \cup \T & \text{values} \\
	\mathcal{C}^\T & =  \mathcal{H} \times \setN \times \mathcal{R} \times \Views & \text{intermediate states} \\
	\mathcal{C}_{\text{final}}^\T & = \valuesDV \times \mathcal{H} \times \Views & \text{final states}
\end{align*}
\end{definition}
\fi

Transition between states are now annotated with the information sent
to or received from the adversary. As ADL is sequential, every
output to the adversary is followed by the adversary's input.
Calls to malicious functions reveal register values to the adversary,
which can be terms, locations, or numerical values. Terms are added to
the adversarial knowledge as they are; locations and values are
translated by means of the embedding.
Calls to the crypto-API immediately yield a term by applying
a constructor or destructor to the term arguments, or their symbolic
representation.
As elaborated in the introduction, applying binary operations (\eg, XOR) to
symbolic terms usually invalidates computational soundness results. We over-approximate this by
treating symbolic terms as a blackbox and let the outcome of any
operation involving a symbolic term be decided by the adversary. As
a side effect, we let the adversary learn the operands. %
As binary operations and tests on bits
are only defined on inputs in 
$\calN$,
rules involving these
% \textsc{rBinop},
% \textsc{rIfTestTrue}, and
% \textsc{rIfTestFalse}
remain unaltered and
are complemented by the rules in
Figure~\ref{fig:modified-inference-rules-symbolic}. 
\iffullversion\else%
For space reasons, the figure does not contain the rules for binary operations (\textsc{Binop}), but they are defined analogously to unary operations. 
\fi%
For brevity, we let $V'$
denote the attacker's view, which is part of the successor state, and
write $V'\vdash m$, if the adversary
can deduce a term $m$ from a view $V$, \ie, there is a symbolic
operation $O$ such that $O(\termsof(V))=m$.
Observe that for
\textsc{IfTest-m},
the adversary can decide the outcome without learning the operands.
This highlights the non-determinism inherent to the symbolic semantics.
\iffullversion%
We thus define the probabilistic transition system
$\ADLs_{\Pi,\stateDVs}$ 
similar to Example~\ref{ex:adl}, but with these modifications.
\else%
Let $\ADLs_{\Pi,\stateDVs}$ denote the 
probabilistic transition system
describing these modified semantical domains and transitions.
\fi

With the symbolic semantics in place, we can finally define the notion
of symbolic equivalence between two programs. We first specify what
constitutes a symbolic view, as 
\iffullversion%
Definition~\ref{def:traces-dist}
\else%
our definition of traces
\fi%.
applies only to fully probabilistic transition systems.

\begin{definition}[Symbolic view\iffullversion (probabilistic transition system)\fi]
    Given a probabilistic (but not necessarily fully probabilistic)
    transition system $T$ with initial state $s_0$,
    the set of symbolic views of $T$ is
    defined:
    \[
        \Views(T) = 
        \set{ 
         (\alpha_1,\ldots,\alpha_m)|_{\mathit{Views}}
         \middle|
         s_0 \xrightarrow{\alpha_1}_{n_1} \cdots \xrightarrow{\alpha_m}_{n_m} s_m 
 }\]
\end{definition}

\iffullversion
As in CoSP, we use the equivalence relation on traces
introduced in Definition~\ref{def:symbolic-knowledge}.
\else
We use the notion of symbolic equivalence introduced in
Section~\ref{section:cosp}.
\fi

\begin{definition}[Symbolic equivalence]\label{def:symb-equivalent}
Let $\Model$ be a symbolic model. 
    Two probabilistic transition systems $T_1$ and $T_2$ are 
    \emph{symbolically equivalent}
    ($T_1 \execEquivalent T_2$)
    if $\Views(T_1) \sim \Views(T_2)$ w.r.t. \Model.
    Two uniform families of ADL programs 
    $\proto_1= \set{\proto^\secpar_1}_{\secpar\in\setN}$
    and 
    $\proto_2= \set{\proto^\secpar_2}_{\secpar\in\setN}$
    and initial configurations 
    $s_1=\stateDVsN{_1}$ 
    and
    $s_2=\stateDVsN{_2}$ 
    are \emph{symbolically equivalent} 
    ($\protos 1 \ADLEquivalent \protos 2$)
    iff,
    for each \TODOE{why not ``sufficiently large''?} $\secpar$,
    the probabilistic transition systems
    $T_1 = \ADLs_{\proto_1,s_1}$
    and
    $T_2 = \ADLs_{\proto_2,s_2}$
    are symbolically equivalent.
\end{definition}

\NOTEE{It would be nice to say something about the fact that other papers do not have a security parameter in the symbolic model. We intentionally do have it, since in practice one is only interested in one security parameter; hence verification per security parameter is feasible. Though actually we have a gap in the proof: we show that if a protocol is symbolically equivalent for sufficiently large $\secpar$, then it is secure for sufficiently large $\secpar$. This does not mean that verifying it for one $\secpar$ implies security for the same $\secpar$.}

 % 0.5 pages

%!TEX root = main.tex

\section{Computational Soundness}
\label{sec:soundness}
Establishing a computational soundness proof for ADL requires  a clear separation between honest program parts and cryptographic API calls with their corresponding augmented adversarial symbolic capabilities. To achieve this, we characterize this partitioning by introducing the concept of a \emph{split-state form} of an operational semantics, and we subsequently
show that it can be naturally used to represent ADL. This section, presents the essence of our proof. For the full proofs and detailed definition, we refer to the technical report~\cite{ourfullversion}.

%!TEX root = main.tex
\subsection{Split-state semantics}
\label{sec:split-state}
The split-state form 
partitions the original semantics into three components, parallely
executed asynchronously: $(i)$ all steps that belong to computing
cryptographic operations (called the \emph{crypto-API semantics}),
$(ii)$ all steps that belong to computing the malicious functions
(called the \emph{attacker semantics}), and $(iii)$ all steps that
belong to the rest of the program (called the \emph{honest-program
semantics}). 
%Moreover, we define explicit transitions between each of these
%components, which gives rise to a precise message-passing interface
%for cryptographic operations or for communicating with the attacker.
The overall operational semantics is a composition of each of these
sub-semantics, i.e., the state space is the Cartesian product of the
sub-states with additional information on which entity is currently
running, defined in terms of asynchronous parallel composition.
\unless\iffullversion
We refer to the technical report for details~\cite{ourfullversion}.
\fi
In a split-state form, all three entities can synchronize through the following sets of labels:
$\tHL$,
$\tLH$,
$\tHA$, and
$\tAH$. 
The first two model message passing between the honest
program and the crypto-API; the latter two model message passing between the honest program and the attacker. There
is no synchronization step between attacker and the crypto-API, as the attacker
can only stop the execution with a transition of the form
$(\tifinal,m)$ for some message $m$.

\iffullversion
The overall operational semantics, called the \emph{split-state semantics}, will then be defined as the composition of each of these three sub-semantics, i.e., 
the state space is the Cartesian product of the sub-states with additional information on which entity is currently running, defined  in
terms of asynchronous parallel composition.
The benefit of this notion is that it makes the communication explicit that
occurs between the actual program and the adversary. This communication is often
fixed but arbitrary, \ie, in the case of black-box usage. 
%TODOR{I do not understand what this means... looking for
%clarification.}
Moreover, it
separates the actual program from the cryptographic API calls, which is crucial for computational soundness proofs.
Finally, the details of the crypto-API are irrelevant for program
analyses, provided that the cryptographic operations are implemented securely.

% \begin{definition}[Operational semantics]
% Let $Prog$ be an underspecified set, which we call \emph{set of programs}.
% A triple $(\states, \init, \delta)$ is an \emph{operational semantics} if $\init : Prog \rightarrow \stateSpace$ is a mapping from programs to states and $\delta : \stateSpace \times \stateSpace$ is a relation such that $(\stateSpace, \delta(\init(\proto)))$ is a transition system. \TODOE{Define a transition system.}
% \end{definition}

\paragraph{Defining split-state.}
In the first step, we represent the respective program semantics of
the adversary, the cryptographic library and the honest programs in
terms of three (possibly different) labelled  transition systems.
In each of those, every
state corresponds to a configuration of the program, 
the set of
initial states corresponds to the set of initial configurations,
and
the transition relation is defined by the set of possible execution
steps. 
Execution steps that involve any of the other entities are
labelled either:
\begin{itemize}
    \item $( \tiHL,f,m)$
            when the honest program
        makes a library call to the function named $f$
        with input $m$.

    \item $(\tiLH,m)$ for the library's response
    \item $(\tiHA,m)$ when the honest program passes message $m$ to
        the adversary, and
    \item  $(\tiAH,m)$ in the opposite case.
\end{itemize}

We can describe our notion of split-state semantics using the
well-studied notion of 
\emph{asynchronous parallel composition}\cite{milner1989communication,hoare1985communicating}\footnote{To be more precise:
CSP-style parallel composition on common actions.},
which
we review below. For readability, we define the transition function in
terms of the transition relation it induces.

\begin{definition}[Asynchronous parallel composition]%
    \label{def:asynchronous-parallel-comp}
    The asynchronous parallel composition of two 
    probabilistic transition systems
    $(\states_1,s^0_1, A_1, \delta_1)$
    and
    $(\states_2,s^0_2, A_2, \delta_2)$
    is the probabilistic transition system 
    $(\states_1 \times \states_2,
     (s^0_1,s^0_2),
     (A_1 \cap A_2),
     \delta)$,
    where 
    $\delta$ is defined such that:
    \begin{itemize}
        \item 
            $(s_1,s_2) \predRel{n_1+n_2}{a} (s_1',s_2')$
            iff 
            $s_1 \predRel{n_1}{a} s_1'$,
            $s_2 \predRel{n_2}{a} s_2'$
            and 
            $a\in A_1 \cap A_2$,
        \item $(s_1,s_2) \predRel{n_1}{a} (s_1',s_2)$
            iff $s_1 \predRel{n_1}{a} s_1' \wedge a\notin A_2$,
        \item $(s_1,s_2) \predRel{n_2}{a} (s_1,s_2')$
            iff $s_2 \predRel{n_2}{a} s_2' \wedge a\notin A_1$, and
        \item $(s_1,s_2) \predRel{n}{[p]} (s_1',s_2')$ iff
            $s_1 \predRel{n}{[p]} s_1' \land s_2'=s_2$ 
            or 
            $s_2 \predRel{n}{[p]} s_2' \land s_1'=s_1$.
    \end{itemize}
\end{definition}

    We extend this form of composition to an arbitrary number of
    transition systems by applying it to the first and second element,
    then applying the result to the third, and so forth. We simplify
    notation by flattening the resulting state space to $\states_1
    \times \states_2 \cdots \times \states_n$ instead of
    $ (\cdots (\states_1 \times \states_2 )\cdots) \times
    \states_n$, modifying initial configurations and transition
    relation accordingly. 

    A split-state composition is an asynchronous parallel composition
    with the following restrictions:
    \begin{inparaenum}[\itshape a\upshape)]
        \item We specify which transitions can be performed in
            synchronisation. Assuming some set of function symbols
            $F$,
            some set of symbols for malicious functions $F_\mathit{mal}$ including
            $\tifinalCall$,
            and a domain for messages $\calD$, we define:
            \begin{align*}
                \tHA: & = \set{(( \tiHA,f,m_1,\ldots,m_n),1,1) \mid f\in F_\mathit{mal}, m_i \in \calD}\\
                \tAH: & = \set{(( \tiAH,m),1,1) \mid  m\in \calD},\text{\ and}\\
                \tHL: & = \set{(( \tiHL,f,m_1,\ldots,m_n),1,1) \mid f\in F, m_i\in \calD},\\
                \tLH: & = \set{(( \tiLH,m),1,1) \mid  m\in \calD},\\
                \tfinal & = \set{(( \tifinal,m),1,1) \mid  m\in \calD}.
            \end{align*}
            The first two allow for message passing between the honest
            program and the cryptographic library, the latter two for
            message passing between honest program and attacker. There
            is no synchronising step between attacker and library.
        \item We enrich the three transition system's state with
            a boolean value that indicates whether the system is
            currently active. Only the active system can perform
            non-synchronising steps or decide to pass information to
            another system via a synchronising step -- but in this
            case, it activates the other system and ceases to be
            active itself.
    \end{inparaenum}

    To simplify notation, we write 
    $s \xrightarrow{A} s'$
    if there exists $a\in A$ such that $s
    \xrightarrow{a} s' $.

\begin{definition}[Split-state operational semantics]\label{def:split-state-semantics}
    Given
three transition systems
$(\states_H, s^0_H, A_H, \delta_H)$,
$(\states_A, s^0_A, A_A, \delta_A)$,
$(\states_L, s^0_L, A_L, \delta_L)$,
such that
$A_H \subseteq \tHA\cup \tAH \cup \tHL \cup \tLH$,
$A_A \subseteq \tHA\cup \tAH \cup \tfinal$,
$A_L \subseteq  \tHL \cup \tLH$,
% $A_H \cap A_A \subseteq \tAH \cup \tHA$,
% $A_H \cap A_L \subseteq \tHL \cup \tLH$,
% $A_A \cap A_L = \emptyset$,
% $A_H\cap\tfinal = A_L\cap\tfinal = \emptyset$,
we define the split-state composition as
the asynchronous parallel composition
of 
$(\mathbb{B} \times \states_H, 
(\top,s^0_H),
A_H,
\delta_H')$,
$(\mathbb{B} \times \states_A, 
(\bot,s^0_A),
A_A,
\delta_A')$,
$(\mathbb{B} \times \states_L, 
(\bot,s^0_L),
A_L,
\delta_L')$,
where
\begin{align*}
\delta_H'(b,s) & =
        \begin{cases}
            \mu \circ (((\top,s),n)\mapsto (s,n)) & \text{if $b=\top$ and  $\delta_H(s)=\mu$} \\
            ((\tiHA,\tifinalCall),(\bot,s),1) & \text{if $b=\top$ and $s$ final}\\
            \delta_H(s)|_{\tHL\cup\tHA}^\bot & \text{if $b=\top$ and $s$
        non-determ. and not final} \\
        \delta_H(s)|_{\tLH\cup\tAH}^\top & \text{if $b=\bot$ and $s$
        non-deterministic} 
        \end{cases}
        \\
\delta_A'(b,s) & =
        \begin{cases}
            \mu \circ (((\top,s),n)\mapsto (s,n)) & \text{if $b=\top$ and  $\delta_A(s)=\mu$} \\
            \delta_A(s)|_\tAH^\bot & \text{if $b=\top$ and $s$
        non-deterministic} \\
            \delta_A(s)|_\tHA^\top & \text{if $b=\bot$ and $s$
        non-deterministic} 
        \end{cases}
        \\
\delta_L'(b,s) & =
        \begin{cases}
            \mu \circ (((\top,s),n)\mapsto (s,n)) & \text{if $b=\top$ and  $\delta_L(s)=\mu$} \\
            \delta_L(s)|_\tLH^\bot & \text{if $b=\top$ and $s$
        non-deterministic} \\
            \delta_L(s)|_\tHL^\top & \text{if $b=\bot$ and $s$
        non-deterministic} 
        \end{cases}
\end{align*}
Here, $P|_S^b$ is short-hand for 
$ \set{(a,(b,s'),n) \mid (a,s',n)\in P \land a \in S}$.
An operational semantics 
is a \emph{split-state semantics}, if it 
can be expressed as a split-state composition of three transition systems.

We call states of the honest program semantics \emph{configurations}, and write $T(s)$ for a split-state semantics $T$ meaning that the initial state of the honest program semantics $T_H$ is replaced by $s$.
\end{definition}

We stress that requiring an $\tiHA$-transition being always followed
by a $\tiAH$-transition is without loss of generality. In case the
original would solely send a message or solely receive a message, one
of the transitions would only carry dummy information.
\else

An operational semantics 
is a \emph{split-state semantics}, if it 
can be expressed as a transition system
$(\states, \initstates,\rightarrow)$
by split-state composition of three transition systems.

\fi

The benefit of this notion is that it makes the communication explicit that
occurs between the actual program and the adversary. This communication is often
fixed but arbitrary, \ie, in the case of black-box usage. Moreover, it
separates the actual program from the cryptographic API calls, which is crucial for computational soundness proofs.
Finally, the details of the crypto-API are irrelevant for program
analyses, provided that the cryptographic operations are implemented securely.

\iffullversion
This notion makes these distinction more explicit and facilitates the
proof, but note that the split-state semantics can be defined in
a meaningless way: the library could be empty, and never be called, or
return values from the adversary be wholly ignored. We will illustrate
in the following section, however, that through small additions, the
ADL semantics can be separated, providing 
\begin{inparaenum}[\itshape a\upshape)]
\item for a generalized proof method, \ie, Lemma~\ref{lemma:embedding-soundness-computational} and~\ref{lemma:embedding-soundness-symbolic} apply to all split-state semantics, not only
    ADL,
\item a communication model that makes the threat model and the
    assumptions on the library implementation obvious.
\end{inparaenum}
\TODOE{Consistently write American English or British English, e.g., modelling vs modeling.}
\fi

A split-state semantics is not necessarily a probabilistic transition
systems, as the transition systems it is composed from might be  
non-deterministic. However, we define the
probability of a certain outcome only for probabilistic split-state
semantics, as in most applications for security, non-determinism in
the honest program semantics stems from the modelling of concurrency
and is usually resolved by specifying a scheduler. 
This non-determinism is typically conservatively resolved by assuming that the
adversary controls the scheduling. 
%%%%%
\ifdraft \TODOE{This still needs to be incorporated into the embedding. Add this text then again.}
In this case, $F_\mathit{mal}$
may contain a symbol $\texttt{ctl}$ that is used to send scheduling
requests of form $(\tiHA,\texttt{ctl},l)$ to the adversary, where $l$
is a unique label identifying, \eg, the current position in the
execution, entering a `pending' state distinct from, but unique to,
the previous state.
From this state, the possible follow-up states are reachable via
transitions labelled $(\tiAH,l')$, where $l'$ is a distinct label for
each follow-up state.
Despite the fact that ADL programs are strictly
sequential, we will define an embedding into CoSP that translates
these scheduling requests into control nodes. This imposes some
light restrictions: there can only be a countable number of labels 
for follow-up states $l'$, and each $l'$ needs to be polynomial in the length of the smallest path to it.
For a reasonable computational model, it should be possible for the
attacker to compute each $l'$, but our soundness result does not rely
on this fact.
\fi
%%%%%

\iffullversion
\begin{definition}[Split-state execution]%
    \label{def:split-state-execution}

    Let $T_H,T_L,T_A$ be transition systems and
    assume that 
    their split-state composition $T$ is probabilistic.

    We write $\Pr[ T(s_0) \downarrow_{n}x]$ for
    the probability that the interaction between between
    honest program, attacker and library 
    with initial state $s_0\in\initstates$
    results in $x$ and terminates within $n$ steps, as follows:
    \[
        \Pr[ T(s_0) \downarrow_{n}x]
        =\\
    \sum_{ (s_0 \xrightarrow{\alpha_1} \cdots s_k)\in\supp(\Exec(T)) }
    \Pr \left[\Exec(T) = s_0 \xrightarrow{\alpha_1}_{n_1} s_1 \cdots
        \xrightarrow{\alpha_k = (\tifinal,x)}_{n_k}
    s_k
    \text{ and  
    % $\sum_{i=1}^{i\le k} n_i \le n$}
    $k \le n$}
    \right].
\]
\end{definition}
\else
We write $\Pr[ T(s_0) \downarrow_{n}x]$ for
the probability that the interaction between
honest program, attacker and crypto-API
with initial state $s_0\in\initstates$
results in $x$ and terminates within $n$ steps. Split-state indistinguishability is
then defined as follows.
\fi

\begin{definition}[Split-state indistinguishability]%
    \label{def:ss-indistinguishability}
    Let $T_{H,1}$, $T_{H,2}$, $T_A$, and $T_L$ be families of transition systems
    indexed by a security parameter in $\setN$. 
    We call $T_{H,1}$ and $T_{H,2}$  
    \emph{computationally indistinguishable}
    for $T_L$ and $T_A$
    and initial states
    $\initH{1}$,
    $\initH{2}$,
    $\initL$ 
    of $T_{H,1}$, $T_{H,2}$, $T_L$ (respectively) in the sense of Definition~\ref{def:initial-state}
    if for all polynomials $p$, 
    there is a negligible function $\mu$ such
    that for all $a,b\in\bit$ with $a\ne b$, 
    \begin{align*}
        &
        \Pr[ 
            T_1^\secpar((\top,\initH{1}),(\bot,(0,z)),(\bot,\initL))
                \downarrow_{p(\secpar)} a]\\
                & + \Pr[
                T_2^\secpar((\top,\initH{2}),(\bot,(0,z)),(\bot,\initL))
            \downarrow_{p(\secpar)} b]
            \le 1 + \mu(\secpar),
        \end{align*}
        where 
        $T_1^\secpar$ is the split-state composition of
        $T_{H,1}^\secpar,T_L^\secpar,T_A$, and
        $T_2^\secpar$ is the split-state composition of
        $T_{H,2}^\secpar,T_L^\secpar,T_A$
         for security parameter $\secpar$. In short, we write $T_{H,1}(\initH{1}) \SSIndistinguishable{T_L(\initL),T_A} T_{H,2}(\initH{2})$.
        Additionally, we call $T_{H,1}$ and $T_{H,2}$ 
        \emph{computationally indistinguishable 
        (written as $T_{H,1}(\initH{1}) \SSIndistinguishable{T_L(\initL)} T_{H,2}(\initH{2})$)}
        if 
        $T_{H,1}(\initH{1}) \SSIndistinguishable{T_L(\initL),T_A} T_{H,2}(\initH{2})$
        for all $T_A\in\advM'$ and for initial states $\initH{1},\initH{2},\initL$ for $T_{H,1},T_{H,2},T_L$ (respectively).
    \end{definition}
\if\else\fullversion
Note that the state state of a split-state composition has form 
$(
(b_H,s_H),
(b_A,s_A),
(b_L,s_L)),$
where $b_H$ is $\top$ iff the honest program is active and similar for
$b_A$ and $b_L$.
\fi

%!TEX root = main.tex
\subsection{Split-state representation of ADL}\label{section:split-state-adl}

\begin{framedFigure}{ADL split-state representation, honest program semantics, $s = \stateDV[ml,h,ppl,r\cdot rl]$}{fig:adl-to-adlss}
\begin{center}
\footnotesize
\begin{align*}
\LibCall\colon &
    s \xrightarrow{(\tiHL,\mathit{mid},r(v_k),\dots,r(v_{k+n-1}),\sliceof{r(v_k)})}_H (\wait,s)
&& 
    \begin{aligned}
        &   \text{for~}
(mid,h(r(v_k))) \in \dom(\libSpec) 
\\ 
& m[pp] = \invokedirectr v_k, n, \mathit{mid}
% o = \libSpec(\mathit{mid},h(r(v_k)).\mathfun{class})
    \end{aligned}
\\
\LibResponse\colon &
(\wait,s) \xrightarrow{(\tiLH,(m_\mathit{lo},m_\mathit{up},h'))}_H 
s\left\{
    \begin{aligned}
        & h[h'],pp+1\cdot ppl,\\
        & r\left[
        \resultLowerDV \mapsto m_\mathit{lo},
        \resultUpperDV \mapsto m_\mathit{up}
\right]\cdot rl 
    \end{aligned}\right\}
&& 
\begin{aligned}%
    & \text{for~}   
     (mid,h(r(v_k))) \in \dom(\libSpec) \\
% o = \libSpec(\mathit{mid},h(r(v_k)).\mathfun{class})
& m[pp] = \invokedirectr v_k, n, \mathit{mid}
\end{aligned}
\\
\LeakMsg\colon &
s \xrightarrow{(\tiHA,\mathit{mid},{r(v_a)},\dots,{r(v_e)})}_H
(\wait,s)
&& 
\begin{aligned}
    & \text{for~} mid \in \Mal_\mathit{static} \\
    & m[pp] = \invokestatic v_a,\dots,v_e,mid 
\end{aligned}
\\
\ReceiveMsg\colon &
(\wait,s) \xrightarrow{(\tiAH,m_\mathit{lo},m_\mathit{up})}_H s
\left\{
    pp+1\cdot ppl,
    r\left[ \begin{aligned}
        \resultLowerDV & \mapsto m_\mathit{lo},\\
        \resultUpperDV & \mapsto m_\mathit{up}
    \end{aligned} 
        \right]\cdot rl
\right\}
&& \begin{aligned}
    &\text{for~} mid \in \Mal_\mathit{static} \\
& m[pp] = \invokestatic v_a,\dots,v_e,mid 
\end{aligned}
\\
\FinalCall: &
\stateDV[u,h] \xrightarrow{(\tiHA,\tifinalCall)}_H \stateDV[u,h] 
&&
\end{align*}
\end{center}
\iffullversion%
\scriptsize
for $v_a,v_b,v_c,v_d,v_e,v_k \in \mathcal{X}$, $n \in \setN$, $mid \in \mathcal{MID}$, 
and $m_\mathit{lo},m_\mathit{up}\in\valuesDV$
\fi
\end{framedFigure}

\begin{framedFigure}{ADL split-state representation, library program semantics, $\_$ 
        is an arbitrary adversarial state.
    }{fig:adl-to-adlss-lib}
\begin{center}
\footnotesize
\begin{align*}
\LLibCall: &
    \stateDV[(),\emptyset,(),(),\_] \xrightarrow{(\tiHL,\mathit{mid},r(\underline{v}),h)}_L
    \stateDV[{(m),h,(0),(\defaultRegistersDV([r(\underline {v})])),\_}]
&& \begin{aligned}
    & 	\text{for~}
	 r(\underline {v}) := r(v_k),\dots,r(v_{k+n-1}), \\
         & m \defeq \lookupDirectDV_\proto(\mathit{mid},h(r(v_k)).\mathsf{class})
	\end{aligned}
\\
\LLibRetVoid: &
\stateDV[(m),h,(pp),(r),\_] \xrightarrow{(\tiLH,(\voidDV,\voidDV,h))}_L 
\stateDV[(),\emptyset,(),(),\_]
&& \text{for~}
m[pp]  = \progcmd{return\text{-}void} 
\\
\LLibRet: &
\stateDV[(m),h,(pp),(r),\_] \xrightarrow{(\tiLH,(\lowerDV(r(v_a)),\upperDV(r(v_a)),h))}_L 
\stateDV[(),\emptyset,(),(),\_]
&& \text{for~}
m[pp] = \progcmd{return}~v_a
\end{align*}
\end{center}
\iffullversion
\scriptsize
for $v_a,v_b,v_k \in \mathcal{X}$, $n \in \setN$, $mid \in \mathcal{MID}$.
\fi
\end{framedFigure}

\begin{framedFigure}{ADL split-state representation, adversary semantics, 
    $\varepsilon$ denotes the empty string.}{fig:adl-to-adlss-att}
\begin{center}
\footnotesize
\begin{align*}
\ALeakMsg: &
    \stateDV[\varepsilon,as,\varepsilon] 
    \xrightarrow{(\tiHA,\mathit{mid},{v_a},\ldots,{v_e})}_A
    \stateDV[(\mathit{mid},{v_a},\ldots,{v_e}),as,\varepsilon]
\\
\AReceiveMsg: &
    \stateDV[i,as,\mathit{res}] 
    \xrightarrow{(\tiAH, (\lowerDV(\mathit{res}), \upperDV(\mathit{res}))) }_A
    \stateDV[\epsilon, as',\epsilon]%
&& \text{for~}
\mathit{res}\in \valuesDV
\\
\AFinal: &
    \stateDV[i,as,\mathit{res}] 
    \xrightarrow{(\tifinal,\mathit{res})}_A 
    \stateDV[\epsilon, as',\epsilon]%
&& \text{for~}
res \in\advrespDV
\end{align*}
\end{center}
\iffullversion
\scriptsize
for $v_a,v_b,v_c,v_d,v_e \in \mathcal{X}$, $n \in \setN$ and $mid \in \mathit{Mal}$, 
\fi
\end{framedFigure}

We now define the condition necessary for splitting an ADL program
into an honest program semantic and a crypto-API semantics. 
\iffullversion\else%
For a formal
definition, we refer to the technical report~\cite{ourfullversion}.
\fi
An ADL program $\proto$ is \emph{pre-compliant} with
a crypto-API specification if:
\begin{inparaenum}[(i)]
    \item \iffullversion\else\label{itm:reads-only-slice}\fi
        the crypto-API only reads the slice of the heap belonging to its
        object, \ie, only objects reachable from locations stored in
        its own member variables,
    \item \iffullversion\else\label{itm:never-invokes-adv}\fi
        the crypto-API never invokes the adversary, and
    \item \iffullversion\else\label{itm:honest-deterministic}\fi
        only the crypto-API and the adversary make use of the
        $\textsf{rand}$-instruction.
\end{inparaenum}

\iffullversion
\begin{definition}[Reachable locations]
    The set of reachable locations from location $l\in\locationsDV$ in
    heap $h\in\heapsDV$ is defined recursively:
    \[
    \lreachabe(l) = \begin{cases}
        l & \text{if $h(l)\in\arraysDV$} \\
        \emptyset & \text{if $h(l)$ undefined}\\
        \bigcup_{f \mid F(f)\in\locationsDV} \lreachabe(l)
        & \text{if $h(l)=(c,F)\in\objectsDV$}
    \end{cases}
    \]
\end{definition}

\begin{definition}[Heap slice]
    Let 
    $h\in \heapsDV =  
    \locationsDV \rightharpoonup (\mathcal{O} \cup \mathcal{AR})$ 
    and $l\in \locationsDV$. We define $l$'s slice of $h$ as
    a partial function on $\locationsDV$ as follows:
    \[
        \sliceof[h]{ l } (l')
        = 
            h(l')  \quad \text{if $l'\in \lreachabe(l)$}
    \]
\end{definition}

\begin{definition}[ADL pre-compliance] \label{def:pre-compliant}
    An ADL program $\proto$ is pre-compliant with
    a library specification \libSpec, if for each $\secpar$ for the Crypto-API semantics in $\ADLss_{\proto^\secpar,\adv^\secpar,s_0}$ (see Definition~\ref{def:adl-split-state-representation}) the following holds:
    \begin{itemize}
        % \item it the library preserves the heap, \ie, for all
        %     $\mathit{mid}\in\midDV$, 
        %     $h\in\heapsDV$,
        %     $v_k,\ldots,v_{k+n-1}\in\valuesDV$,
        %     $\advs\in\advss$,
        %     $h_n\in\heapsDV$,
        %     $u\in\valuesDV$,
        %     such that
        %     $(\mathit{mid},h(v_k))\in\dom(\libSpec)$
        %     and
        %     $m = \lookupDirectDV_\proto(\mathit{mid},h(r(v_k)).\mathsf{class})$,
        %     \begin{align*}
        % \Pr[
        %     \stateDV[(m),h,(0),(\defaultRegistersDV([r_k,\ldots,r_{k+n-1}])),\advs] 
        %     \xrightarrow{\alpha_1}
        %      \cdots
        %     \xrightarrow{\alpha_n}
        %     \stateDV[u,h_n,\advs_n] 
        %     \wedge
        %     h_n \neq h 
        % ] = 0
        %     \end{align*}
        \item \label{itm:reads-only-slice}
            the Crypto-API semantics only reads the slice of the heap belonging to its
            object, \NOTEE{For the soundness proof, we do not need the requirement that the library solely reads a data slide from the heap. We might be able to just as well give the entire heap as an argument. For extracting a model that is used for verification, we can later on describe several methods for optimizing (i.e., reducing) the model.}
			\ie,
            for all
            $\mathit{mid}\in\midDV$, 
            $h_1,h_2\in\heapsDV$,
            $v_k,\ldots,v_{k+n-1}\in\valuesDV$,
            $\advs\in\advss$,
            $u\in\valuesDV$,
            such that
            $(\mathit{mid},h(r(v_k)))\in\libSpec$
            and
            $m = \lookupDirectDV_\proto(\mathit{mid},h(r(v_k)).\mathsf{class})$,
            \begin{align*}
                \Pr[
                    \stateDV[{(m),h_1[\sliceof{v_k}],(0),(\defaultRegistersDV([r_k,\ldots,r_{k+n-1}])),\advs}] 
                    \xrightarrow{\alpha_1}
                    \cdots
                    \xrightarrow{\alpha_n}
                    \stateDV[{u,h_1[\sliceof{v_k}],\advs_n}] 
                ] = \\
                \Pr[
                    \stateDV[{ (m),h_2[\sliceof{v_k}],(0),(\defaultRegistersDV([r_k,\ldots,r_{k+n-1}])),\advs }] 
                    \xrightarrow{\alpha_1}
                    \cdots
                    \xrightarrow{\alpha_n}
                    \stateDV[{u,h_2[\sliceof{v_k}],\advs_n}] 
            ]. 
            \end{align*}
        \item \label{itm:never-invokes-adv}
            the Crypto-API semantics never invokes the adversary, \ie,
            for all
            $\mathit{mid}\in\midDV$, 
            $h_1,h_2\in\heapsDV$,
            $v_k,\ldots,v_{k+n-1}\in\valuesDV$,
            $v_a,\ldots,v_e\in\valuesDV$,
            $\advs\in\advss$,
            $u\in\valuesDV$,
            such that
            $(\mathit{mid},h(r(v_k)))\in\libSpec$
            and
            $m = \lookupDirectDV_\proto(\mathit{mid},h(r(v_k)).\mathsf{class})$,
            \begin{align*}
                \Pr[ &
                    \stateDV[(m),h,(0),(\defaultRegistersDV([r_k,\ldots,r_{k+n-1}])),\advs] 
                    \rightarrow^*
                    \stateDV[m'\cdot ml,h,pp\cdot ppl,rl,\advs'] \\
                    & \wedge
m[pp] = \invokestatic~ v_a,\dots,v_e,\mathit{mid} 
\wedge \mathit{mid} \in \Mal_\mathit{static}
                ] = 0\text{, and}
            \end{align*}
        \item\label{itm:honest-deterministic}
        only the Crypto-API semantics and the adversary semantics in $\ADLss_{\proto^\secpar,\adv^\secpar,s_0}$ make use of the
        $\progcmd{rand}$-instruction, \ie,
            for any transition 
            $s_0 \rightarrow_{\Pi,n_1,p_1} s_1 \cdots
            \rightarrow_{\Pi,n_l,p_l} s_l$ to some 
            state $s_l=\stateDV$ with $m[pp]=\progcmd{rand} v_a$ for
            some $v_a$, 
            there is a position in the method stack, i.e.,
            an $i\in\setN$ 
            and $m_i\cdots ml'$, $r_i\cdot rl'$ and
            $pp_i\cdot rl'$ the $i$th suffix of $ml$, $rl$ or $ppl$,
            respectively, 
            such that
            $m_i[pp_i]= \invokedirectr~v_k,n,\mathit{mid}$
            and
            $(\mathit{mid},h(r(v_k)))\in\dom(\libSpec)$,
            for some $v_k$,$n$ and $\mathit{mid}$, and
            $h$ the heap when this method was called, \ie, $h$
            such that, for
            the
            largest $j\in\set{0,l}$ and some $\advs'\in\advss$,
            $s_j = \stateDV[
            m_i\cdots ml', h, r_i\cdot rl',
            pp_i\cdot rl', \advs']$.
    \end{itemize}
\end{definition}
\fi 

\begin{definition}[ADL split-state representation]%
\label{def:adl-split-state-representation} %{{{ 
    Given a uniform family of ADL programs 
    $\proto= \set{\proto^\secpar}_{\secpar\in\setN}$,
    every member thereof
    pre-compliant with a crypto-API
    specification \libSpec,
    a family of attackers
    $\set{\adv^\secpar}_{\secpar\in\setN}\in\advM$,
    and an initial configuration
    $s_0=\stateDVs$,
    we define the
    ADL split-state representation 
    $(\ADLss_{\proto^\secpar,\adv^\secpar,s_0})_{\secpar\in\setN}$
    as the family of
    split-state compositions of the following three transition
    systems for every $\secpar\in\setN$:
    \begin{itemize}
        \item \emph{honest-program semantics:}
            The transition system
            $(\states_H,s^0_H,A_H,\delta_H)$ is defined
            by the ADL semantics%
            \iffullversion
            (not including the transitions 
            in
            Figure~\ref{fig:adl-prob-adv})
            \fi,
            extended with rules in
            Figure~\ref{fig:adl-to-adlss}, and the initial state
            $\stateDV[{ml,h,ppl,rl,\advsd^H}]$ 
            for an
            arbitrary initial adversary state $\advsd^H$ which will be
            ignored.

        \item \emph{crypto-API semantics:}
            The transition system
            $(\states_L,s^0_L,A_L,\delta_L)$ is defined
            by the ADL semantics except
            \rReturnVoidF,  \rReturnF,
            extended with rules 
            in Figure~\ref{fig:adl-prob-adv} and
            $\Prob$ (see Figure~\ref{fig:adl-to-adlss-lib}).
            The initial state
            is $\stateDV[(),\emptyset,(),(),\advsd^L]$ for some
            arbitrary initial adversary state $\advsd^L$ which will be
            ignored.

        \item \emph{attacker semantics:}
            The transition system
            $(\states_A,s^0_A,A_A,\delta_A)$
            consists of the transitions of the adversary \adv,
            extended with the transitions in
            Figure~\ref{fig:adl-to-adlss-att}.
    \end{itemize}
    To make our notation more concise, we abbreviate the initial state
    $(\top,\stateDV[{ml,h,ppl,rl,\advsd^H}]),(\bot,s_A),
    (\bot,\stateDV[ (),\emptyset,(),(),\advsd^L ])$
    for some honest initial state $\stateDVs$ as
    % and $\secpar\in\setN$, $z\in\bits*$ as
    $\stateDVs^\mathit{ss}$.
% definition adl-split-state-representation (end)
\end{definition}%}}}

\iffullversion

Observe that the adversarial steps are explicit in this
representation, and hence each step is annotated with a computation
time of $1$. 

\begin{definition}[Single-step transition system]\label{def:single-step}
    A probabilistic transition system is \emph{single-step} if
    $s \predRel{n}{[p]} s', p>0$ or 
    $s \predRel{n}{a} s'$  imply $n=1$.
\end{definition}

For probabilistic transition steps that are single-step, such as any
system that is part of an ADL split-state representation,
we will omit
the number of computation steps in each transition from here forth.

While the overall split-state representation is probabilistic, the
honest program semantics is now deterministic, except for transitions
which call the Crypto-API or the adversary. This property is helpful
in formulating the embedding, as this means that from any given state, any 
follow-up state up to the point where the Crypto-API or the adversary
is called are uniquely defined.

\begin{definition}[Internally deterministic]\label{def:internally-deterministic}
    A single-step probabilistic transition system
    $T = (S,s_0,A,\delta)$ 
    is \emph{internally deterministic} if 
    every probabilistic state is deterministic,
    and if every non-deterministic state $s$ is either
    deterministic,
    or if for $A'=\tiAH$ and $A'=\tiLH$, there are
    $\delta_{A'} \colon S\times A' \to S$ such that
    \[
        \delta_H(s) \subseteq \set{(a,s',1) \mid a\in A' \land s'= \delta_{A'}(s,a) }.
    \]
    % its state space
    % $\states_H$ can be partitioned into three sets 
    % $\states_{H,1}$,
    % $\states_{H,2}$,
    % $\states_{H,3}$,
    % \TODOR{Find better names.}
    % such that
    % \begin{itemize}
    %     \item $s\in \states_{H,1} \Leftrightarrow \forall s'.
    %         s \xrightarrow{\alpha,n,p} s' \implies p=1 
    %         \land \alpha \not\in \tAH\cup\tLH\cup\tHA\cup\tHL
    %         \land \delta_\mathit{int}(s)=s'$
    %     \item $s\in \states_{H,2} \Leftrightarrow \forall s'. 
    %          s \xrightarrow{\alpha,n,p} s' 
    %         \implies 
    %            p=1 \land 
    %            ( \alpha\in\tHA
    %            \land s'=s )
    %             \lor 
    %             (\alpha\in\tAH \land  \exists m.~s'=\delta_\tiAH(s,m))$
    %     \item $s\in \states_{H,3} \Leftrightarrow \forall s'. 
    %          s \xrightarrow{\alpha,n,p} s' 
    %         \implies 
    %            p=1 \land 
    %            ( \alpha\in\tHL
    %            \land s'=s )
    %             \lor 
    %             ( \alpha\in\tLH \land  \exists m.~s'=\delta_\tiLH(s,m))$
    % \end{itemize}
    % for some functions
    % $\delta_\mathit{int}: \states_{H,1}\to \states$,  
    % $\delta_\tiAH: \states_{H,3} \times \calD \to \states$,  
    % and
    % $\delta_\tiLH: \states_{H,4} \times \calD \to \states$.
\end{definition}
\begin{lemma}
    The honest-program semantics within any ADL split-state
    representation is internally deterministic.
\end{lemma}
\begin{proof}
    By definition, the honest program semantics exclude probabilistic
    transitions. 
    By careful inspection, we can verify that 
    all states except those covered by the side-conditions 
    in Figure~\ref{fig:adl-to-adlss} are deterministic ( only
    one inference rule applies to each state, and each rule determines
    a unique follow-up state).
    The rules in Figure~\ref{fig:adl-to-adlss} are either
    deterministic, or in case of \LibResponse and \ReceiveMsg,
    determine the follow-up state depending on the action. 
\end{proof}
\fi

\begin{lemma}\label{lem:adl-exec-equals-adlss-exec}
    For all ADL programs $\proto$ that are
    pre-compliant with a crypto-API specification
    \libSpec,
    all initial configurations $\stateDVs$,
    all adversaries
   all $n\in\setN$, we have
   \[
    \Pr[
        \interact{ \Pi\stateDVs}
        {\adv}
    \downarrow_{n}x]
     = \iffullversion\else\\\fi
     \Pr[ T \downarrow_{n}x],
 \]
 where $T$ denotes the ADL split-state representation
    of $\proto$ and $\adv$
    for initial configuration
    $\stateDVs$.
\end{lemma}
\iffullversion
%!TEX root = main.tex
\begin{proof}
    We need to show the equivalence between
    \[
    \sum_{\parbox{3cm}{\begin{center}
    $(s_0 \xrightarrow{\alpha_1}_{n_1}$\\$\cdots\xrightarrow{\alpha_k} s_k)$\\
    $\in\supp(\Exec(\ADL_{\proto,\adv,\stateDVs}))$\end{center}}
}
    \Pr \left[\Exec(\ADL_{\proto,\adv,\hat{s_0}}) = s_0 \xrightarrow{\alpha_1}_{n_1} s_1 \cdots
    \xrightarrow{\alpha_{k-1}}_{n_{k-1}} s_{k-1} \xrightarrow{\alpha_k}_{n_k}
    \stateDV[x] \land
    \sum_{i=1}^{k} n_i \le n
    \right]
\]
where $\hat{s_0}$ is the configuration belonging to $s_0$, i.e., $s_0$ without the attacker state, and 
    \[
    \sum_{ (s_0 \xrightarrow{\alpha_1} \cdots s_k)\in\supp(\Exec(T)) }
    \Pr \left[\Exec(T) = s_0 \xrightarrow{\alpha_1}_{n_1} s_1 \cdots
        \xrightarrow{\alpha_k = (\tifinal,x)}_{n_k}
    s_k
    \land  
    % $\sum_{i=1}^{i\le k} n_i \le n$}
    k \le n
    \right].
\]
\newcommand{\sysa}{\ensuremath{\ADL_{\Pi^\secpar}}\xspace}
\newcommand{\sysb}{\ensuremath{T^\secpar}\xspace}

    We will prove this equivalence by induction on the length of the
    execution in \sysa, however, in order to define what it means for
    two executions to correspond, we
    partition execution steps into
    honest program states,
    library states,
    and final states.
    The initial configuration is an honest program state by definition.
    Final states stay final states. 
    Each other state
    is an honest state, if the previous state was an honest state and the
    transition between the two was not an instance of $\rIDR$ where
    $(\mathit{mid},h(r(v_k)))\in \libSpec$. If the
    previous state was an honest state, but this condition is not
    true, the successor state is a library state. 
    Similarly, when the
    previous state is a library state and the transition is an instance of 
    $\rReturnVoid$ or $\rReturn$
    such that the state of the
    successor state matched the stack of the latest previous state that
    was not a library state, then the successor is an honest program
    state, more formally: if $s_i \rightarrow s_{i+1}$ is an instance
    of $\rReturnVoid$ or $\rReturn$, the last previous honest program
    state $s_j$ with $j<i$ is such that 
    $s_j = \stateDV[ml,h,pp\cdot ppl,rl,\advs]$, then $s_{i+1}$ is a program
    state if $s_{i+1}=\stateDV[ml,h',pp+1\cdot ppl,rl',\advs']$ for
    some $h'$, $rl'$ and $\advs'$.

    Given this transition, we define a relation $\calR$ between executions,
    which is, for the most part, based on the resulting state. 
    \begin{itemize}
        \item A state $s$ in \sysa, which is a final state, 
            \ie, $s
            = \stateDV[\mathit{res}]$ for some $\mathit{res}$,
            corresponds to a state 
            $s_T=(s_H,s_A,s_L)$ in \sysb,
            if
            the transition to $s_T$ in $\sysb$ was labelled
            $(\tifinal,\mathit{res})$.
        \item A state $s$ in \sysa, which is an honest program state
            of form $s=\stateDV[u,h,\advs]$, corresponds to a state
            $s_T=(s_H,s_A,s_L)$ in \sysb,
            if
            $s_H=(\top,\stateDV[u,h])$ and
            $s_A=(\bot,\stateDV[\epsilon,\advs,\epsilon])$.
        \item A state $s$ in \sysa, which is an honest program state
            $\stateDV[ml,h,ppl,rl,\advs]$
            corresponds to state 
            $(s_H,s_A,s_L)$
            iff.
            $s_H=(\top,\stateDV[ml_H,h_H,ppl_H,rl_H,\advs_H])$,
                $s_L=(\bot,\stateDV[ml_L,h_L,ppl_L,rl_L,\advs_L])$,
                $s_A=(\bot,(\mathit{in},\advs,\mathit{out}))$,
                and
                \begin{itemize}
                    \item $ml=ml_H$,
                    \item $h=h_H$,
                    \item $ppl=ppl_H$,
                    \item $rl=rl_H$.
                \end{itemize}

            \item A state $s$ in \sysa, which is a library state
                $\stateDV[ml,h,ppl,rl,\advs]$
            corresponds to state 
            $(s_H,s_A,s_L)$
            iff.
                $s_H=(\bot,\stateDV[ml_H,h_H,ppl_H,rl_H,\advs_H])$,
                $s_L=(\top,\stateDV[ml_L,h_L,ppl_L,rl_L,\advs_L])$,
                $s_A=(\bot,(\mathit{in},\advs,\mathit{out}))$,
                and
                \begin{itemize}
                    \item $ml=ml_L\cdot ml_H$, where  $ml_L\neq
                        \epsilon$,
                    \item $h=h_H[h_L]$,
                    \item $ppl=ppl_L\cdot ppl_H$,
                    \item $rl=rl_L\cdot rl_H$.
                    \item $ml_H=ml_\mathit{enter}$, $rl_H=rl_\mathit{enter}$
                        and $ppl=ppl_\mathit{enter}$ for
                        $s_i=\stateDV[ml_\mathit{enter},h_i,ppl_\mathit{enter},rl_\mathit{enter},\advs_i]$,
                        $i<k$ the last program state such that $s_{i+1}$ and all
                        subsequent states were library states
                \end{itemize}
    \end{itemize}

 Fix $n$.
 We show by induction on number $k$ of steps in execution:
 For any execution of
 $k' = \sum_{i=0}^k n_i$
 steps that ends in a state 
 $\stateDV[ml,h,ppl,rl,\advs]$ in \sysa,
 the sum of probabilities of all executions of $k'$ steps in $\sysb$
 that ends in a corresponding state is exactly the same.
 Note that for $a\neq a'$,  $a \calR b$ implies $\neg(a' \calR b)$,
 hence this implies the claim.
 
Initially, this is the case, because if 
$\stateDV[ml,h,ppl,rl,(\secpar,z)]$ in \sysa,
then $s_H=(\top,\stateDV[ml,h,ppl,rl])$,
$s_L=(\bot,(\epsilon,\emptyset,\epsilon,\epsilon,\advsd))$ and
$s_A=(\bot,(\secpar,z))$.
(W.l.o.g., we use $\advsd$ for $\advsd^L$ or $\advsd^H$, depending on
context.)

Assume the induction hypothesis (IH) holds for $k$ steps, \ie, up to
a state $s_k$ in \sysa. We perform a case distinction on $s_k$.

\paragraph{If $s_k$ is a library state,} it has the form
$s_k=\stateDV[ml,h,ppl,rl,\advs]$, as the rules $\rReturnVoidF$ and
$\rReturnF$ are not part of the library semantics.
% which need to be
% active (i.e., $s_L=(\top,\ldots)$) as $s_k$ is a library state and the
% IH holds. 
%
Thus $ml=ml'\cdot ml_\mathit{enter}$ for some non-empty sequence
$ml'$, as 
$\rReturnVoid$, 
$\rReturnVoidF$, 
$\rReturn$ and
$\rReturnF$ 
are the only rules reducing the method stack, but if one of them 
would have been applied, $s_k$ would not be a library state.
By IH, we have that $ml=ml_L\cdot ml_H$, and $ml_H=ml_\mathit{enter}$
so any rule that does not alter the method stack
preserves the correspondence,
given that $\Pi^\secpar$ is pre-compliant (see 
Condition~\ref{itm:reads-only-slice}). This holds for
method calls to honest libraries, too, as they 
extend $ml$ and, likewise, $ml_L$ to the left.
Method call to the adversary are excluded by
Condition~\ref{itm:never-invokes-adv}.
As $s_k$'s frame is of length at least $2$,
$\rReturnVoidF$
$\rReturnF$ never apply.
Consider now the two remaining cases, where either
$\rReturnVoid$
or
\rReturn are applied. We treat only the
latter, as the former is completely analogous. 
Let
$s_k=\stateDV[m\cdot ml',h,pp \cdot
ppl,r\cdot r'\cdot rl,\advs]$ for $ml'\neq \epsilon$ and
$m[pp] = \progcmd{return} v_a$.
As $m\cdot ml' = ml_L\cdot ml_\mathit{enter}$,
either $ml_L=(m)$ and thus $ml'=ml_\mathit{enter}$,
or $|ml_L|>1$.
In the first
case, $ml_L = (m)$, 
only \LLibRetVoid or \LLibRet
(depending in $m[pp]$) apply. 
These are labelled 
$(\tiLH,\lowerDV(r(v_a)),\upperDV(r(v_a)),h)$
with probability $1$ and $1$ computation step.
By asynchronous parallel composition, the
the honest program semantics 
moves with rule 
\LibResponse. The following state is thus
an honest program state, and relation holds.
In particular, the library's heap overwrites
the honest program's, thus it holds that
$h'=h$ (by definition of the transition),
$h=h_H[h_L]$ (by IH)
$h_H[h_L]=h_H'$ (by \LibResponse)
and thus $h'= h_H'$.
In the second case, $|ml_L|>1$,
\rReturnVoid (or \rReturn) are used in both
cases, preserving the IH\@.
It is not possible to enter a final state from $s_k$,
as 
\rReturnVoidF and
\rReturnF only apply if the method stack is of size 1.

\paragraph{If $s_k$ is an honest program state,} it has the form
$s_k=\stateDV[ml,h,ppl,rl,\advs]$, or
$s_k=\stateDV[u,h,\advs]$.
In the second case,
$s_k=\stateDV[u,h,\advs]$,
the only applicable rule is
$\AdvFin$,
resulting 
in a transition labelled
$(\tifinal,\mathit{res})$.
to state 
$\stateDV[\mathit{res}]$, for the
sum of all probabilities of adversary steps that go
from $\adv$ to $\mathit{res}$ in $n$ steps, for some $n$.
Let $p$ be this sum.
The only applicable rule in $\sysb$ is
\FinalCall, which 
has label $(\tHA,\tifinalCall)$ and thus
executes
\ALeakMsg 
in parallel.
There is a set of
corresponding executions of the adversary for each
$n_k$,  the sum of which has probability $p$, as the
number of steps, starting states and final step are
the same. Including the final step $\AFinal$ (which
is always applicable when the third element of the triple is in
$\advrespDV$, and excludes \AReceiveMsg),
each of these results in a transition 
$(\tifinal,\mathit{res})$.
for $n+2$ steps.
Now we consider the first case,
$s_k=\stateDV[ml,h,ppl,rl,\advs]$.
Let the transition to $s_{k+1}$ be an instance of
\rIDR where
$(\mathit{mid},h(r(v_k)))\in\libSpec$.
Then, the next state is a library state
$sk' = \stateDV[m'\cdot ml,h,0\cdot ppl, \defaultRegistersDV ( [r(v_k),\dots,r(v_{k+n-1})] )\cdot rl,as]$
For each corresponding state in \sysb,
\LibCall in parallel with \LLibCall is the only transition possible in
\sysb. By asynchronous parallel composition, each follow-up state will
be such that 
$s_H=(\bot,ml,h,ppl,rl,\advs')$,
$s_L=(\top,(m),\sliceof{r(v_k)},(0),
\defaultRegistersDV([r(v_k),\dots,r(v_{k+n-1})]),\advs'')$, and
$s_A=(\bot,\mathit{in},\advs,\mathit{out}$
for some $\advs',\advs''\in\advrespDV$,
$\mathit{in} \in (F_\mathit{mal}, \valuesDV^* )$, and
$\mathit{out} \in \times (\valuesDV \cup \advrespDV \cup \set{\epsilon})$.
These library states correspond to each other, in particular,
$h=h_H=h_H[h_L]=h_H[\sliceof{r(v_a)}]$, and 
$ml = ml_\mathit{enter}$,
$ppl = ppl_\mathit{enter}$,
$rl = rl_\mathit{enter}$.
Let the transition be an instance of the rule \AdvInv with 
probability $p$ in $n$ steps. Similar to the previous case involving
the attacker, the sum of all possible transitions in $n+1$ steps from
$(s_H,s_A,s_L)$,
equals $p$,
as only rules \LeakMsg and \ALeakMsg apply (in parallel).
The adversary is activated on same input, \ie, all
instances of states $(i,\advs,\mathit{res}$
with $\mathit{res}\in\valuesDV$, reachable in
$T$ synchronise with \ReceiveMsg by definition
of \AReceiveMsg and
Definition~\ref{def:attacker}. 
The number of computation steps
account for the extra activation with
$\ReceiveMsg$ in parallel with $\ReceiveMsg$.
Let the transition be an instance of the rule \AdvInv with 
probability $p$ in $n$ steps. 
Similar to the
previous case, where \AdvInv was invoked from a library state,
this transition is labelled
$(\tifinal,\mathit{res})$.
and results in a 
state 
$\stateDV[\mathit{res}]$, for the
sum of all probabilities of adversary steps that go
from $\adv$ to $\mathit{res}$ in $n$ steps, for some $n$. 
The argument remains the same. Note that
the
adversary produces a final state
(which cannot have successor, by
Definition~\ref{def:attacker}), 
thus
there is no final state before any final state. 
Hence 
$\LeakMsg$ and $\ALeakMsg$ synchronise as before,
and $\AFinal$ is the only way to proceed from
final state.
If the transition was an instance of \Prob{}, by
Condition~\ref{itm:honest-deterministic},
Definition~\ref{def:pre-compliant}, $s_k$ would not be an
honest-program state.
Any other transition affects only the left-most element of $ml$, $ppl$
and $rl$, thus the same transition can be applies in \sysb.
\paragraph{If $s_k$ is a final state,} 
there is no next step in $s_k$ and all corresponding states. This
concludes the proof.
\end{proof}

\fi

%!TEX root = main.tex

\subsection{Over-approximating ADL}\label{section:approx}
Recall that our plan is to instantiate $\calD_c$ with symbolic terms and even with positions in a CoSP tree. However, some operations are not defined on symbolic terms. As an example consider the XOR operation. It has been shown that each computationally sound symbolic representation of XOR has to work on symbolic bitstrings. Hence, for a typical symbolic representation of a ciphertext the XOR operation is not defined in the symbolic model. 

In order to get rid of any undefined operations after instantiating $\calD_c$ with symbolic terms or CoSP tree positions, we over-approximate each undefined operation by querying the attacker for the result.
To this end, we define an over-approximation of the ADL semantics
in terms of its split-state representation. This over-approximation tracks
values resulting from calls to the crypto-API. Whenever
a computation, e.g., a unary operation, a binary operation or
a test, is performed on these values, the over-approximation gives
the adversary more power: she can decide the values of these
computations. 
This is necessary, as these operations cannot be computed in
the symbolic model. Already performing this over-approximation on the ADL
split-state semantics simplifies the embedding.
Thanks to the split-state composition, we can define this
over-approximation canonically. 
\iffullversion\else%
Yet for space constraints, we
refer the reader to the technical report~\cite{ourfullversion}, and give here only the core of
it.
\fi%
%!TEX root = ../main.tex

\iftikzavailable
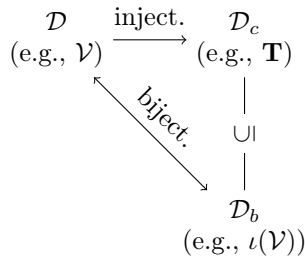
\begin{figure}[h] %{{{fig
\centering
	% \hspace{-3em}
	% \small
\begin{tikzpicture}[
    node distance=2.5cm,
    commnode/.style={align=center},
]
    \node[commnode]	(D)
    {$\calD$\\ (e.g., $\valuesDV$)} ;
    \node[commnode, right of = D]	(Dc) 
    {$\calD_c$\\ (e.g., $\T$)};
    \node[commnode, below of = Dc]	(Db) 
    {$\calD_b$\\ (e.g., $\iota(\valuesDV)$)};

    \draw[->] (D) to node[above] {inject.} (Dc);
    \draw[] (Dc) -- node[thin, fill=white] {\rotatebox[origin=c]{90}{$\subseteq$}} (Db);
    \draw[<->] (D) -- node[above,sloped] {biject.} (Db);
\end{tikzpicture}
\caption{Domains of the canonical \iffullversion over-approximation.
\else over-approx.\fi}
\label{fig:relation-between-domains-underlying-the-canonical-over-approximation}
% figure relation-between-domains-underlying-the-canonical-over-approximation (end)
\end{figure} %}}}

	% \draw[thick,double,->] ($(ADL-symb)+(2,0)$) -- node {Lemma~\ref{lemma:symbolic-adl-symbolic-execution}} ($(ADL-symbolic-execution)-(1.5,0)$);
	% \draw[thick,double,->] ($(ADL-symbolic-execution)+(1.7,0)$) -- node {Lemma~\ref{lemma:embedding-soundness-symbolic}} ($(ADL-symbolic-cosp)-(1.8,0)$);
        % \draw[thick,double,->] ($(ADL-symbolic-cosp)+(0,-0.8)$) -- node [left] {\parbox{\widthof{Soundness in CoSP}}{\centering Computational Soundness in CoSP}} ($(ADL-computational-cosp)+(0,0.8)$);
	% \draw[thick,double,->] ($(ADL-computational-cosp)-(1.8,0)$) -- node {Lemma~\ref{lemma:embedding-soundness-computational}} ($(ADL-computational-execution)+(1.7,0)$);
	% \draw[thick,double,->] ($(ADL-computational-execution)-(1.5,0)$) -- node {Lemma~\ref{lemma:computational-execution-adl}} ($(ADL)+(2,0)$);
	% \draw[dotted,double,->] ($(ADL-symb)+(0,-0.8)$) -- node {Theorem~\ref{theorem:symbolic-adl-computational-adl}} ($(ADL)+(0,0.8)$);
\fi % if tikzavailable

We assume some domain $\calD$ underlying the transition function and the set of states
of the honest program
semantic, and a distinct new domain $\calD_c$ (for values resulting from the
crypto API), a subset of which $\calD_b \subset \calD_c$ 
(representation of bitstrings) have
a bijection to the original domain (see
Figure~\ref{fig:relation-between-domains-underlying-the-canonical-over-approximation}).
We assume the honest-program semantics to be defined
via a set of
rules, and interpret these rules in the new domain $\calD_c$. If the
re-interpreted rule can be instantiated regardless of whether some
state carries values in $\calD$ or $\calD_c$, this rule is transferred
to the over-approximated semantics. Most rules only move values from
registers to heaps and are homomorphic in this sense.
Rules for which the above do not hold, but which can be expressed using
a constructor or destructor, are split into four rules, two of which
send these values to the adversary and use her input for the follow-up
state in case one of the variables is in $\calD_c\setminus \calD_b$.
Otherwise, \ie, if only values from $\calD$ are used, or values
representable in $\calD$, then the crypto-API is used for the
computation. This way, we were able to encode binary operations like
XOR as destructors, so representations of bitstrings that have passed
the crypto API (\eg, a bitstring was encrypted and then decrypted
again) can be treated without unnecessary imprecision.
If any rule $r\in R_H$ falls in neither of the above cases, the
canonical over-approximation is undefined. 

Using this over-approximation, we can derive $\ADLo$, the
over-approximated ADL semantics.
We tag messages resulting from calls to the crypto-API using a set
$\calN_c$, such that $\calN_c \cap \calN = \emptyset$, and assume
a bijection between the two that is efficiently computable.
For $n\in\calN_c \cup \calN$, let $[n]_{\calN_c}$ and
$[n]_{\calN}$ denote its representation in $\calN_c$ or $\calN$
according to the bijection. We lift this notation to values in
\calV, too. 
The inference rules include those previously defined,
but modified such that register values
are converted to $\calN$ before addressing the crypto-API,
\ie,
$ \tiHL(f,r({v_a}),\ldots,r({v_e}))$ is substituted by
$ \tiHL(f,[r({v_a})]_{\calN_c},\ldots,[r({v_e})]_{\calN_c})$,
and $\tiHA(\mathit{mid},{r(v_a)},\dots,{r(v_e)})$ is
substituted by
$\tiHA(\mathit{mid},{[r(v_a)]_{\calN_c}},\dots,{[r(v_e)]_{\calN_c}})$.
Finally, the rules in Figure~\ref{fig:adlss-to-adlo} are added.

%!TEX root = main.tex
\begin{framedFigure}{over-approximated honest program semantics for ADL ($\ADLo$). Here $\calN_b := range(\iota)$ and $\calN_c' := \calN_c \setminus \calN_b$.}{fig:adlss-to-adlo}
\codesize
\setlength{\columnsep}{-0.1cm}
\begin{multicols}{3}
\textbf{Cmd:} $m[pp] = \progcmd{unop} v_a,v_b,uop$\\[1ex]
$r(v_b)\not\in\calN_c'$\\[1ex]
$~\qquad \implies (l_1,l_2) = (\tiHL,\tiLH)$\\[1ex]
$r(v_b)\in\calN_c'$\\[1ex]
$~\qquad\implies (l_1,l_2) = (\tiHA,\tiAH)$
\begin{align*}
\textsc{Unop-$l_1$:} ~& s \xrightarrow{(l_1,\mathit{uop},r(v_b))}_H s
% && \text{for~} m[pp] = \progcmd{unop} v_a,v_b,uop ~\land  r(v_b)\not\in\calN_c
\\
\textsc{Unop-$l_2$:} ~& s \xrightarrow{(l_2,u)}_H s\subst{r[v_a \mapsto u]}
% && \text{for~} m[pp] = \progcmd{unop} v_a,v_b,uop ~\land r(v_b)\not\in\calN_c
% \\
% \textsc{Unop-Out:} ~& s \xrightarrow{(\tiHL,\mathit{uop},r(v_b))}_H s
% % && \text{for~} m[pp] = \progcmd{unop} v_a,v_b,uop ~\land r(v_b)\in\calN_c
% \\
% \textsc{Unop-In:} ~& s \xrightarrow{(\tiLH,u)}_H s\subst{r[v_a\mapsto u]}
% % && \text{for~} m[pp] = \progcmd{unop} v_a,v_b,uop ~\land r(v_b)\in\calN_c
\end{align*}

% For $m[pp] = \progcmd{unop} v_a,v_b,uop ~\land  r(v_b)\in\calN_c'$
% \begin{align*}
% \textsc{Unop-LC:} ~& s \xrightarrow{(\tiHA,\mathit{uop},d_\mathit{uop},r(v_b))}_H s
% % && \text{for~} m[pp] = \progcmd{unop} v_a,v_b,uop ~\land  r(v_b)\not\in\calN_c
% \\
% \textsc{Unop-LR:} ~& s \xrightarrow{(\tiAH,u)}_H s\subst{r[v_a \mapsto u]}
% % && \text{for~} m[pp] = \progcmd{unop} v_a,v_b,uop ~\land r(v_b)\not\in\calN_c
% \\
% \textsc{Unop-Out:} ~& s \xrightarrow{(\tiHA,\mathit{uop},d_\mathit{uop},r(v_b))}_H s
% % && \text{for~} m[pp] = \progcmd{unop} v_a,v_b,uop ~\land r(v_b)\in\calN_c
% \\
% \textsc{Unop-In:} ~& s \xrightarrow{(\tiAH,u)}_H s\subst{r[v_a\mapsto u]}
% % && \text{for~} m[pp] = \progcmd{unop} v_a,v_b,uop ~\land r(v_b)\in\calN_c
% \end{align*}
% \end{multicols}
\newpage

\textbf{Cmd:} $m[pp] = \progcmd{binop} v_a,v_b,v_c, bop$\\[1ex]
% \begin{align*}
$\lnot(r(v_b),r(v_c) \in \calN_c')$\\[1ex]
$~\qquad\implies (l_1,l_2) = (\tiHL,\tiLH)$\\[1ex]
$r(v_b),r(v_c) \in \calN_c'$ \\[1ex]
$~\qquad\implies (l_1,l_2) = (\tiHA,\tiAH)$
% \end{align*}
\begin{align*}
\textsc{Binop-$l_1$:} ~& s \xrightarrow{(l_1,\mathit{bop},r(v_b),r(v_c))}_H s
\\
\textsc{Binop-$l_2$:} ~& s \xrightarrow{(l_2,u)}_H s\subst{r[v_a \mapsto u]}
\end{align*}
\newpage
\textbf{Cmd:} $m[pp] = \progcmd{if\text{-}test} v_a,v_b,n,rop$\\[1ex]
$\lnot(r(v_a), r(v_b)\in\calN_c')$\\[1ex]
$~\qquad\implies (l_1,l_2) = (\tiHL,\tiLH)$\\[1ex]
$r(v_a), r(v_b)\in\calN_c'$\\[1ex]
$~\qquad\implies (l_1,l_2) = (\tiHA,\tiAH)$
\begin{align*}
\textsc{T1-$l_1$:} ~& s \xrightarrow{(l_1,\mathit{rop},r(v_a),r(v_b))}_H s
% && \text{for~} m[pp] = \progcmd{if\text{-}test} v_a,v_b,n,rop \land \lnot(r(v_a), r(v_b)\in\calN_c)
\\
\textsc{T2-$l_1$:} ~& s \xrightarrow{(l_2,x)}_H s\subst{pp + n}
\text{, if } x = r(v_a)
% && \text{for~} m[pp] = \progcmd{if\text{-}test} v_a,v_b,n,rop \land \lnot(r(v_a), r(v_b)\in\calN_c)
\\
\textsc{T3-$l_1$:} ~& s \xrightarrow{(l_2,x)}_H s\subst{pp + 1}
\text{, if } x = \bot
% && \text{for~} m[pp] = \progcmd{if\text{-}test} v_a,v_b,n,rop \land \lnot(r(v_a), r(v_b)\in\calN_c)
\end{align*}
\end{multicols}
\iffullversion
\begin{align*}
\LibCall': &
    s \xrightarrow{(\tiHL,\mathit{mid},[r(v_k),\dots,r(v_{k+n-1}),\sliceof{r(v_k)}]_\calD)}_H s
&& \text{for~}
m[pp] = \invokedirectr v_k, n, \mathit{mid}\\
&&& (mid,h(r(v_k))) \in \libSpec % o = \libSpec(\mathit{mid},h(r(v_k)).\mathfun{class})
\end{align*}
\fi
% \let\tmpcolumnsep\columnsep
% \setlength{\columnsep}{1.5cm}
% \setlength{\columnsep}{\tmpcolumnsep}
% \end{center}
\end{framedFigure}

%%%%%%%%%%%%%%%%%%%%%%%%%%%%%%%%%%%%%%%%%%%%%%%%%%%%%%%%%%%%%%%%%%%%
%%% Rest only in full version.                                   %%%
%%%%%%%%%%%%%%%%%%%%%%%%%%%%%%%%%%%%%%%%%%%%%%%%%%%%%%%%%%%%%%%%%%%%

\iffullversion% Rest only for fullversion
\paragraph{Precise treatment of bitstring representations output by
Crypto-API.}
If some value was produced by the Crypto-API, \eg, an encryption, in
\ADLs, it
will be represented by a symbolic term, \eg, $\enc(k,m)$ for some
terms $k$ and $m$. When $\enc(k,m)$ is tested for equality to
a bitstring, there is no meaningful way to decide for the outcome of
this test. To simplify the next proof steps, we over-approximate in
these situations, even though we have not made the transition into the
symbolic model yet.
But not every output from the Crypto-API should be over-approximated
in this. Consider the following example: a bitstring is encrypted by
the Crypto-API, stored and later decrypted again. The result of this
decryption is very well a library output, but it is perfectly feasible
to compute equivalence to a second bitstring. It should not be
necessary to send the bitstring to the adversary in this case.
In terms of the symbolic model,
the bitstring would be represented as outlined in
Example~\ref{ex:bitstring-representation} on page
\pageref{ex:bitstring-representation}. Each bitstring has a unique term
representation, thus there exists a subset of \T which is bijective to
the set of bitstrings. In general, we assume a subset of $\calD_c$,
called $\calD_b$, which is injective to the original domain $\calD$.
If an operation, \eg, a binary operation like XOR, is to be applied to
two values in $\calD_b$, we require the library to compute this
operation. For all relevant operations in ADL, this is possible
w.l.o.g (see Appendix~\ref{app:adl}. We thus obtain a relatively
precise over-approximation.

\paragraph{Rules.}\label{par:rules}
We assume the transition relation to be defined in terms of a set of
inference rules employing (meta-language) variables parametric in some
domain $\calD$.
We say we interpret a rule in the domain $\calD'$, if we alter the
domain of each variable in $\calD$ to $\calD'$. Naturally, this might
lead to rules that cannot be satisfied any more, e.g., if
$\calD=\mathbb{R}$ and a 
premise requires a variable $v\in\mathbb{R}$ to equal $\pi$, the same
rules interpreted in the domain $\calD'=\setN$ are unsatisfiable.
However, 
most rules defining the \ADL semantics can be re-interpreted easily.
Many of them move
values from one register to another, or from heap to registers and are
thus oblivious of the actual type of data moved. If a transition from
a state $s$ to a state $s'$ is possible with, \eg, \textsc{rMove},
and there is a well-defined mapping from states in the domain
$\calD=\calV$ to states in any other domain $\calD'$, then an
interpretation of \textsc{rMove} in \calD' can be instantiated for $s$
and $s'$ mapped to $\calD'$.

Some ADL rules, \eg, \textsc{rBinop}, cannot be re-interpreted for an
arbitrary $\calD$. In this case, the operation $\underline{\mathit{bop}}$ is
only defined on values in $\valuesDV$, more precisely $\calN$. Hence
the predicate $x = r(v_a) \mathrel{\underline{\mathit{bop}}} r(v_b)$ is
unsatisfiable if $r(v_a),r(v_b)\notin \valuesDV$. As mentioned before,
we can deal with this, by
\begin{inparaenum}[\itshape a\upshape)]
\item transforming $r(v_a)$ and $r(v_b)$ to $\calD_b$ and sending them
    to the library, if they are in $\calD$ or $\calD_b$, or
\item sending them $r(v_a)$ and $r(v_b)$ if they are in $\calD_c
    \setminus \calD_b$, \ie, if it is not clear how to evaluate this
    operation.
\end{inparaenum}

We define this over-approximation for any semantics in split-state
form as follows. Here we benefit from the split-state composition, as
it allows for rewriting rules in a way that involves the attacker, or
the Crypto-API\@. The over-approximated ADL semantics is an instance of
the canonical over-approximation.

\begin{definition}[Crypto-API space]
    Let $\calD$ be any set. A pair of domains $(\calD_c,\calD_b)$ is a \emph{Crypto-API
    space}, if 
    \begin{itemize}
        \item $\calD_c\cap\calD = \emptyset$,
        \item there is an injection
    $f$ from $\calD$ to $\calD_c$, and
        \item $\calD_b\subset \calD_c$, and
        \item there is a bijection $b$ between $\calD$ and $\calD_b$.
    \end{itemize}

    In addition, we use the following notation for
    all $d\in\calD\cup\calD_b$:
    \begin{align*}
        [d]_{\calD} =
        \begin{cases}
            d & \text{if $d\in\calD$}\\
            b^{-1}(d) & \text{if $d\in\calD_b$}
        \end{cases}
       \qquad &
        [d]_{\calD_b} = 
        \begin{cases}
            b(d) & \text{if $d\in\calD$}\\
            d & \text{if $d\in\calD_b$}
        \end{cases}
       \qquad &
        [d]_{\calD_c} = 
        \begin{cases}
            f(d) & \text{if $d\in\calD$}\\
            d & \text{if $d\in\calD_c$}
        \end{cases}
    \end{align*}
    We lift this notation to finite sets, tuples and sequences.
\end{definition}

\begin{definition}[Canonical over-approximation (honest program
    sementics)]\label{def:sso-h}

    Let $(\calD_c,\calD_b)$ be a Crypto-API space for $\calD$ and
    $\calD' := \calD \cup \calD_b \cup \calD_c$. 
    Given a symbolic model $\Model=(\C,\N,\Terms,\D)$, and
    an internally deterministic transition system $T_H$ with
    $\rightarrow_H$ defined in terms of a finite set of
    inference rules $R_H$, which are parametric in a domain
    \calD,
    we define $\rightarrow_{H'}$ from the set of inference rules
    $R_H$  that defines $\rightarrow_H$:
    \begin{enumerate}[$(i)$]
        \item\label{item:domain-independence} 
            Each rule $r\in R_H$, for which any instance 
            with domain \calD that concludes 
            $s \xrightarrow{(l,1,1)}_H s'$ for some $s$ and $s'$,
            can be instantiated 
            with domain $\calD_b$ to conclude
            $[s]_{\calD_b} \xrightarrow{(l,1,1)}_{H'} [s']_{\calD_b}$,
            and
            with domain $\calD_c$ to conclude
            $[s]_{\calD_c} \xrightarrow{(l,1,1)}_{H'} [s']_{\calD_c}$,
            is part of the rules when
            interpreted with domain $\calD_c$ (which $\calD_b$ is
            a subset of).

        \item\label{item:externalize-computations} 
            Each rule $R\in R_H$, for which the above does not hold,
            but for which there is a symbolic operation $O\in \SO$ 
            such that it has the form 
            \begin{align*}
                R = 
                \inference
                {\phi(s) \wedge \rho(s,s') 
                    \wedge \Pr\left[ [r]_{\calD_b} = A_O([v_1,\ldots,v_l]_{\calD_b})
                \right]=1
                }
                {s \xrightarrow{(\varepsilon,1,1)}_H s'} 
            \end{align*}
            for
            variables $v_1,\ldots,v_l$ in $s$,
            and variables $r$ in $s'$,
            we define four rules:  
            \begin{align*}
                \inference{\phi([s]_\calD) \wedge \neg (v_1,\ldots,v_l\in\calD\cup\calD_b) }
                {s \xrightarrow{(( \tiHA,O,v_1,\ldots,v_l  ),1,1)}_{H'} s}
          \qquad     &
                \inference
                {\phi([s]_\calD)\wedge\rho'([s]_\calD,[s']_\calD) \wedge \neg (v_1,\ldots,v_l\in\calD\cup\calD_b)}
                {s \xrightarrow{(( \tiAH,r),1,1)}_{H'} \delta_{\tiAH}(s,r) = s'} 
                \\[1em]
                \inference
                {\phi([s]_\calD) \wedge v_1,\ldots,v_l\in\calD \cup \calD_b
                }
            {s \xrightarrow{((\tiHL,O,[v_1]_{\calD_b},\ldots,[v_l]_{\calD_b}),1,1)}_{H'} s} 
            \quad &
                \inference
                {\phi([s]_\calD) \wedge \rho([s]_\calD,[s']_\calD) \wedge v_1,\ldots,v_l\in\calD\cup\calD_b }
            {s \xrightarrow{((\tiLH,r),1,1)}_{H'} \delta_{\tiLH}(s,r)= s'} 
            \end{align*}
    \end{enumerate}
	\TODOE{I think our notation for $\delta_{\ell}$ is not consistent anymore. Should it, e.g., for $\delta_{\tiLH}$ not be something like $\delta_L(s)|_\tLH^\bot$, as in Definition~\ref{def:split-state-semantics}?}

    If any rule $r\in R_H$ falls in neither of the above cases, the
    canonical over-approximation is undefined. Otherwise,
    the \emph{canonical over-approximation of the honest program semantics $T_H$ for Crypto-API space 
    $(\calD_c,\calD_b)$}
    is the transition system
    \[ T_H' := ([\states_H]_{\calD'},[\initstates_H]_{\calD'},\rightarrow_{H'}). \]
\end{definition}

\begin{definition}[Canonical over-approximation]\label{def:sso}
    Given a symbolic model $\Model=(\C,\N,\Terms,\D)$, and
    a split-state composition $T$ of
    an internally deterministic transition system $T_H$,
    and transition systems 
    $T_A$ and $T_L$, such that
    $T$ is probabilistic,
    let $T_H'$,  the canonical over-approximation of the honest program
    sementics $T_H$, be defined. 

    Then, the \emph{canonical over-approximation of $T$ for the Crypto-API
    space $(\calD_c,\calD_b)$}
        is the split-state
    composition of $T_H'$, $T_A$, $T_L'$, where 
    \begin{align*}
     T_H' & := ([\states_H]_{\calD'},[\initstates_H]_{\calD'},\rightarrow_{H'})\\
     T_L' & := (\states_L,\initstates_L,\rightarrow_{L'})\\
    \end{align*}
    with $s\xrightarrow{(\tiHL,f,[v_1]_{\calD_b},\ldots,[v_l]_{\calD_b})}_{L'} s' $  
    if $s\xrightarrow{(\tiHL,f,v_1,\ldots,v_l)}_L s'$ and 
    $s \xrightarrow{\alpha}_{L'} s'$ if $s \xrightarrow{\alpha}_{L} s'$ and $\alpha\not\in\tHL$. 
\end{definition}

\begin{lemma}\label{lem:overapproximation-preserves-internally-det}
    For any split-state composition $T$ from an
    honest program
    $T_H$,
    and some attacker and library transition system,
    the honest
    program semantics
    $T_H'=(\states_H',\states^{0'}_H,\rightarrow_H')$
    resulting from the canonical over-approximation of $T$ 
    to the crypto-API space $(\calD_c,\calD_b)$
    are internally deterministic, if $T_H$ is internally
    deterministic.
\end{lemma}
\begin{proof}
    Consider the subset of the rules defining $\rightarrow_H'$ that
    are just re-interpreted according to
    Condition~\ref{item:domain-independence},
    Definition~\ref{def:sso}. If the transition relation relating from
    this subset was not internally deterministic, $T_H$ would not have
    been internally deterministic, as the mapping from $\calD$ to
    $\calD_c$ is injective.
    Consider now any additional rule derived according to
    Condition~\ref{item:externalize-computations},
    Definition~\ref{def:sso}. Either the first and second rule
    applies, or the third and fourth, as their premises are mutually
    exclusive on the state. By definition of these rules, the
    follow-up state $s'$ depends only on $s$ and value $r\in\calD'$ contained
    in the label. 
    It is left to show that no state $s$ to which any of these four
    rules can be instantiated, is covered by any rule resulting
    included according to Condition~\ref{item:domain-independence},
    Definition~\ref{def:sso}, or by another rule resulting fro
    Condition~\ref{item:externalize-computations}. In the both cases,
    this would imply that $T_H$ was not internally deterministic, in
    the first case directly, in the second case because the same
    predicate on $s$ is in the premise of all four translated rules.
\end{proof}

\begin{lemma}
    For any split-state composition $T$ from an
    honest program
    $T_H=(\states_H,\initstates_H,\rightarrow_H)$,
    an attacker $T_A$,
    and a
    library transition system 
    $T_L=(\states_L,\initstates_L,\rightarrow_L)$,
    if
    \begin{itemize}
        \item $\Pr[ s \xrightarrow{(\tiHA,O,u_1,\ldots,u_l)}_H s']=0$ and
            $\Pr[ s \xrightarrow{(\tiHL,O,u_1,\ldots,u_l )}_H s']=0$ for
            any two states $s,s'\in\states_H$ and any sequence
            $u_1,\ldots,u_l$, and
        \item $ \Pr [ s \xrightarrow{(\tiHL,O,v_1,\ldots,v_l)}_L
            \rightarrow^*\xrightarrow{(\tiLH,[A_O(v_1,\ldots,v_l)]_\calD)}_L s']=1$ for any
            $s'\in\states_L$ and any $s$ for which holds
            $s\in\initstates_L$ or $\Pr[s' \xrightarrow{\tLH} s]\neq
            0$. 
    \end{itemize}
    then there is an attacker $T_A'$ such that for the canonical
    over-approximation $T'$ from $T_H,T_L$ and $T_A'$, and
    for all $\secpar,n\in\setN$, $z\in\bits*$,
    \begin{multline*}
        \Pr[ T(\stateDVs,(\bot,(\secpar,z)),\stateDV[(),\emptyset,(),(),\advsd]) \downarrow_{n}x]
     = \\
     \Pr[ T'([\stateDVs]_{\calD_b},(\bot,(\secpar,z)),((),\emptyset,()(),\advsd)) \downarrow_{p(n)}x].
    \end{multline*}
for some polynomial $p$.
\end{lemma}
\begin{proof}
    Let $T_A'$ be $T_A$, with all transitions of form
    $(\tHA,O,\ldots)$ for $O\in\SO$ removed, and all transitions
    labelled
    $(\tHL,f,u_1,\ldots,u_l)$ for $f\not\in \SO$ 
    interpreted with $[u_1,\ldots,u_l]_{\calD}$ instead of
    $u_1,\ldots,u_l$ in the successor state. 
    We show that for any transition from step 
    $(s_H,s_A,s_L)$ to $(s_H',s_A',s_L'$ in $T$,
    there is a sequence of transitions with the same
    overall probability (\ie, the product of their probabilities)
    from $([s_H]_{\calD_b},s_A,s_L)$ to $([s']_{\calD_b},s_A',s_L')$ in T'. 
    Initially, this holds by definition of the initial state.
    Inductive step. Let the honest program be active.
    For any transition described by the
    transformation~\ref{item:domain-independence}, this holds
    immediately, unless the adversary or the library is called. If the
    adversary is called, observe that, by assumption one, 
    the adversary is never invoked on $(\tHA,O,\ldots)$ for $O\in\SO$,
    so $T_A$ differs only in the transformation, from $\calD_b$ to
    $\calD$, which, by bijection, produces the same follow-up state.
    If the library is called, the same argument is made by definition
    of the library semantics of the over-approximation from
    Definition~\ref{def:sso}.
    transformation~\ref{item:externalize-computations}, consider the
    instantiation of $R$ that applies from $s_H$ to $s_H'$. By
    definition of $[\cdot]_{\calD_b}$, the first two rules derived
    from $R$ never apply, but the third applies. 
    By assumption two and the definition of split-state composition,
    $T'$ proceeds with probability $1$ until the fourth rule applies.
    In particular, the same predicates hold, and, as in the previous
    case, by the fact that $\calD$ and $\calD_b$ are bijective, and by
    definition of the over-approximated library semantics in
    Definition~\ref{def:sso}, the library transitions operate on the
    same internal states. 
    The resulting state contains $r\in\calD$, \ie, the original
    domain, which is left intact by the conversation $[s']_{\calD}$.
    In the last case, the library or the attacker is active. If no
    message is passed to the honest program, the
    the induction hypothesis is trivially preserved. If a message is
    passed to the honest program, the honest program needs to have
    a transition with a label in \tAH or \tLH, and as $T_H$ needs to
    be internally deterministic for the over-approximation to be
    defined, this transition is an instance of a rule that matches
    Condition~\ref{item:domain-independence}. By the transformation
    applied to these rules, the values in $\calD$ are correctly
    translated to on reception $\calD_b$. 
\end{proof}

\begin{definition}[Harmonizing Crypto-API]\label{def:harmonizing-crypto-api}
    A Crypto-API $T_L := (\states_L,\initstates_L,\rightarrow_L)$ \emph{harmonizes} with 
    a computational implementation \implementation of a symbolic model \Model and a \libSpec,
    if    
    \[
        \Pr\left[
            \exists s'\in\initstates_L. 
            s \xrightarrow{((\tiHL,\mathit{op},u_1,\ldots,u_l),1,1)}_{L'} 
            \rightarrow^*
            \xrightarrow{((\tiLH,r),1,1)}_{L} s'
            \right] = \Pr\left[ A_{\libSpec(\mathit{op)}}(u_1,\ldots,u_l)=r \right]
            \]
        and
    \[
        \Pr\left[
            \exists s'\in\initstates_L. 
            s \xrightarrow{((\tiHL,f,u_1,\ldots,u_l,h),1,1)}_{L'} 
            \rightarrow^*
            \xrightarrow{((\tiLH,r),1,1)}_{L} s'
            \right] = \Pr\left[ A_{\libSpec(f,h(u_1))}(u_1,\ldots,u_l,h)=r \right]
            \]
 for any $s$ for which $s\in\initstates_L$ or $\Pr[s'
    \xrightarrow{\tLH} s]\neq 0$. 
\end{definition}

\begin{definition}[Crypto-API compliance]\label{def:compliant}
    An ADL program $\proto$ is compliant with a library specification \libSpec w.r.t. a symbolic model \Model
    and a computational implementation of \Model called \implementation,
    if $\proto$ is
    pre-compliant (see Definition~\ref{def:pre-compliant}),
    and there is a Crypto-API ~$T_L := (\states_L,\initstates_L,\rightarrow_L)$ such that $T_L$ harmonizes with $\implementation$ (see Definition~\ref{def:harmonizing-crypto-api}).
\end{definition}

\begin{example}[ADLo: over-approximated ADL semantics]%
\label{ex:adl-overapprox} %{{{ 
    We use the canonical over approximation to define $\ADLo$, the
    over-approximated version of the ADL split-state semantics. 
    Let $\calD = \valuesDV \cup \heapsDV$, 
    and
    $(\calD_c,\calD_b)$ any Crypto-API space for $\calD$. This can,
    for example, created by tagging library output, and thus having
    $\calD_b=\calD_c$.
    Given a uniform family of ADL programs 
    $\proto= \set{\proto^\secpar}_{\secpar\in\setN}$,
    every member of which
    is pre-compliant with a library 
    specification \libSpec,
    and a family of attackers
    $\adv = \set{\adv^\secpar}_{\secpar\in\setN}\in\advM$,
    we call the canonical over-approximation
    of the ADL split-state representation (see
    Definition~\ref{def:adl-split-state-representation})
    the over-approximated ADL semantics, denoted
    $\ADLo_{\proto,\adv,s_0}=
    (\ADLo_{\proto^\secpar,\adv^\secpar,s_0})_{\secpar\in\setN}$.
% definition adls-split-state-representation (end)
\end{example}%}}}

\begin{corollary}\label{cor:adlo}
    For any ADL program
    $\proto$ 
    compliant with a library specification \libSpec
    and a symbolic
    model \Model,
    and an attacker $\adv$,
    there is an
    attacker $\adv'$ and a polynomial $p$ such that
    for all $n$: 
    \[
     \Pr[ \ADLss_{\proto,\adv,s_0} \downarrow_{n}x]
     = 
     \Pr[ \ADLo_{\proto,\adv',s_0} \downarrow_{p(n)} x].
 \]
\end{corollary}

We define the the split-state equivalence for ADL.

\begin{definition}[ADL split-state equivalence]\label{def:adlss-ind}
Let $\proto_1$ and $\proto_2$ be two families of ADL programs compliant with the
same library specification \libSpec and $s_1,s_2$ initial configurations for $\proto_1$ and $\proto_2$, respectively. We write
$ \protos 1 \execIndistinguishable \protos 2$,
if, for all adversaries $\adv\in\advM$,
for the over-approximated ADL split-state representations 
$(\ADLo_{\proto_1^\secpar,\adv^\secpar,s_1})_{\secpar\in\setN}$ of $\protos 1$ and \adv,
and
$(\ADLo_{\proto_2^\secpar,\adv^\secpar,s_2})_{\secpar\in\setN}$ of $\protos 2$ and \adv, we have
$T_{H,1}(s_1) \SSIndistinguishable{T_L,T_A} T_{H,2}(s_2)$.
\end{definition}

Finally, we are in a position to connect the ADL semantics to the over-approximated split-state form.
\begin{lemma}\label{lemma:computational-execution-adl}
Let $\proto_1$ and $\proto_2$ be two families of ADL programs pre-compliant with the
same library specification \libSpec and $s_1,s_2$ initial configuration for $\proto_1$ and $\proto_2$, respectively. Then
$\protos 1 \execIndistinguishable \protos 2
\implies \protos 1 \ADLIndistinguishable \protos 2$.
\end{lemma}
\begin{proof}
    Follows from Lemma~\ref{lem:adl-exec-equals-adlss-exec} and Corollary~\ref{cor:adlo}.
\end{proof}

\fi

\iffullversion
%!TEX root = main.tex
\subsection{Canonical symbolic semantics}%
\label{sec:canonical-symbolic-model}

The canonical over-approximated split-state semantics define
a canonical symbolic semantics, when the Crypto-API space $\calD_c$ is
instantiated with the set of terms
defined by the CoSP-symbolic model $\Model$ (in
which our result in parametric),
when the 
attacker semantics is given by 
the symbolic attacker from $\Model$, 
and when the crypto-API is simplified to
the constructors and destructors in $\Model$.

\begin{definition}[Canonical symbolic split-state semantics]
    Given a symbolic model $\Model=(\C,\N,\Terms,\D)$,
    a subset $\T_b \subset T$ such that
    $(\T,\T_b)$ constitute a Crypto-API space, 
    a split-state semantics $T$ with an honest-program semantics
    $T_H$,
    and an attacker strategy, i.e., a sequence $I\in\attstrat$
    for 
            $\attstrat = 
            \set{(\texttt{in},O) \mid O\in\SO} \ifdraft \cup
            \set{(\texttt{ctl},l') \mid l'\in\bits*}\fi$
    (see Definition~\ref{def:symbolic-execution}), 
    the \emph{canonical symbolic split-state semantics
    of $T$ w.r.t. $(\T,\T_b)$ and $I$}
    is 
    the split-state
    composition of 
    \begin{itemize}
        \item the canonical over-approximation of the honest program
            semantics for Crypto-API space $(\T,\T_b)$ (see
            Definition~\ref{def:sso-h}),
        \item the attacker semantics 
            $(\states_I,\states_I^0,\rightarrow_I)$ with
            state space
            $\states_I \defeq \attstrat \times \Eventout^* \times ( \T \cup \set{\bot})$,
            initial state
            $\initstates_I =  (I,\epsilon,\bot)$
            ($\epsilon\in\Eventout^*$),
            and $\rightarrow_I$ the smallest relation such that
            \begin{align*}
                ((\texttt{in},O)\cdot \mathit{IL},V,\bot)
                \xrightarrow{(\tiHA,\mathit{mid},u_1,\ldots,u_l)}_{I,1,1}
                (\mathit{IL},V', \eval_O(\termsof(V')))
                && \text{for~}
                V'= V\cdot (\texttt{out}, \widehat{\mathit{mid}},\widehat{u_1},\ldots,\widehat{u_l})
                \\
                (\mathit{IL},V,\pair(u_\mathit{lo},u_\mathit{up}))
                \xrightarrow{(\tiAH,u_\mathit{lo},u_\mathit{up})}_{I,1,1}
                (\mathit{IL},V, \bot),
            \end{align*}
        \item and the Crypto-API semantics 
            $(\states_L,\states_L^0,\rightarrow_L)$ with
            state space
            $\states_L= \T^3 \cup \set{\bot,\epsilon}$,
            initial state space
            $\initstates_L=\epsilon$
            and $\rightarrow_L$ the smallest relation such
            that:
            \begin{align*}
                \epsilon 
                \xrightarrow{(\tiHL,\mathit{mid},u_1,\ldots,u_l,h)}_{L,1,1}
                (t_\mathit{lo},t_\mathit{up},t_h,)
                && \text{if~}
                \pair(t_\mathit{lo},\pair(t_\mathit{up},t_h)) = \eval_O (\widehat{\mathit{mid}},\widehat{u_1},
                \ldots,\widehat{u_l},\widehat{h})
                \\
                && \text{or $t_\mathit{lo}=t_\mathit{up}=t_h=\bot$ if~}
                \eval_O (\widehat{\mathit{mid}},\widehat{u_1},
                \ldots,\widehat{u_l},\widehat{h})=\bot
                \\
                (t_\mathit{lo},t_\mathit{up},t_h)
                \xrightarrow{(\tiLH,(t_\mathit{lo},t_\mathit{up},t_h))}_{L,1,1}
                \epsilon.
            \end{align*}
    \end{itemize}
\end{definition}

    We can define the ADL symbolic split-state semantics as an
    instance of the canonical symbolic semantics.

\begin{definition}[$\ADLsss$, ADL symbolic split-state semantics]
    Let $\T_b$ be the domain of $\iota$
    (see Definition~\ref{def:adl-embeddable-symbolic-model}) and
    observe that $(\T,\T_b)$ constitute a Crypto-API space. 
    For $\Pi$ a uniform family of ADL programs, every member of which
    is pre-compliant with a library 
    specification \libSpec,
    the \emph{ADL symbolic split-state semantics of $\Pi$ w.r.t.\ an
    attacker strategy $I$ }
    is the canonical symbolic split-state semantics
    of 
    the ADL split-state representation
    (see~\ref{def:adl-split-state-representation})
    w.r.t. $(\T,\T_b)$ and $I$.
\end{definition}

Note that, given a program $\Pi$ and an attacker strategy $I$, the ADL
symbolic split-state semantics of $\Pi$ w.r.t. $I$ is deterministic.

We show that the ADL symbolic split-state semantics can
be trivially simplified to the semantics presented in
Section~\ref{sec:dalvik-symbolic-semantics}.  

We define symbolic equivalence in the 
spirit of Definition~\ref{def:symbolic-indistinguishable}
with equivalence of views defined as in
Definition~\ref{def:static-equivalence}.

\begin{definition}[Symbolic equivalence $\execEquivalent$]
\label{def:ss-symbolic-indistinguishable}
Two ADL programs $\proto_1$ and $\proto_2$,
and initial configurations 
$s_1=\stateDVsN{_1}$ 
and
$s_2=\stateDVsN{_2}$ 
are 
\emph{symbolically split-state equivalent ($\protos 1 \execEquivalent \protos 2$)}
if for all attacker strategies $I$, their respective ADL symbolic
split-state semantics $T_1,T_2$ w.r.t. to $I$ are symbolically equivalent, i.e.,
 if  $\Views(T_1) \sim \Views(T_2)$.
\end{definition}

Now we are ready to state that symbolic equivalence in the
split-state setting implies symbolic equivalence in the sense of
Definition~\ref{def:symbolic-indistinguishable}, or in other words,
symbolic equivalence with respect to the semantics introduced in
Section~\ref{sec:dalvik-symbolic-semantics}.

\begin{lemma}\label{lemma:symbolic-adl-symbolic-execution}
Let $\proto_1$ and $\proto_2$ be two ADL program pre-compliant to the same library specification \libSpec and $s_1,s_2$ initial configuration for $\proto_1$ and $\proto_2$, respectively. Then,
$$\protos 1 \ADLEquivalent \protos 2 \implies \protos 1 \execEquivalent \protos 2$$
\end{lemma}
\iffullversion%
\begin{proof}
    Assume $\proto_1 \notExecEquivalent \proto_2$, then there 
    are initial configuration 
    $s_1\in T_{\proto_1}$,
    $s_2\in T_{\proto_2}$
    and
    views
    $t_1 \in \Views_{s_1}(T_{\proto_1})$,
    $t_2 \in \Views_{s_2}(T_{\proto_2})$
    such that $t_1 \not\sim t_2$ for
    $T_{\proto_1}$
    and
    $T_{\proto_2}$
    their respective ADL symbolic split-state semantics with respect
    to some attacker strategy $I$.
    We obtain two views witnessing that 
    $\protos 1 \notADLEquivalent \protos 2$
    from $t_1$ and $t_2$ as follows.
    If there is a consecutive pair of elements of
    form $((\tHA,\mathit{rop},x,y),(\tAH,x))$ 
    or
    $((\tHA,\mathit{rop},x,y),(\tAH,\bot))$ 
    with
    $\mathit{rop}\in\mathcal{ROP}$ and $x,y\in\T$
    in $t_1$ or $t_2$,
    then the second
    element is removed.

    Observe that for
    transition that invoke neither attacker nor
    library both semantics permit the same transitions.
    Transitions of the form
    $(\tiHA,\ldots),(\tAH,\ldots)$ or
    $(\tiHL,\ldots),(\tLH,\ldots)$ 
    correspond immediately, as the library and attacker semantics can
    be inlined, 
    the only exception being the rule \textsc{IfTest-m}.
    Wherever \textsc{IfTest-m} is invoked, an event of form
    $\mathit{rop}\in\mathcal{ROP}$ and $x,y\in\T$ is in the view,
    hence $V\vdash x$ as well as $V\vdash y$ for $V$ the view at
    this point. Thus, if this transformation (let us call it $\theta$) is
    applied,
    $\theta(t_1)\in \Views_{s_1}(\proto_1)$ 
    and 
    $\theta(t_2)\in \Views_{s_2}(\proto_2)$.
    If $t_1 \not\sim t_2$, then 
    $\theta(t_1) \not\sim \theta(t_2)$,
    as the transformation only removes input elements that are
    irrelevant for the knowledge set, and is preserving the structure
    in the following sense: if $\theta$ removes the $i$th element of $t_1$,
    it also removes the $i$th element of $t_2$, otherwise the $i-1$
    prefix of $t_1$ and $t_2$ was not symbolically equivalent.
\end{proof}
\else%
The proof of this lemma as well of all upcoming ones are postponed to the technical report~\cite{ourfullversion}.
\fi

In
Lemma~\ref{lemma:embedding-soundness-symbolic}, we show
that for the embedding presented in Section~\ref{sec:embedding}
static equivalence in CoSP implies static equivalence in the symbolic
variant of ADL\@.

%!TEX root = main.tex
\subsection{Constructing the CoSP-embedding}\label{section:construction-embedding}
%!TEX root = main.tex
\newcommand{\extractState}{\mathfun{exSt}}
\newcommand{\extractAllStates}{\mathfun{exAllSt}}
\newcommand{\ident}{\mathfun{id}}
Previous work~\cite{BaHoUn_09:cosp,BaMaUn_10:rcf} defined the
embedding into CoSP indirectly, namely via a symbolic and
computational execution that 
followed the structure of a CoSP execution (see
Definition~\ref{def:symbolic-execution}
and~\ref{def:computional-execution-challenger}).
This work, in contrast, explicitly defines an embedding.
As CoSP trees are infinite, we define the embedding in a co-recursive manner, i.e., as the largest fixpoint of a co-recursive construction. 
Each step in this recursion is defined by a function $\AlgoName$ that takes as input a trace from a leaf-node in the so-far constructed CoSP tree to the root node and outputs a finite subtree. 
After defining this largest fixpoint construction, we concentrate on defining the recursion $\AlgoName$.
We stress that our construction is defined on the honest program semantics of any over-approximated semantics. 
Hence, this construction is valid for the canonical symbolic model,
too.

\subsubsection{Instantiating the over-approximated split-state semantics with references}
Within the embedding, we instantiate the honest
program semantics in the over-approximated split-state form (see Definition~\ref{def:sso}),
but replace values originating from the cryptographic library or the
attacker by pointers to computation nodes or input nodes, respectively.
Formally, we instantiate
the set $\calD_c := \mathit{Pos}$ with the set of positions in a CoSP tree.
Here, positions are sequences of natural
numbers that encode which path through the CoSP-tree was taken, \ie, $\mathit{Pos} := \setN^*$.
In the case of ADL, registers and heap locations thus store values that have been input by the
adversary or the crypto-API by pointing to the position of
the respective input or computation nodes in addition to numerical
values, locations and \voidDV{}\@.
In order to transmit these values to the adversary or the crypto-API,
these positions are 
resolved to a node identifier. The CoSP execution itself takes care of
translating, \eg, node identifiers of input nodes to the value the
attacker choose to send at this point.

Note that there is no unique
representation in $\calN_c'$
for all values in $\calN$ (\ie, there is no
bijection between the sets
as required by Definition~\ref{ex:adl-overapprox}),
but (as we will see below) the transitions
in Figure~\ref{fig:adl-to-adlss-att}, which are the only ones using 
$\calN_c$ representations, are never be used. All other transitions
are still well-defined.
\TODOE{What does this comment mean? Does it mean that we can actually relax our definition?}

Recall that states Lemma~\ref{lem:overapproximation-preserves-internally-det} that the
ADL split-state representation is not probabilistic anymore, and
contains non-determinism only in global transitions, i.e., transitions with labels
in
$\tHA\cup\tAH\cup\tHL\cup\tLH$. 

\begin{figure}[t]
\begin{framed}
\begin{center}
\includegraphics[width=\textwidth]{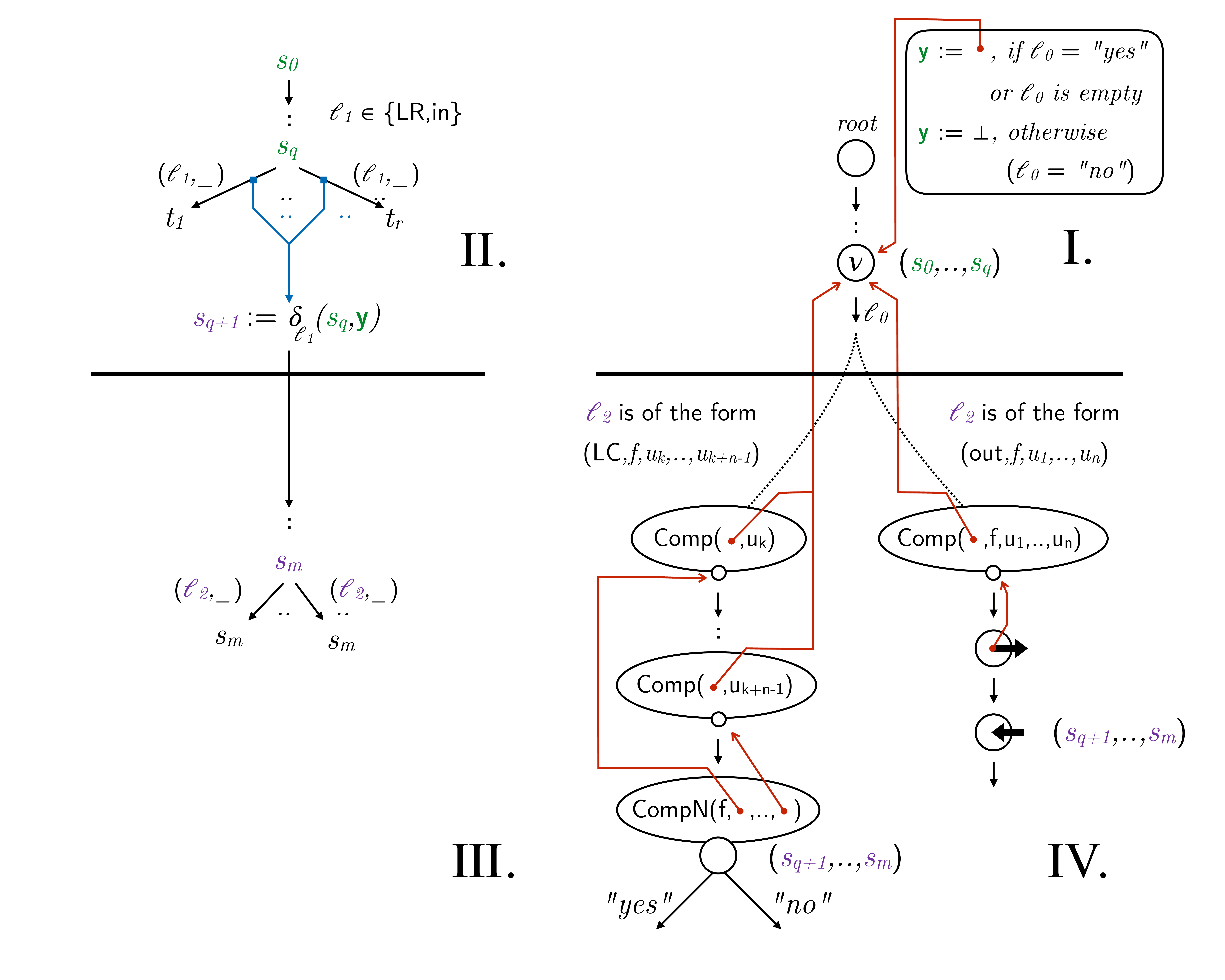}
\end{center}
\vspace{-2em}
\end{framed}
\vspace{-1em}
\caption{Construction of the recursion step $\AlgoName$}\label{figure:construction-embedding-4phases}
\end{figure}

\ifdraft
\fi
\subsubsection{Constructing the algorithm for the recursion step}\label{section:construction-embedding-recursion-step}
The core of the embedding is the recursion step algorithm $\AlgoName(p, \pos)$. As an input, the algorithm expects a path $p$ from a leaf edge $\ell_0$ (see Figure~\ref{figure:construction-embedding-4phases}, phase I.) of the already constructed CoSP tree to the tree's root and the position of the last node in this path. If the leaf node in this already constructed CoSP tree has several outgoing edges, the input $p$, in particular, also contains one of the outgoing edges (labelled with $\ell_0$ in Figure~\ref{figure:construction-embedding-4phases}). We first give an overview of the recursion step, then describe all auxiliary functions that we use in the full description, and finally present the construction in Figure~\ref{figure:construction-embedding} in full detail.

% The construction is presented in detail in Figure~\ref{figure:construction-embedding}.
%
\pparagraph{Overview}
Let the split-state transition system $T = (T_H,T_L,\attacker)$ (i.e., the program in the case of ADL) to be embedded and the initial configuration $s$ be fixed. 
As Figure~\ref{figure:construction-embedding-4phases} depicts, the algorithm can be divided into the following four phases.
\begin{compactitem}
\item[I.] Extract from the prefix $p$ the label $\ell_0$ of the last edge in the CoSP tree that has already been constructed. This last label $\ell_0$ determines whether the last subtree that was constructed in the co-recursive computation of the CoSP-tree ended with a computation node or an input node. This check corresponds to split-state transition system ending with a global transition $\tiLH$ or $\tiAH$. We call this last node $\nu$. We construct $\AlgoName$ such that $\nu$ is annotated with the sequence of states $({\color{darkgreen}s_0},\dots,{\color{darkgreen}s_q})$ of the honest program semantics.\footnote{This sequence of states corresponds to the resulting CoSP subtree from last (co-recursive) invokation of $\AlgoName$. This overview makes this correspondence precise.}
\begin{compactitem}
\item If $\nu$ is a computation node and $\ell_0 = ``no''$, set the variable {\bf\color{darkgreen}\textsf{y}} to the empty reference: {\bf\color{darkgreen}\textsf{y}}$:= \bot$.
\item Otherwise (i.e., $\ell_0 = ``yes''$ or $\nu$ is an input node), set the variable {\bf\color{darkgreen}\textsf{y}} to the position of $\nu$, as illustrated in Figure~\ref{figure:construction-embedding-4phases}.
\end{compactitem}
\item[II.] Reconstruct a path through the over-approximated honest program semantics $T_H$ (together with the initial state $s$) with the states $({\color{darkgreen}s_0},\dots,{\color{darkgreen}s_q})$. Check how many possible successors $s_q$ has. \AlgoName is constructed such that after $s_q$ there is always a global transition, either $\tiHL$ or $\tiHA$. \TODOE{Probably, we have to have the full sequence (from the beginning) of the honest program semantics. Otherwise, if it is a subsequence (e.g., the body of a loop) that potentially fits several prefixes, it might happen that $s_{q+1}$ is not unique and depends on the prefix. Check whether the proof uses the fact that we only talk about subsequences here.} 
We determine the successor node $s_{q+1}$ via the global transition mapping $\delta_{\ell_1}$: $s_{q+1} := \delta_{\ell_1}({\color{darkgreen}s_q},{\bf\color{darkgreen}\textsf{y}})$.
\item[III.] Run the over-approximated honest program semantics $T_H$ from $s_{q+1}$ until the first global transition $\ell_2$.
\item[IV.] Construct the corresponding CoSP-subtree as follows.
\begin{compactitem}
\item If $\ell_2$ is of the form $(\tiHL, f, u_k, \dot, u_{k+n-1})$, place a sequence $C_i$ of computation nodes that computes the value for each of the arguments $u_{k+i}$ ($i\in{0,\dots,n-1}$). Then, compute a sequence $C_f$ of computation nodes for the function $f$ that uses the results of the respective sequences $C_0, \dots, C_{n-1}$. Finally, annotate the final computation nodes of $C_f$ with the sequence $({\color{darkpurple}s_{q+1}}, \dots, {\color{darkpurple}s_{m}})$ of states that occurred in step III.
\item If $\ell_2$ is of the form $(\tiHA, f, u_1, \dot, u_{n})$, place a sequence $C$ of computation nodes that compute the value for $(f,u_1,\dots,u_n)$, where $f$ is only an identifier for the malicious function $f$ that triggered the communication to the attacker. Below, place an output node that sends the result of $C$ to the attacker, and one node further place an input node. Annotate the final input node with the sequence $({\color{darkpurple}s_{q+1}}, \dots, {\color{darkpurple}s_{m}})$ of states that occurred in step III.
\end{compactitem}
\end{compactitem}
\TODOE{CompN and potentially also Comp can have several final nodes. Is that a problem?}

\begin{framedFigure}{Construction of the embedding, where $\nu[(s,\pos)]$ means that the identifier of $\nu$ is set to $(s,\pos)$.}{figure:construction-embedding}
\footnotesize
\noindent\textbf{Phase I.}\\
Let $p' := p \rightarrow \nu \xrightarrow{[\ell_0]}$. 
Let ${\bf\color{darkgreen}\textsf{y}} := \bot$ if $\ell_0 = \mathstring{no}$ and ${\bf\color{darkgreen}\textsf{y}} := \ident({\color{darkgreen}\pos})$ otherwise. 
Let $({\color{darkgreen}s_0},\dots,{\color{darkgreen}s_q}) := \extractState(\nu)$.\\[0.5ex]
\noindent\textbf{Phase II.}\\
Let ${\color{darkpurple}s_{q+1}} := \delta_{\ell_1}({\color{darkgreen}s_q},{\bf\color{darkgreen}\textsf{y}})$ if $z = ({\ell_1},\_) \land \ell_1\in\set{\tiAH,\tiLH}$% and $s_{q+1} := s'$ otherwise
, and ${\color{darkgreen}s_0} \rightarrow \dots \rightarrow {\color{darkgreen}s_q} \xrightarrow{z}_H {\color{darkpurple}s_{q+1}}$\\[0.5ex]
\noindent\textbf{Phase III.}\\
Let ${\color{darkpurple}\ell_2}$ be defined as $t(p',{\color{darkpurple}s_{q+1}},1) = ({\color{darkpurple}s_m},m-(q+1))$
    $\land ~ {\color{darkpurple}s_m} \xrightarrow{{\color{darkpurple}\ell_2}}_H s_m$. In particular, ${\color{darkpurple}s_{q+1}} \rightarrow \dots \rightarrow {\color{darkpurple}s_m}$. \\[0.5ex]
\noindent\textbf{Phase IV.}
\begin{align*}
 &   \AlgoName(p \rightarrow \nu \xrightarrow{[\ell_0]}, {\color{darkgreen}\pos}) :=\\
&		\begin{cases}
        \begin{aligned}
            c_1 & \rightarrow^* c_l
        \rightarrow
        \OutN(\ident(\pos''))\\
        & \rightarrow
        \InN[(({\color{darkpurple}s_{q+1}}, \dots, {\color{darkpurple}s_m}), \pos')]
        \rightarrow
        \end{aligned}
		% right column
        &
        \begin{aligned}
            \text{if}~ & {\color{darkpurple}\ell_2} = (\tiHA, f, u_1,\dots, u_n) \text{, where } \pos' := \pos'' \cdot (0)^2\\
			&\text{and } (c_1 \rightarrow^* c_l,\pos'') := \mathit{Comp}({\color{darkgreen}\pos},f,u_1,\dots,u_n,p')
        \end{aligned}%
		% next case
		\\[1.5em]
        \begin{aligned}
        &c_1 \rightarrow^* c_l \rightarrow \\
		& \CompN(f, \ident(\pos_{k}), \ldots, \ident(\pos_{k+n-1}))\\
  & \qquad [(({\color{darkpurple}s_{q+1}}, \dots, {\color{darkpurple}s_m}), \pos')] \set{ \xrightarrow{\mathstring{yes}}, \xrightarrow{\mathstring{no}} }
        \end{aligned}
		% right column
        &
        \begin{aligned}
            \text{if}~ & {\color{darkpurple}\ell_2} = (\tiHL,f,u_k,\ldots,u_{k+n-1})\text{, where }  \\
			& (c_1 \rightarrow^* c_{l_k},\pos_k) := \mathit{Comp}({\color{darkgreen}\pos},u_k,p'), \dots, \\
            & (c_{1 + l_{k+n-2}} \rightarrow^* c_{l},\pos_{k+n-1}) := \mathit{Comp}({\color{darkgreen}\pos},u_{k+n-1},p'),\\
			& \pos' := \pos_{k+n-1} \cdot (0)^u,\\
			& \text{and $u$ is the length of $\CompN(f, \ident(\pos_{k}), \ldots, \ident(\pos_{k+n-1}))$}
        \end{aligned}%
		\end{cases}
		\\[3ex]
& \mathit{Comp}(\pos,\underline u,p) := \\
		&    \begin{cases}
        ( c_1 \rightarrow^* c_k, \pos\cdot(0)^k )
        & 
            \text{if } \underline u = u_1 \land ~ u_1\in\calS \land  S \in \set{\calN,\locationsDV,\set{\voidDV}} \\
        & \text{\phantom{if }} c_1,\ldots,c_k = \mathit{CoSPRep}(\iota_\calS(u_1),p)
        \\
        ((),u_1) & \text{if } \underline u = u_1 \land u_1\in \mathit{Pos}
		\\[2ex]
		(c_1 \rightarrow^* c_k \rightarrow^* c_l
        \rightarrow 
		& \text{if } \underline u = (u_1,\ldots,u_h) \land
        (c_1\rightarrow^* c_k, \pos_1) = \mathit{Comp}(\pos,u_1,p) \land 
        \\
		\CompN(\pair,\ident(\pos_1,p),\ident(\pos_2,p)),\pos_2 \cdot (0))
        &
        \phantom{\text{if }}(c_{k+1}\rightarrow^* c_l,\pos_2) = \mathit{Comp}(\pos_1,(u_2,\ldots,u_h),p)
        \rightarrow c_{k+1}
	    \end{cases}		
\end{align*}

$\mathit{CoSPRep}(t,p)$ outputs the representation of a term $t$ in
terms of a CoSP path $p$ in an topological sorting of
the following tree: for each position in $t$, there is a node
annotated with the constructor at the position, referencing the
identifiers of the nodes of its subterms. An edge exists between two
nodes iff the term corresponding to the second node is a subterm of
the term corresponding to the first node.\\[2ex]
$\CompN(f,\mathit{ref}_1,\dots,\mathit{ref}_n)$ is the CoSP tree that corresponds to the symbolic operation $f$ with the references $\mathit{ref}_1,\dots,\mathit{ref}_n$. Let $u$ be the length of the (linear) CoSP tree that corresponds to $f$.\TODOE{More precisely, $SO$ is a parametric CoSP protocol (see previous work). Shall we mention that?}

\TODOE{$\pos'$ is the position of the node to which it belongs. Shall we introduce a notation for that? It is actually something very intuitive. At the moment, it looks quite complicated.}

\TODOE{Would the proof be easier if the node identifier (e.g., of $\nu$) contains the entire execution sequence instead of the states of the quasi atomic trace (as it is now)?}
\end{framedFigure}

In the construction of $\AlgoNameString$, we use four auxiliary functions: $\extractState$, to extract the honest program semantics-state encoded in a node identifier, $\extractAllStates$, to extract a list of states that are encoded in a trace of a CoSP tree, $t$, to find the next global $\tHA$- or $\tHL$-transition, and $\ident$, to find a node for a given position parameter $\pos$.

\pparagraph{Extracting states: the auxiliary functions  $\extractState$ and $\extractAllStates$}
We define $\extractState(\nu)$ as the last state reconstructed out of the identifier $\nu$, and $\extractAllStates(\nu_1 \rightarrow \dots \nu_n)$ as the function that extract the entire sequence of states from $\nu_1 \rightarrow \dots \nu_n$.

\TODOE{Lift $\extractState$ to lists of states instead of defining $\extractAllStates$.}

\pparagraph{Next global $\tHA$- or $\tHL$-transition: the auxiliary function $t$}
We use $p\rightarrow$ to denote a path ending at an edge.
\begin{align*}
    t(p\rightarrow, \hat s,n) & := 
    \begin{cases}
        (\hat s,n) & 
        \text{if $S(p \rightarrow, \hat s)=\set{s'} \wedge \hat s \xrightarrow{\tHA\cup\tHL}_H s'$}\\
        % t(s') & \text{if $s'$ s.t. $s' := s(p \rightarrow, s)$, $s \xrightarrow{\tAH}_H$ s'}\\
        t(p \rightarrow, s', n+1) & \text{if  $S(p \rightarrow, \hat s)=\set{s'}$ but none of the above}\\
        (\hat s,n) & \text{if $S(p \rightarrow, \hat s)=\emptyset$}
    \end{cases}
    \intertext{where $S$ denotes the set of successor states, except
        when called on the state encoded in the last edge of $p$,
        denoted $\tl(p)$. \ifdraft In case that $s$ has more than one
        successor, we choose the successor with the lexicographically smallest label  % based on some (\eg, the lexicographic)
 %        ordering of the transition rules and the label 
 		of the last edge.\fi \TODOE{Do we still have multiple successors? Did we not get rid of them?}
    } 
    S(p \xrightarrow{l},\hat s)
    & := 
    \begin{cases}
        \set{ s'} &
        \text{if $|\set{s'' \mid 
		\hat s \rightarrow_H s''}| > 1 \wedge 
        \hat s \text{ equals the last state in } \extractAllStates(p)$} 
        \wedge \hat s \rightarrow_H^l s'
        \\
        \set{s' \mid \hat s \rightarrow_H s'} & \text{otherwise} \\
	\end{cases}
\end{align*}
\TODOE{Something is wrong here. I think, we do not need the set anymore. By letting the adversary resolve all non-determinism, we do not have non-determinism inside an atomic trace.}

\pparagraph{Node for a position: the auxiliary function $\ident$}
We construct a pair representation for every register name $v\in l$,
by first computing the node identifiers for each value at $r(v)$. 
If $r(v)\in\calN_c$, it is a reference to a previous node, which we
resolve with the following function, that identifies the position on
the current path. We slightly abuse notation by making the position in
the node identifier explicit. Note that positions must refer to node
on the previous path, as node identifiers, too, can only reference 
nodes in the prefix of the their path.
\begin{align*}
    \ident(pos,\emptyset) & := \bot \\
    \ident(pos,p' \rightarrow (\hat s,pos') \rightarrow) & := 
	\begin{cases}
		(\hat s,pos) & \text{if $pos = pos'$}\\
		\ident(pos,p' \rightarrow) & \text{otherwise}
	\end{cases}
\end{align*}

\subsubsection{Corecursively defining the embedding}
As outlined above, we are at this point in the position to define the mapping of a (potentially infinite) transition system (the honest program semantics) to an infinite CoSP tree. We define the mapping as the largest fixpoint of a corecursive procedure. For technical reasons (to make the fixpoint unique) and since CoSP trees are necessarily infinite, we first append to the translation of each leaf of the tree that is spanned by the transition system and an initial configuration an infinite chain of control nodes, which we call a dummy tree.

\begin{definition}[Dummy subtree]\label{def:dummy-subtree}
A \emph{dummy node} is a control nodes that has as out-metadata its position in the tree and that has a single successor with in-metadata $0$. A \emph{dummy subtree} is an infinite tree that solely consists of dummy nodes. We call the tree that solely consists of dummy nodes the \emph{empty CoSP protocol}.
\end{definition}

Finally, we can define the embedding of an ADL program as the largest fixpoint as the procedure that step-wise applies \AlgoNameString or appends a dummy subtree if the transition system reaches a final state.

\begin{definition}[Embedding of an ADL program]\label{def:embeddingsv}
Let $rootpath(edge)$ be the unique path from the edge to the root.
Consider the following definition $U(T_H,p)$:
\begin{itemize}
\item if $\AlgoName(p)$ is undefined return a dummy subtree (see Definition~\ref{def:dummy-subtree}) to $p$; \item otherwise, 
	\begin{itemize}
	\item let $pa := \AlgoName(p)$, 
	\item append $pa$ to $p$ resulting in a tree $par$, and 
	\item return the tree where to each leaf $le$ of $par$ the subtree $U(T_H(s),rootpath(le))$ is appended.
	\end{itemize}
\end{itemize}
The \emph{embedding $U(T_H(s))$ of an honest program transition system $T_H$ with initial configuration $s$} is co-recursively defined as the greatest fixpoint of $U(T_H(s),\varepsilon)$.
\end{definition}
We stress that this greatest fixpoint $\embedding(T_H(s))$ is unique.%, since $\AlgoName(p)$ completely describes the tree.
\NOTEE{We could prove that thoroughly.}

\pparagraph{Each CoSP protocol $\embedding(T_H(s))$ is efficient} For each honest program semantics $T_H(s)$ with initial configuration $s$, $\embedding(T_H(s))$ is efficient in the sense of CoSP. Our algorithm $\AlgoNameString$ gives rise to an algorithm that produces an efficient CoSP protocol
in the sense of Definition~\ref{def:cosp:efficient-protocol}, i.e.,
that only outputs the identifier of the 
identifier of the next node $N$ and
a set of labels for its outgoing edges. This is trivially done by
encoding the list of next steps within the identifier.

\subsection{Soundness of the CoSP-embedding}\label{sec:embedding}
We have to show that for every pair of transition system $T_{H,1}, T_{H,2}$ and respective initial configurations $s_1,s_2$ with $T_{H,1}(s_1) \execEquivalent T_{H,2}(s_2)$ (i.e., that are equivalent in the split-state semantics) that $T_{H,1}(s_1) \CoSPEquivalent T_{H,2}(s_2)$ holds (i.e., they are is symbolically equivalent in the sense of CoSP). First, we introduce some notation that we use to make the proof more readable, and then we present the soundness proof.

% The edges of the tree follow a possible path in the non-deterministic interaction $\interact{\SExec}{\attacker_{DY}}$
%
% Two states $s_1, s_2$ are subsequent in the non-deterministic interaction $\interact{\SExec}{\attacker_{DY}}$ if there is a message that $\attacker_{DY}$ would send such that the state variable $s$ changes in Line~\ref{state-application} its value from $s_1$ to $s_2$. \TODOE{Better define possible traces and say that the tree consists of all possible traces.}

\subsubsection{Preliminary definitions}
In order to make the proof more readable, we introduce some notation: first the view of an attacker, called out-traces, and second quasi atomic traces, which are used to characterize the subsequences of an honest program semantics execution sequence that is processed by $\AlgoNameString$.

\pparagraph{Distribution of out-traces}
In the CoSP-computational execution $\interact{\Exec_{\Model,\implementation,T_H(s)}(\secpar)}{\attacker(\secpar)}$, let $\mathit{OutTraces}_{\Model,\implementation,T_H(s),\attacker,p}(\secpar)$ be the distribution of messages that $\Exec_{\Model,\implementation,T_H}(\secpar)$ sends to $\attacker(\secpar)$ if output nodes are encountered such that the combined runtime of $\Exec_{\Model,\implementation,T_H}(\secpar)$ and $\attacker(\secpar)$ is $\le p(\secpar)$.

\TODOR{Proposal: Instead of $\mathit{OutTraces}$ for transition systems.}
Define for transition systems:
\begin{itemize}
    \item 
    \[
        \Pr [ (a_1,\ldots,a_n) = \Traces_{\le n}(T)   ]
        =
        \sum_{\parbox{2cm}{\footnotesize$(\alpha_1,\ldots,\alpha_m)|_A \\= (a_1,\ldots,a_n) \\ \land
        m\le n$}}
    \Pr [ \Exec(T) = s_0 \xrightarrow{\alpha_1}_{m_1} \cdots \xrightarrow{\alpha_m}_{m_m} s_m ] \]
\item And in general, for random variable $X$,
    $ \Pr[ t = X|_S ] \defeq \sum_{t'|_S = t} \Pr[ t' = X ]$.
\end{itemize}

Then, let $T$ be any split-state semantics
with Crypto-API semantics $T_L$ which harmonizes with the implementation \implementation
\ifdraft
\TODOR{Something similar to
Definition~\ref{def:compliant} here.
Need to split Definition~\ref{def:compliant}, so it says
ADL program is compliant if it is pre-compliant and library semantics
of split-state representation is compliant to libspec. Then have
definition of ``compliant library semantics'' before that and refer to
it here.
}, honest program semantics
$T_H$ and attacker semantics $T_A$.
Then $ \Traces_{\le p(n)}(T)|_{\tHA\cup\tfinal}$ is one distribution.

Then, later if CoSP is a transition system, we can use the same
notation. For now, I propose to fix $\Model$ and $\implementation$
globally, and define
$\Traces(\interact{\Exec_{T_H}(\secpar)}{\attacker(\secpar)})$
as the distribution over the whole trace (not just the final judgement), and
prefixing
message so they are in $\tHA$ and $\tAH$, and
the final output with final so it is in
$\tfinal$.
If we additionally use notation
$\Traces_{\le n}(\interact{\Exec_{T_H}(\secpar)}{\attacker(\secpar)})$
like above, we arrive at
$\Traces_{\le n}(\interact{\Exec_{T_H}(\secpar)}{\attacker(\secpar)})|_{\tHA\cup\tfinal}$.
\TODOR{ProposalEnd}
\fi

Let 
    $\mathit{OutTraces}_{SSo, T_H(s),\attacker,p}(\secpar)$
be the distribution of $\tHA$-traces 
of the transition system $T = (T_H,T_L,\attacker)$
with runtime less or equal than
$p(\secpar)$, i.e.,
\[
    \Pr[t= \mathit{OutTraces}_{SSo, T_H,\attacker,s,p}(\secpar)]
    =
    \sum_{ t'|_\pi = t } \Pr[ t' = \Traces^\mathit{ss}_{\secpar,z}(T(s)) ],
\]
where $\pi( (\tHA,m) ) =m$ and undefined otherwise.

    % $$\Pr[(out_1, \dots, out_k) \leftarrow T(s_H,s_A,s_L)] \sum_{ (s_0 \xrightarrow{\alpha_1} \xrightarrow{(\mathstring{out}, out_1)} \cdots s_k)\in\supp(\Exec(T)) }
    % \Pr [\Exec(\ADL) = s_0 \xrightarrow{\alpha_1} s_1 \cdots
    %     \xrightarrow{\alpha_k = (\tifinal,x)}
    % s_k
    % \text{ and  $k \le n$}
    % ]$$

% \begin{definition}[Quasi-maximal partitioning into quasi atomic traces]
% For a sequences of transitions $s\xrightarrow{\underline l}^n s_n$, a partitioning $p_1, \dots, p_m$ of a sequence of transitions into quasi atomic traces is a \emph{quasi-maximal} partitioning if for every $i$ $p_i$ is a $n_i$ quasi atomic trace and the $n_i$ is maximal, i.e., $p_i,p_{i+1}$ cannot be combined into a $n_i + n_{i+1}$ quasi atomic trace.
% \end{definition}
%

% \begin{definition}[Embedding mapping quasi atomic traces into CoSP]
% Let $T_H$ be an honest program semantics. For a quasi atomic trace $q$ in $T_H$, let the embedding $\embedding(q)$ quasi atomic trace be defined as follows:
% $$\embedding(q) = \AlgoName(q)$$
% \end{definition}

\pparagraph{Quasi atomic traces}
Given a path ending on an edge, and a state, the first step in the
embedding is to iteratively compute the follow-up state until it is
no longer uniquely determined, or it involves the attacker or crypto-API.
We introduce the notion of an atomic trace and a global transition.

\begin{definition}[Global transition]\label{def:global-transition}
We call a transition $s \xrightarrow{l} s'$ a \emph{global transition} if $l \in \tHA\cup\tAH\cup\tHL\cup\tLH$.
\end{definition}

\begin{definition}[Atomic trace]\label{def:atomic-trace}
A sequence $\underline s$ of of transition steps $\underline s := s_1 \xrightarrow{\underline l}^n s_n$ is called an \emph{atomic trace} if no global transition occurs. We call a sequence $\underline s$ a \emph{quasi atomic trace} if $\underline s =  s_1 \xrightarrow{l_1} \underline s' \xrightarrow{l_2} s_{|\underline s'|+2}$ and $l_1\in \set{\tiAH,\tiLH}$ and $l_2\in \set{\tiHA,\tiHL}$.
\end{definition}

\subsubsection{The proofs}

\begin{lemma}\label{lemma:embedding-state-correspondence}
For a sequence $\underline s$ of transitions from an honest program semantics $T_H$ that begins at an initial configuration $s$, let $q_0, \rightarrow_H \dots \rightarrow_H q_n$ be a partitioning into quasi atomic traces $q_i$. Let $\rightarrow$ denote an edge in the CoSP tree. 
Let $\embedding(q_i) := \AlgoName(q_i)$, and let $\extractAllStates$ be defined as in Section~\ref{section:construction-embedding-recursion-step}. % denote all CoSP subtrees that are generated for the quasi atomic trace $q_i$. This solely contains one element; hence we treat $e(q_i)$ as a CoSP subtree.
% and pairs of global transitions $\underline l_i \in \underline l \in \set{((\tiHA,y_1),(\tiAH,y_2)), ((\tiHL,y_1),(\tiLH,y_2)) \mid \text{for some } y_1,y_2}$.
With this notation, the following two properties hold.
\begin{enumerate}[$(i)$]
\item\label{item:embedding-base-case} For all initial atomic traces $q_0$, we have $\extractAllStates(\embedding(q_0)) = q_0$.
\item\label{item:induction-step} For all $i\in\mathbb{N}$ there is a polynomial $p$ such that 
% for $\underline l \in \set{((\tiHA,y_1),(\tiAH,y_2)), ((\tiHL,y_1),(\tiLH,y_2)) \mid \text{for some } y_1,y_2}$
we have $$q_i \rightarrow_H q_{i+1} \implies e(q_i) \rightarrow e(q_{i+1}) \land \extractAllStates(\embedding(q_{i+1})) = q_{i+1}$$
\end{enumerate}
\end{lemma}
\begin{proof}
% First we show that $e(q_i)$ (for all $i$) is a singleton set, since looking at the construction of $\AlgoName[T_H]$ we can see that it produces the same CoSP subtree for the same quasi atomic trace.
%
% Next, w
We show Property~$(\ref{item:embedding-base-case})$. Upon each invocation of $\AlgoName[T_H]$ a quasi atomic trace or an atomic trace following by one global $\tiHA$- or $\tiHL$-transition is generated inside $\AlgoName[T_H]$ and stored in the last node; hence Property~$(\ref{item:embedding-base-case})$ holds.

Next, we show that Property~$(\ref{item:induction-step})$ also follows from the construction of $\AlgoName[T_H]$ and $\embedding$. If $q_i \rightarrow 
q_{i+1}$, then $\embedding$ calls $\AlgoName[T_H]$ with a prefix $p \rightarrow \nu \xrightarrow{x}$ such that $\extractAllStates(\nu) = e(q_i)$. 
Then, $\AlgoName[T_H]$ internally computes the transitions beginning from the final state of $q_i$. The branchings coincide since by 
Definition~\ref{def:sso} $\rightarrow_H$ only branches after $\tiLH$-transitions, which 
are properly captured by the definition of $\AlgoName[T_H]$.
We know that the quasi atomic trace beginning from the final state of $q_i$ is a sequence due to the fact that whenever a branching occurs, the library or the adversary is queried. Moreover, the computation in $\AlgoName[T_H]$ does not get stuck, despite the fact that the registers carry references instead of values, due to Property~$(\ref{item:domain-independence})$ of Definition~\ref{def:sso-h}.
Since $q_i \rightarrow q_{i+1}$, there is a subtree $e(q_i) \rightarrow e(q_{i+1})$ in $e(T_H)$ and as above $\extractAllStates(e(q_{i+1})) = q_{i+1}$ holds by inspecting the construction of $\AlgoName[T_H]$.
\end{proof}

\TODOE{Do we have to require for all binary operations that for all $x,y\in\set{0,1}^*$ $d_{op}(\iota(x),\iota(y)) = \bot$ iff $\underline {op}(x,y) = \false$ holds?}

\begin{lemma}\label{lem:embedding-equal-outtraces}
Let a symbolic model $\Model$ together with an implementation $\implementation$ be given. For each honest program semantics $T_H$,
for each (attacker) interactive machine $\attacker$ such that $(T_H, T_L, \attacker)$ satisfies Definition~\ref{def:sso}, the following holds: there is an interactive machine $R_\attacker$ and a polynomial $p'$ such that for sufficiently large security parameters $\secpar$ we have 
	$$\mathit{OutTraces}_{SSo, T_H,\attacker,s,p}(\secpar) = \mathit{OutTraces}_{\Model,\implementation,\embedding(T_H(s)),R_\attacker,p'}(\secpar)$$
	if $\mathit{OutTraces}_{SSo, T_H,\attacker,p}(\secpar)$ uses some Crypto-API $T_L := (T_{L,\secpar})_{\secpar\in\mathbb{N}}$ that harmonize with $\implementation$ (see Definition~\ref{def:harmonizing-crypto-api}).
    % \[
    % \Pr[
    %     \interact{ \Pi\stateDVs}
    %     {\adv}(\secpar,z)
    % \downarrow_{n}x]
    % =
    % \Pr[ \interact{\Exec_{\M,\implementation,\Pi}(\secpar)}{\advE(\secpar)}
    % \downarrow_n x]
    % \]
\TODOE{Do we also have to quantify over all initial configurations in $T_H$?}
\end{lemma}
\begin{proof}
We show that for every honest program semantics $T_H$, every (attacker) interactive machine $\attacker$, and every resulting out-trace, there is a corresponding out-trace in the CoSP execution. We prove this by induction over the length of the out-trace of the honest program semantics $T_H$. More formally, for all honest program semantics $T_H$, we consider a modified transition system $\restrict{T_H}{n}$ that behaves just like $T_H$ except that $\restrict{T_H}{n}$ stops after it produced $n$ out labels. We prove by induction over $n$ that for all honest program semantics $T_H$ and all $n\in\mathbb{N}$, the statement holds for $\restrict{T_H}{n}$; hence the statement holds for polynomial-time computable prefixes.
\smallskip

\noindent\textbf{Constructing the reduction $R_\attacker$.} Before we begin with the induction proof, we define an interactive machine $R_\attacker$ from the attacker semantics $\attacker$.
\begin{itemize}
\item The machine $R_\attacker$ internally executes the attacker semantics $\attacker$ by iteratively computing the successors until either a transition with an $\transitionAHIdentifier$ label or until there is no unique successor. 
\item Whenever there is no unique successor but a distribution of successors, $R_\attacker$ conducts a weighted random choice, weighted with the probabilities of each successor.
\item Whenever the semantics conducts a transition $s \xrightarrow{(\transitionAHIdentifier, m)} s'$, store $s'$ and send $m$ to the communication partner (i.e., the computational CoSP execution).
\item Whenever $R_\attacker$ receives a message $f,m_1, \dots, m_n$ from the interaction partner (i.e., the computational CoSP execution) and the (internal) attacker semantics is currently in state $\stateDV[\varepsilon,as,\varepsilon]$, $R_\attacker$ sets the successor state to be $s' := \stateDV[(f,m_1, \dots, m_n),as,\varepsilon]$, and using $s'$ the attacker $R_\attacker$ then internally executes the attacker semantics as above.
\end{itemize}

\newcommand{\OutSSo}{\mathcmd{\text{\textrm{\small{}out}}\textsc{ss}\text{\textrm{\small{}o}}}}
\newcommand{\OutCoSP}{\mathcmd{\text{\textrm{\small{}out}}\textsc{cosp}}}

For notational convenience, we introduce the following notation $\OutSSo(i) := \mathit{OutTraces}_{SSo, \restrict{T_H}{i},\attacker,p}(\secpar)$ and $\OutCoSP(i) := \mathit{OutTraces}_{\Model,\implementation,\embedding(\restrict{T_H}{i}),R_\attacker,p}(\secpar)$.
\medskip

\par\noindent{\bf Induction base: $n=0$.~} By definition, for all honest program semantics $T_H$, the support of $\OutSSo(0)$ is the empty set. The support of $\OutCoSP(0)$ is empty as well, since $\AlgoName[\restrict{T_H}{0}]$ (and thus $\embedding(\restrict{T_H}{0})$) solely produces output nodes whenever there is an $\tiHA$-label in the honest program semantics of SSo.
\medskip

\par\noindent{\bf Induction step: $n > 0$.~} We assume that for all honest program semancs $T_H$ the restricted program $\restrict{T_H}{n-1}$ satisfies the statement. Let $T_H$ be an arbitrary but fixed honest program semantics. 
\smallskip

\noindent\textit{Fixing the randomness.} It suffices to prove the statement for an arbitrary but fixed set of random choices $r := (r_A,r_L)\in\set{0,1}^*$ in $\attacker$, written as $\attacker(r_A)$, and the cryptographic library, written as $\implementation(r_L)$ in CoSP and $SSo(r_L)$ in the over-approximated split-state semantics. We show that for sufficiently large $\secpar$, for all polynomials $p$, there is a polynomial $p'$ such that for $r := (r_A,r_L)$ 
$$\OutSSo(n,r) = \OutCoSP(n,r),$$
where $\OutSSo(i,r) := \mathit{OutTraces}_{\ADLo(r_L), \restrict{T_H}{i},\attacker(r_A),p}(\secpar)$ and $\OutCoSP(i,r) := \mathit{OutTraces}_{\Model,\implementation(r_L),\embedding(\restrict{T_H}{i}),R_{\attacker(r)},p}(\secpar)$.
\smallskip

\noindent\textit{Reduction to analyzing the last element in $\OutSSo(n,r)$.} We first show that it suffices to concentrate on proving the statement for the last element in $\OutSSo(n,r)$. 
Recall that $l|_0^k$ denotes the $k$-prefix of $l$.
If
\[ \OutSSo(n,r)|_0^{n-1} \neq \OutCoSP(n,r)|_0^{n-1} \]
holds, then there is a subset $r'$ of $r$ such that 
$$\OutSSo(n-1,r') \neq \OutCoSP(n-1,r')$$
holds by induction hypothesis.
\smallskip

\noindent\textit{Analyzing the last element in $\OutSSo(n,r)$.} For the last element in $\OutSSo(n,r)$, we treat all out-transitions uniformly. The last element in $\OutSSo(n,r)$ 
is of the form $(f,([r(v)]_\calN)_{v\in l})$ for $f\in\mathcal{UNOP}
\cup \mathcal{BINOP} \cup \mathcal{RELOP}\cup \mathit{Mal}_{static}$
 and $l\in\set{(v_b), (v_b,v_c), (v_a,v_b), (v_a,\dots,v_e)}$, caused by a honest program transition $s \xrightarrow{(\tiHA, f,([r(v)]_\calN)_{v\in l})} s$. $R_{\attacker(r)}$ behaves just like $\attacker(r)$ by definition and in particular the input messages are the same. Hence, we know that the inputs, sent to the computational CoSP execution are the same as those that the attacker semantics sends via $\tiAH$-transition labels. Moreover, for fixed randomness, the output of the cryptographic library in the CoSP execution are exactly the same as in the library semantics. \TODOE{Use the argument here that the cryptographic library is used as an algorithm.Use as $\implementation$ the implementation from the ADL program that encodes the library.} 
 
As a consequence and by Lemma~\ref{lemma:embedding-state-correspondence}, the state $s$ (from above) is the same (up to the content of the registers) as the state that is internally computed by $\AlgoName[\restrict{T_H}{n}]$ before the output node is produced.

Finally, we first stress that the number of computational steps of each invocation of $\AlgoName$ is polynomially bounded (in $\secpar$). Hence, for each polynomial $p$ there is a polynomial $p'$ such that the computations coincide. In conclusion, Claim~1 together with definition of $\AlgoName[\restrict{T_H}{n}]$ implies that the last element in $\OutSSo(n,r)$ coincides with the last element in $\OutCoSP(n,r)$.
 \end{proof}

We first prove the preservation in the computational model.
\begin{lemma}\label{lemma:embedding-soundness-computational-general}
    Let $T_{H,1}$ and $T_{H,2}$ be two uniform families with respective initial configurations $s_1,s_2$
    of ADL programs compliant (see Definition~\ref{def:compliant}) with the same
    library specification \libSpec w.r.t.\ to a symbolic model \Model
    and an implementation of it, called \implementation.
    Let
$e$ the injective embedding into CoSP\@,
    and
    $s_1$ and $s_2$ initial configurations.
    Then,
    \[ 
        \embedding(T_{H,1}(s_1)) 
        \CoSPIndistinguishable \embedding(T_{H,2}(s_2)) 
    \implies 
T_{H,1}(s_1) \execIndistinguishable T_{H,2}(s_1). \]
\end{lemma}
\begin{proof}
    We assume
    $\embedding(T_{H,1}(s_1)) \CoSPIndistinguishable 
    \embedding(T_{H,2}(s_2))$,
    \ie,
    for all machines $\advE$ and all polynomials $p$, 
    there is a negligible function $\mu$ such
    that for all $a,b\in\bit$ with $a\ne b$, all $z\in\bits*$ and
    $\secpar\in\setN$:
    \begin{align*}
        & \Pr[\interact{\Exec_{\Model,\implementation,e(T_{H,1}(s_1))}(\secpar)}{\advE(\secpar,z)}\downarrow_{p(\secpar)} a]\\
        & + \Pr[\interact{\Exec_{\Model,\implementation,e(T_{H,2}(s_2))}(\secpar)}{\advE(\secpar,z)}\downarrow_{p(\secpar)} b]
        \le 1 + \mu(\secpar).
    \end{align*}

    Proof by contradiction. Assume 
    $T_{H,1}(s_1) \notExecIndistinguishable
    T_{H,2}(s_2)$, \ie,
    for some family of adversaries $(A^\secpar)_{\secpar\in\setN}\in\advM$,
    $\ADLss_{T_{H,1}^\secpar,\adv^\secpar,s_1}$
    and
    $\ADLss_{T_{H,2}^\secpar,\adv^\secpar,s_2}$
    are
    distinguishable (see Definition~\ref{def:adlss-ind}). Thus,
    there are $a,b\in\bit$, $a\ne b$ and $\secpar\in\setN$ such that,
    there is a  polynomials $p$ such that for all negligible functions
    $\mu$,
    \[
        \Pr\left[ 
    \ADLss_{T_{H,1}^\secpar,\adv^\secpar,s_1} \downarrow_{p(\secpar)} a\right]
                 + \Pr\left[
    \ADLss_{T_{H,2}^\secpar,\adv^\secpar,s_2} 
                  \downarrow_{p(\secpar)} b\right]
            > 1 + \mu(\secpar).
        \]
        (See Definition~\ref{def:ss-indistinguishability}.)
    By Corollary~\ref{cor:adlo}, there is an attacker $A$,
    and a polynomial $p'$ s.t. 
    \[
        \Pr\left[ 
    \ADLo_{T_{H,1}^\secpar,\adv',s_1} \downarrow_{p(\secpar)} a\right]
                 + \Pr\left[
    \ADLo_{T_{H,2}^\secpar,\adv',s_2} 
                  \downarrow_{p(\secpar)} b\right]
            > 1 + \mu(\secpar).
        \]
    Thus, by Lemma~\ref{lem:embedding-equal-outtraces}, we
    obtain the contradiction to the inequality assumed to, 
    concluding
    the proof.
    \TODOR{This is not immediate to see, as the equality in  
     Lemma~\ref{lem:embedding-equal-outtraces} is on OutTraces, but
     I'd rather simplify the notation, so it is obvious that the final
 output is part of this traces.} \TODOE{Include the final guess into the out-traces. The proof of Lemma 10 assumes it.}
\end{proof}

In order to match the notation in the overview, we write $\embedding(\protos i)$ for $\embedding(T_{H,i}(s_i))$ in the following lemmas, for two ADL programs $\proto_1, \proto_2$ and two initial configurations $s_1, s_2$.

\begin{lemma}\label{lemma:embedding-soundness-computational}
    Let $\proto_1$ and $\proto_2$ two ADL program together with input configurations $s_1,s_2$
in the canonical symbolic model with respect to a CoSP-symbolic model $\Model$. Let $\proto_1$ and $\proto_2$ be compliant (see Definition~\ref{def:compliant}) with the same library specification w.r.t\ a symbolic model $\Model$ and an implementation $\implementation$.
    Then,
    \[ 
        \embedding(\protos 1) 
        \CoSPIndistinguishable \embedding(\protos 2) 
    \implies 
\protos 1 \execIndistinguishable \protos 2. \]
\end{lemma}

Next, we can prove the statement for the symbolic side.

\begin{lemma}\label{lemma:embedding-soundness-symbolic-general}
    Let $T_{H,1}$ and $T_{H,2}$ be honest program semantics with respective input configuration $s_1,s_2$
in the canonical symbolic model with respect to a CoSP-symbolic model $\Model$.
Then,  
$$T_{H,1}(s_1) \execEquivalent T_{H,2}(s_2)
\implies 
\embedding(T_{H,1}(s_1)) \CoSPEquivalent \embedding(T_{H,2}(s_2))$$
\end{lemma}
\begin{proof}
First, observe that the embedding takes as input any honest program semantics in the over-approximated semantics. Hence, it does exactly the same for the honest program semantics in canonical symbolic model.
% We show that if the trace equivalence breaks in CoSP, then also trace equivalence breaks in the canonical symbolic model.
The same argumentation as in Lemma~\ref{lem:embedding-equal-outtraces} can be made to show that the states in the canonical symbolic semantics and in the identifiers of the nodes in the embedded program (characterized by the honest program semantics) coincide. The most significant difference is that the attacker is not defined as an interactive machine but rather as a set of deduction rules that are induced by the constructors and the destructors. Instead of fixing an attacker machine, as in the proof of Lemma~\ref{lem:embedding-equal-outtraces}, we fix an attacker strategy (see Definition~\ref{def:symbolic-execution}) and then conduct the same induction proof over the length of the deduction sequence (including the steps of the attacker's deductions). Trace equivalence of $T_{H,i}(s_i)$ (for $i\in\set{1,2}$) in the canonical symbolic model and of $\embedding(T_{H,i}(s_i))$ then implies
$$T_{H,1}(s_1) \execEquivalent T_{H,2}(s_2) \implies \embedding(T_{H,1}(s_1)) \CoSPEquivalent \embedding(T_{H,2}(s_2))$$
\end{proof}

\begin{lemma}\label{lemma:embedding-soundness-symbolic}
    Let $\proto_1$ and $\proto_2$ two ADL program together with input configurations $s_1,s_2$
in the canonical symbolic model with respect to a CoSP-symbolic model $\Model$. Let $\proto_1$ and $\proto_2$ be pre-compliant with the same library specification.
Then,  
$$\protos 1 \execEquivalent \protos 2
\implies 
\embedding(\protos 1) \CoSPEquivalent \embedding(\protos 2)$$
\end{lemma}
\begin{proof}
The statement immediately follows from Lemma~\ref{lemma:embedding-soundness-symbolic-general}.
\end{proof}
\else
%!TEX root = main.tex

\begin{definition}[ADL split-state equivalence $\execIndistinguishable$]\label{def:adlss-ind}
Let $\proto_1$ and $\proto_2$ be two families of ADL programs compliant with the
same library specification \libSpec and $s_1,s_2$ initial configurations for $\proto_1$ and $\proto_2$, respectively. We write
$ \protos 1 \execIndistinguishable \protos 2$,
if, for all adversaries $\adv\in\advM$,
for the over-approximated ADL split-state representations 
$(\ADLo_{\proto_1^\secpar,\adv^\secpar,s_1})_{\secpar\in\setN}$ of $\protos 1$ and \adv,
and
$(\ADLo_{\proto_2^\secpar,\adv^\secpar,s_2})_{\secpar\in\setN}$ of $\protos 2$ and \adv, we have
$T_{H,1}(s_1) \SSIndistinguishable{T_L,T_A} T_{H,2}(s_2)$.
\end{definition}

Next, we state the lemma that connects the ADL semantics to the over-approximated split-state form.
\begin{lemma}\label{lemma:computational-execution-adl}
Let $\proto_1$ and $\proto_2$ be two families of ADL programs pre-compliant with the
same library specification \libSpec and $s_1,s_2$ initial configuration for $\proto_1$ and $\proto_2$, respectively. Then
$\protos 1 \execIndistinguishable \protos 2
\implies \protos 1 \ADLIndistinguishable \protos 2$.
\end{lemma}

\subsection{Canonical symbolic semantics}
The canonical over-approximated split-state semantics directly defines
a canonical symbolic semantics. We replace the domain $\calD_c$ by the
set of terms that are defined by the CoSP-symbolic model $\Model$, we replace the attacker semantics by
the symbolic attacker from $\Model$, and we replace the crypto-API by
the constructors and destructors in $\Model$.

%We define symbolic equivalence as follows.

\begin{definition}[Symbolic equivalence $\execEquivalent$]
\label{def:ss-symbolic-indistinguishable}
Two ADL programs $\proto_1$ and $\proto_2$,
and initial configurations 
$s_1=\stateDVsN{_1}$ 
and
$s_2=\stateDVsN{_2}$ 
are 
\emph{symbolically split-state equivalent ($\protos 1 \execEquivalent \protos 2$)}
if for all attacker strategies $I$, their respective ADL symbolic
split-state semantics $T_1,T_2$ w.r.t. to $I$ are symbolically equivalent, i.e.,
 if  $\Views(T_1) \sim \Views(T_2)$.
\end{definition}

Next, we state the lemma that connects symbolic ADL to the symbolic split-state composition of ADL.

\begin{lemma}\label{lemma:symbolic-adl-symbolic-execution}
Let $\proto_1$ and $\proto_2$ be two ADL program pre-compliant to the same library specification \libSpec and $s_1,s_2$ initial configuration for $\proto_1$ and $\proto_2$, respectively. Then,
$$\protos 1 \ADLEquivalent \protos 2 \implies \protos 1 \execEquivalent \protos 2$$
\end{lemma}

%!TEX root = main.tex
% \subsection{Soundness for ADL Split-State}
% We list the main technical lemmas that prove the computational soundness of Dalvik Bytecode, beginning with the lemmas that connect symbolic ADL to the symbolic split-state composition of ADL and the computational split-state composition of ADL with (computational) ADL. For the proofs, we refer to the technical report~\cite{ourfullversion}.

% \begin{lemma}\label{lemma:symbolic-adl-symbolic-execution}
% Let $\proto_1$ and $\proto_2$ be two ADL program pre-compliant to the same library specification \libSpec and $s_1,s_2$ initial configuration for $\proto_1$ and $\proto_2$, respectively. Then,
% $$\protos 1 \ADLEquivalent \protos 2 \implies \protos 1 \execEquivalent \protos 2$$
% \end{lemma}
%
% \begin{lemma}\label{lemma:computational-execution-adl}
% Let $\proto_1$ and $\proto_2$ be two families of ADL programs compliant with the
% same library specification. Then
% $\proto_1 \execIndistinguishable \proto_2
% \implies \proto_1 \ADLIndistinguishable \proto_2$.
% \end{lemma}
%

\subsection{The CoSP-embedding}
We are finally in the position to construct the embedding $\embedding$ from the part of the program that is encoded in the honest-program semantics into CoSP. Recall that the symbolic variant of ADL is also solely defined on the honest-program semantics of a program. Recall that we map the honest-program semantics $T_H$ that belongs to an ADL program $\proto$ into CoSP. In the context of a pair of ADL programs $\proto_1, \proto_2$, we write $T_{H,i}$ for the honest program semantics for $\proto_i$.

The main lemmas of this section show that for every pair of ADL programs $\proto_1, \proto_2$ with initial configurations $s_1,s_2$ the equivalence of $\protos 1$ and $\protos 2$ in symbolic ADL implies the symbolic equivalence of $\embedding(T_{H,1}(s_1))$ and $\embedding(T_{H,2}(s_2))$ in CoSP and computational indistinguishability of $\embedding(T_{H,1})$ and $\embedding(T_{H,2})$ in CoSP implies computational indistinguishability of $\protos 1$ and $\protos 2$.

First we construct the embedding, and then we give the main arguments why the embedding preserves symbolic equivalence and computational indistinguishability.
\pparagraph{Constructing the embedding into CoSP}\label{section:construction-embedding}
%!TEX root = main.tex
\newcommand{\extractState}{\mathfun{exSt}}
\newcommand{\extractAllStates}{\mathfun{exAllSt}}
As CoSP trees are infinite, we define the embedding in a co-recursive manner, i.e., as the largest fixpoint of a co-recursive construction. Each step in this recursion is defined by a function $\AlgoName$ that takes as input a trace from a leaf-node in the so-far constructed CoSP tree to the root node and outputs a finite subtree. 
Within the embedding, we re-interpret the transition system of the honest
program introduced in Section~\ref{section:approx} by instantiating
the set $\calD_c$ to be the set of positions in a CoSP tree.
With this representation, registers and heap locations can store values input by the
adversary or the crypto-API by pointing to the position of
the respective input or computation nodes in addition to numerical
values, locations, and \voidDV{}.
Whenever this kind of special
register values are supposed to be given to the adversary or the crypto-API,
the position is resolved to a node identifier. 

The recursion step $\AlgoName$ in the construction of the embedding
internally runs the honest semantics of this modified version of the
over-approximated split-state semantics. It computes
the next state from deterministic
(internal) transitions until it arrives at a non-deterministic
(calling) transition of
the honest-program semantics, i.e., either a $\tiHL$- or
a $\tiHA$-transition. 
Each call to the library (i.e., an \tiHL-transition)
is transformed into
a computation node with two successors (labelled $\mathit{yes}$ and
$\mathit{no}$), referencing freshly constructed values or previous
nodes depending on the information in the label of the transition.
Each call to the adversary (i.e., an \tiHA-transition)
is transformed into
an output node, again referencing values or nodes based on the label
of the transition, and an input node for the $\tiAH$-transition
following suit.

\pparagraph{Soundness of the embedding}
The main idea in the proof is to consider the sub-sequences of the traces
that start and end with calls to, or responses from the attacker or
the
library. 
Each step between the sub-sequences corresponds
to an edge in the CoSP-tree. 
If we 
choose a
canonical over-approximation such that
the domain of the variable values is the set of positions in the
CoSP-tree, then we can show the following relation:
for the over-approximated split-state semantics,
there is a direct correspondence between the states 
at the end of such a sub-sequence 
and the nodes in the CoSP-tree,
which store the current state in the node-identifier.

In order to match the notation in the overview, we write $\embedding(\protos i)$ for $\embedding(T_{H,i}(s_i))$ in the following lemmas, for two ADL programs $\proto_1, \proto_2$ and two initial configurations $s_1, s_2$.

\begin{lemma}\label{lemma:embedding-soundness-symbolic}
    Let $\proto_1$ and $\proto_2$ two ADL program together with input configurations $s_1,s_2$
in the canonical symbolic model with respect to a CoSP-symbolic model $\Model$. Let $\proto_1$ and $\proto_2$ be pre-compliant with the same library specification.
Then,  
$$\protos 1 \execEquivalent \protos 2
\implies 
\embedding(\protos 1) \CoSPEquivalent \embedding(\protos 2)$$
\end{lemma}

For our the next lemma and our main theorem, we require that two ADL programs are \emph{compliant}, i.e., the programs are pre-compliant with the same library specification (see Section~\ref{section:split-state-adl}), and the library produces the same distribution as the implementation $\implementation$ from the computational soundness result with a library specification $\libSpec$.

\begin{lemma}\label{lemma:embedding-soundness-computational}
    Let $\proto_1$ and $\proto_2$ two ADL program together with input configurations $s_1,s_2$
in the canonical symbolic model with respect to a CoSP-symbolic model $\Model$. Let $\proto_1$ and $\proto_2$ be compliant with the same library specification w.r.t\ a symbolic model $\Model$ and an implementation $\implementation$.
    Then,
    \[ 
        \embedding(\protos 1) 
        \CoSPIndistinguishable \embedding(\protos 2) 
    \implies 
\protos 1 \execIndistinguishable \protos 2. \]
\end{lemma}

\fi

Our computational soundness result is parametric in a given symbolic
model and given conditions $C_I$ to the implementation and the protocols
such that computational soundness in the sense of the CoSP framework
holds. Our result states that for any symbolic model with conditions
C in CoSP, equivalence in symbolic ADL (see
Section~\ref{sec:dalvik-symbolic-semantics}) implies
indistinguishability in ADL (see Section~\ref{sec:securityframework}).
Since all CoSP results in the literature characterize the protocol
class by a set of protocol conditions $C_P$, we use these protocol
conditions in our theorem as well. 

\iffullversion
For our main theorem, we require that two ADL programs are \emph{compliant} (see Definition~\ref{def:compliant}), i.e., use the library in the same way, and the library produces the same distribution as the implementation $\implementation$ from the computational soundness result with a library specification $\libSpec$ (see Definition~\ref{def:harmonizing-crypto-api}).
\fi

\begin{theorem}\label{theorem:symbolic-adl-computational-adl}
Let a symbolic model $\Model$, protocol conditions $C_P$ and implementation conditions $C_I$ that are computationally sound in the sense of CoSP (Definition~\ref{def:computational-soundness-equivalence}) be given. 
    Let $\proto_1$ and $\proto_2$ be two uniform families 
    of ADL programs compliant \iffullversion(see Definition~\ref{def:compliant}) \fi with initial configurations $s_1,s_2$ and the same
    library specification \libSpec w.r.t.\ to a symbolic model \Model
    and all implementations \implementation that satisfy $C_I$.
% we have that for any two ADL programs $\proto_1$ and $\proto_2$ that obey the protocol conditions $C_P$, such that the cryptographic operations (i.e., the parts of the program specified by $\libSpec$) satisfies the implementation conditions $C_I$.
Then, the following implication holds 
$$\protos 1 \ADLEquivalent \protos 2 \implies \protos 1 \ADLIndistinguishable \protos 2.$$
% For two programs $\proto_1$ and $\proto_2$ be in the language ADL, the following implication holds
% $$\proto_1 \ADLEquivalent \proto_2 \implies \proto_1 \ADLIndistinguishable \proto_2$$
\end{theorem}
\iffullversion
\begin{proof}
    Follows by transitivity of $\Rightarrow$ from
    Lemmas~\ref{lemma:symbolic-adl-symbolic-execution},
    \ref{lemma:embedding-soundness-symbolic},
    \ref{lemma:embedding-soundness-computational}, 
    and~\ref{lemma:computational-execution-adl},
    as well as the assumption, that the \implementation is a computationally
    sound implementation of the symbolic model \M.
\end{proof}
\fi

\TODOE{Check that we never tried to use a transitivity of tic-indistinguishability, since it does not hold. However, tic-indistinguishability is implies by statistical indistinguishability (i.e., indistinguishability against an unbounded attacker) and of course by equality of distributions.}
 % 2 pages

%\input{declassification} % 1 page

%!TEX root = main.tex

\section{Related work}\label{section:related-work}

\pparagraph{Operational semantics for Dalvik Bytecode}
We have opted to ground our work on the Abstract Dalvik Language (ADL)~\cite{LoMaStBaScWe_14:cassandra}.
ADL currently excels over alternative semantics such as the ones proposed by Wognsen et al.~\cite{WoKaOlHa_14:formalisation},
Xia et al.~\cite{XiGoLyQiLi_15:android-auditing},
TaintDroid~\cite{EnGiChCoJuDaSh_14:taintdroid}, and
Chaudhuri~\cite{Ch_09:language-based-android} because of its comprehensive treatment of the Dalvik language, even though ADL
currently only provides sequential executions and does not support exceptions.
TaintDroid is more general in this respect in that it contains an
ad-hoc modelling of concurrent execution. However, in the ADL extension put forward
with the adversary that includes probabilistic choices and adversary interactions, concurrency in ADL 
can be modelled via program transformations, as we discuss in 
Section~\ref{sec:soundness}.

% \begin{itemize}
% \item Cassandra~\cite{LoMaStBaScWe_14:cassandra},
% \item Wognsen et al.~\cite{WoKaOlHa_14:formalisation},
% \item Xia et al.~\cite{XiGoLyQiLi_15:android-auditing},
% \item intermediate representation Tyde~\cite{OcJhMc_12:android-java},
% \item TaintDroid~\cite{EnGiChCoJuDaSh_14:taintdroid} semantics has no if-then-else instructions (if-test and if-testz),
% \item Chaudhuri~\cite{Ch_09:language-based-android} first work but too abstract, e.g., forking is modeled as a primitive language construct whereas in Android it is a framework-call that spawns a new thread in the VM (or, in ART, a new thread in the OS).
% \end{itemize}

% Operational semantics for Java:
% \begin{itemize}
% \item Snelting et al.
% \item Barthe et al.
% \end{itemize}

\pparagraph{Information-flow control with cryptographic primitives}
The standard notion of security in information flow control -- non-interference --
is  too strong when cryptographic operations
are being  considered, as, \eg, an encryption of a secret key and
a secret message is intuitively safe to be stored in a public
variable, but it nontheless results in different values depending on the key and the
message.
While the
declassification of values (see~\cite{DBLP:journals/jcs/SabelfeldS09}
for an introduction and overview of results) can be
used to relax this notion, it is difficult to decide under what circumstances
cryptographic values can be safely declassified. 

Our approach is similar to Askarov~\etAl's,  permitting so-called 
\emph{cryptographically
masked flows} by considering a relaxed equivalence notion on public
values, masking acceptable information flow when ciphertexts are made
public~\cite{DBLP:journals/tcs/AskarovHS08}.
\iffullversion%
Instead of actually comparing low-values, e.g.,
encryptions of secret messages, an equivalence relation relates
them, e.g., if the encrypted messages have equal length, both
encryptions are considered low-equivalent. A security-type system is
introduced, which is sound with respect to possibilistic
non-interference and said notion of low-equivalence.
\fi%

Laud has shown  the computational soundness of this
approach~\cite{DBLP:conf/popl/Laud08} if the employed encryption scheme satisfies
key-dependent message security, provides plaintext integrity, and \iffullversion the program is
\emph{well-structured}, \ie,\fi keys are only used in the correct key position, etc.  As cryptographically masked flows are captured in a possibilistic
setting, leaks through probabilistic behavior are not captured; hence the
program must not be able to branch on probabilistic values. 
\iffullversion% 
Consider the following program:
\[
    x := \mathsf{rnd} (0 , 1);~\mathsf{if}~x~\mathsf{then}~l := h~\mathsf{else}~l := \mathsf{rnd} (1 , 100)
\]
While the program is non-interferent in the possibilistic setting, as
any final value $l$ has could be random choice in the second branch,
it is not non-interferent (insecure) in the probabilistic setting as
it is very likely that $l$ indeed contains the value of $h$.

The consequence for this work is that, since we consider
non-interference in the possibilistic setting, we
cannot allow any computation on cryptographic values, as these will
allow to implement (something close to) $\mathsf{rnd}$ using
probabilistic operation like encryption. 
\fi%
In our work, we treat computations performed on cryptographic values 
conservatively by relaying them to the adversary. 
Furthermore, previous work introduced security type-systems, such as~\cite{DBLP:conf/fct/LaudV05}, and program
analyses, such as~\cite{DBLP:conf/esop/Laud03}, that directly operate on the
computational semantics and are thus capable of verifying non-interference in
a probabilistic setting. This approach avoids the aforementioned problem at the cost of less
modularity and less potential for automation than our approach.

\TODOR{More through analysis in terms of precision. 
    We can: branch on secrets (but consider both branches possible),
    apply operations to decryptions, 
apply operations to secrets (but leak the secret in this case)}
%TODO so where is their disadvantage?
%1. we can use cosp to easily support more primitives
%2. maybe we allow for easier analyses? not sure.
%TODO http://cyber.ee/uploads/2013/05/Haav_Laud_nordsec08.pdf,
%Section 2 gives an excellent introduction. need to read papers cited
%there.
%TODO intro in http://users.cis.fiu.edu/~smithg/papers/fmse06.pdf
%gives good reason why declassification is not well suited.
\pparagraph{Computational Soundness}
On computational soundness, there is a rich body of literature for various cryptographic primitives~\cite{BaPW4_03,CoWa_05,GaGaRo_08:commitments,BaBeMaMoPe_15:malleable-zk}, with malicious keys \cite{CoCoSc_12,BaMaUn_12:without-restrictions-cosp}, and even composable computational results~\cite{CoWa_11:composable,BoCoWa_13:for-free}.
Even though these works covered the applied $\pi$-calculus~\cite{CoCo_08:observational-equivalence,CoCoSc_12,BaHoUn_09:cosp}, the stateful applied $\pi$ calculus~\cite{ShQiFe_16:stateful-pi} and RCF~\cite{BaMaUn_10:rcf} and even a embedding of a fragment of C~\cite{AiGoJu_12:embedding-c}, none of these works provide a computational
soundness result is known for Dalvik bytecode and hence for Android app analysis.

Another line of work does not provide a complete symbolic characterization of the attacker but concentrates on single rules that hold for certain symbolic terms~\cite{BaCo_12:complete-attacker,BaCo_14:complete-equivalence}. This restriction enables a much more flexible and composable preservation notion but pays with omitting any guarantee that the set of rules characterizes all attacker-actions. Hence, it is not clear how well-suited this approach is to automation.

\pparagraph{Interactive proof assistants for cryptography}
Cryptographic proof assistants enable the mechanized verification of cryptographic proofs, without first abstracting the cryptographic operations~\cite{BaGrHeBe_11:easycrypt,Pe_15:fcf,Lo_16:isabelle-crypto}. Consequently, these tools only offer limited automation. Yet complementarily, these tools could be used to verify that a library satisfies the conditions that a computational soundness result requires.
 %0.5 page

%!TEX root = main.tex
\section{Conclusion and future work}%
\label{sec:conclusion}
We have shown how cryptographic operations can be faithfully included into existing approaches for automated app analysis.
Our results overcome the overly pessimistic results that arise when current automated approaches deal with apps that
contain cryptographic operations, as these results do not account for secrecy properties offered by cryptographic operations such
as encryption. We have shown how cryptographic operations can be expressed as symbolic abstractions within Dalvik bytecode, so
that these abstractions can be conveniently added to existing app analysis tools using minor changes in their semantics. 
%Casting these symbolic abstractions in Dalvik in particular required us to derive a novel semantic characterization of Dalvik bytecode 
%as a split-state semantics that provides a clear separation between cryptographic API calls and corresponding augmented adversarial symbolic capabilities
%as well as the existing non-cryptographic part of a given program. We consider this representation of independent interest.
Moreover, we have established the first computational soundness result for the Abstract Dalvik Language (ADL)~\cite{LoMaStBaScWe_14:cassandra}, which
currently  
constitutes the most detailed and comprehensive operational
semantics for Dalvik in the literature.

A result that we only scratched on in this work is that  any small-step semantics expressed in our novel split-state form entails a canonical small-step semantics for a symbolic model that is computationally sound. This hence provides a recipe for establishing computationally sound symbolic abstractions for \emph{any} given programming language, provided that one can show that the interaction with the attacker and the cryptographic operations can be expressed by means of our concept of split-state semantics. We plan to further investigate this claim
and its applicability to modern programming languages in future work.

 % 0.25 pages

\section{Acknowledgements}

This work has been partially funded by the German Research Foundation
(DFG) via the collaborative research center ``Methods and Tools for
Understanding and Controlling Privacy'' (SFB 1223), project B3.
This work has been partially supported by the Zurich Information Security Center (ZISC).

%%%%%%%%%%%%%%%%%%%%%%%%%%%%%%%%%%%
%%% The main body ends here 
%%%%%%%%%%%%%%%%%%%%%%%%%%%%%%%%%%%
%%%%%%%%%%%%%%%%%%%%%%%%%%%%%%%%%%%%%%%%%%

%%%%%%%%%%%%%%%%%%%%%%%%%%%%%%%%%%%
%%% The bibliography begins here 
%%%%%%%%%%%%%%%%%%%%%%%%%%%%%%%%%%%

\iffullversion
\else
\clearpage
\fi
\bibliographystyle{IEEEtran}

\bibliography{local}
%%%%%%%%%%%%%%%%%%%%%%%%%%%%%%%%%%%
%%% The bibliography ends here 
%%%%%%%%%%%%%%%%%%%%%%%%%%%%%%%%%%%
\iffullversion
{}
\else
\end{document}
\fi
\clearpage
\appendix
\addcontentsline{toc}{section}{Appendix}
\begin{center}
\LARGE\textbf{Appendix}
\end{center}

%!TEX root = main.tex

\section{Inference rules for the DEX Bytecode Semantics}
\label{app:adl}

\begin{framedFigureSingle}{Semantics of control flow instructions}{fig:control-flow-instructions}
\begin{inferenceRules}
\inference[rNop]{m[pp] = \progcmd{nop}}{\stateDV \redRelSt \stateDV[h,pp+1,r]}\\
\\%
\inference[rGoto]{m[pp] = \progcmd{goto} n}{\stateDV \redRelSt \stateDV[h,pp+n,r]}\\
\end{inferenceRules}
\end{framedFigureSingle}

\begin{framedFigure}{Semantics of object-related instructions}{fig:object-related-instructions}
\begin{inferenceRules}
\inference[rInstanceOfTrue]%
{m[pp] = \progcmd{instance\text{-}of} v_a,v_b,cl & r(v_b) \in \dom(h) \\%
assignmentCompatible(h(r(v_b))).\class, cl}%
{\stateDV \redRelSt \stateDV[h,pp+1,r{[v_a \mapsto 1]}]}\\
\\%
\inference[rInstanceOfFalse]%
{m[pp] = \progcmd{instance\text{-}of} v_a,v_b,cl & r(v_b) \in \dom(h) \\%
\lnot(assignmentCompatible(h(r(v_b)).\class,cl))}%
{\stateDV \mapsto \stateDV[h,pp+1,r{[v_a \mapsto 0]}]}\\
\\%
\inference[rNewInstance]%
{m[pp] = \progcmd{new\text{-}instance} v_a,cl & h\in \dom(nextFreeLocation) \\%
l = nextFreeLocation(h)}%
{\stateDV \redRelSt \stateDV[{h[l\mapsto defaultObject(cl)], pp+1, r[v_a \mapsto l]}]}\\
\\%
\inference[rConstString]%
{m[pp] = \progcmd{const\text{-}string} v_a,s & s\in \dom(nameToReference)}%
{\stateDV \redRelSt \stateDV[{h,pp+1,r[v_a \mapsto nameToReference(s)]}]}\\
\\%
\inference[rConstClass]%
{m[pp] = \progcmd{const\text{-}class} v_a,cl & cl\in \dom(nameToReference)}%
{\stateDV \redRelSt \stateDV[{h,pp+1,r[v_a \mapsto nameToReference(cl)]}]}\\
\\%
\inference[rIget]{m[pp] = \progcmd{iget} v_a,v_b,fid & fid \in \dom(lookup\text{-}field_\proto) \\%
r(v_b) \in \dom(h) & o = h(r(v_b)) \\%
f = lookup\text{-}field_\proto(fid) & f \in \dom(o.fields)}%
{\stateDV \redRelSt \stateDV[{h,pp+1, r[v_a\mapsto o.f]}]}\\
\\%
\inference[rIput]%
{m[pp] = \progcmd{iput} v_a,v_b,fid & fid \in \dom(lookup\text{-}field_\proto) \\%
r(v_b) \in \dom(h) & o = h(r(v_b))\\%
f=lookup\text{-}field_\proto(fid) & f\in \dom(o.fields)}%
{\stateDV \redRelSt \stateDV[{h[r(v_b) \mapsto o[f\mapsto r(v_a)]],pp+1,r}]}\\
\\%
\inference[rSget]%
{m[pp] = \progcmd{sget} v_a,fid & fid\in \dom(nameToReference) \\%
l=nameToReference(fid) & fid\in \dom(lookup\text{-}field_\proto)\\%
f = lookup\text{-}field_\proto(fid) %
& f\in \dom(h(l).fields) %
& u = h(l).f}%
{\stateDV \redRelSt \stateDV[{h,pp+1,r[v_a \mapsto u]}]}\\
\\%
\inference[rSput]%
{m[pp] = \progcmd{sput} v_a,fid & fid \in \dom(nameToReference) \\%
l = nameToRefrence(fid) & fid \in \dom(lookup\text{-}field_\proto)\\%
f=lookup\text{-}field_\proto(fid) & o=h(l) & f\in \dom(o.fields)}%
{\stateDV \redRelSt \stateDV[{h[l \mapsto o[f\mapsto r(v_a)]],pp+1,r}]}
\end{inferenceRules}
\end{framedFigure}

\begin{framedFigure}{Semantics of array-related instructions}{fig:array-related-instructions}
\begin{inferenceRules}
\inference[rArrayLength]%
{m[pp] = \progcmd{array\text{-}length} v_a,v_b & r(v_b) \in \dom(h)}%
{\stateDV \redRelSt \stateDV[h,pp+1,r{[v_a \mapsto h(r(v_b)).length]}]}\\
\\%
\inference[rNewArray]%
{m[pp] = \progcmd{new\text{-}array} v_a,v_b & h \in \dom(nextFreeLocation)\\%
l = nextFreeLocation(h) & 0\le r(v_b)}%
{\stateDV \redRelSt \stateDV[{h[l \mapsto defaultArray(r(v_b))],pp+1,r[v_a \mapsto l]}]}\\
\\%
\inference[rFilledNewArrayR]%
{m[pp] = \progcmd{filled\text{-}new\text{-}array\text{-}range} v_k,n & h \in \dom(nextFreeLocation)\\%
l = nextFreeLocation(h) & x = defaultArray(n)\\%
ar = x[0\mapsto r(v_k), \dots, n-1 \mapsto r(v_{k+n-1})]}%
{\stateDV \redRelSt \stateDV[{h[l \mapsto ar],pp+1,r[\resultLowerDV \mapsto l]}]}\\
\\%
\inference[rAget]%
{m[pp] = \progcmd{aget} v_a,v_b,v_c & r(v_b) \in \dom(h) & ar = h(r(v_b))\\%
u = ar[r(v_c)] & 0 \le r(v_c) < ar.length}%
{\stateDV \redRelSt \stateDV[{h,pp+1,r[v_a \mapsto u]}]}\\
\\%
\inference[rAput]%
{m[pp] = \progcmd{aput} v_a,v_b,v_c & r(v_b) \in \dom(h) & ar = h(r(v_b))\\%
x = ar[r(v_c) \mapsto r(v_a)] & 0 \le r(v_c) < ar.length}%
{\stateDV \redRelSt \stateDV[{h[r(v_b) \mapsto x],pp+1,r}]}%
\end{inferenceRules}
\end{framedFigure}

\begin{framedFigure}{Semantics of method-related instructions}{fig:method-related-instructions}
\begin{inferenceRules}
\inference[rIVR]%
{m[pp] = \progcmd{invoke\text{-}virtual\text{-}range} v_k,n,mid & r(v_k) \in \dom(h)\\%
(mid,h(r(v_k)).\class) \in \dom(\lookupVirtualDV_\proto) & m' = \lookupVirtualDV_\proto(mid,h(r(v_k)).\class)\\%
\stateDV[{h,0,defaultRegisters([r(v_k), \dots, r(v_{k+n-1})])}] \Downarrow^{(n')}_{P,m'} \stateDV[{u,h'}]}%
{\stateDV \redRel{n'+1}{P}{m} \stateDV[{h',pp+1,r[\resultLowerDV \mapsto lower(u), \resultUpperDV \mapsto upper(u)]}]}\\
\\%
\inference[rIStR]%
{m[pp] = \progcmd{invoke\text{-}static\text{-}range} v_k,n,mid \\%
mid \in \dom(\lookupStaticDV_\proto) & m' = \lookupStaticDV_\proto(mid)\\%
\stateDV[{h,0,defaultRegisters([r(v_k), \dots, r(v_{k+n-1})])}] \Downarrow^{(n')}_{P,m'} \stateDV[{u,h'}]}%
{\stateDV \redRel{n'+1}{P}{m} \stateDV[{h',pp+1,r[\resultLowerDV \mapsto lower(u), \resultUpperDV \mapsto upper(u)]}]}\\
% \\%
% \inference[rMoveR]%
% {m[pp] = \progcmd{move\text{-}range} v_a}%
% {\stateDV \redRelSt \stateDV[{h,pp+1,r[v_a \mapsto r(\resultLowerDV)]}]}\\
% \\%
% \inference[rReturnVoid]%
% {m[pp] = \progcmd{return\text{-}void}}%
% {\stateDV \redRelSt \stateDV[{void,h}]}\\
% \\%
% \inference[rReturn]%
% {m[pp] = \progcmd{return} v_a}%
% {\stateDV \redRelSt \stateDV[{r(v_a),h}]}%
%
\end{inferenceRules}
\end{framedFigure}

\begin{framedFigureSingle}{Semantics of conversion instructions for 64 bit values}{fig:conversion-sixty-four-instructions}
\begin{inferenceRules}
\inference[rUnopWideS]%
{m[pp] = \progcmd{unop\text{-}wideS} v_a,v_b,uop \\%
u=\underline{uop}(r(v_b) \bullet r(v_{b+1}))}%
{\stateDV \redRelSt \stateDV[{h,pp+1,r[v_a\mapsto u]}]}\\
\\%
\inference[rUnopWideT]%
{m[pp] = \progcmd{unop\text{-}wideT} v_a,v_b,uop \\%
u=\underline{uop}(r(v_b))}%
{\parbox{5cm}{$\stateDV \redRelSt$ \\ $\stateDV[{h,pp+1,r[v_a\mapsto lower(u),v_{a+1}\mapsto upper(u)]}]$}}
\end{inferenceRules}
\end{framedFigureSingle}

%!TEX root = main.tex

\section{Extendability}\label{app:extendability}

\newcommand{\transparent}{\mathcmd{trp}}
\newcommand{\restrictionDefined}[2]{\mathcmd{{\cal P}(#1,#2)}}

\begin{definition}[Efficient transparent function]
An n-ary constructor $f$ is \emph{transparent} if for every argument $i\in\set{1,\dots,n}$ there is a destructor $d_{inverse,i}$, called \emph{the $i$-th inverse function of $f$}, such that for all terms $M_1, \dots, M_n$, we have $\eval_{d_{inverse,i}}{f(M_1, \dots, M_n)} = M_i$.
\end{definition}
\NOTEE{The inverse function could be extended from destructors to terms, but then the proof of Claim 1 becomes for involved.}

\begin{definition}[Finitely generated message types]
A message type $\T$ is \emph{finitely generated} if there exists a grammar with rules $r_1, \dots, r_n$ such that $\T$ equals the set that is generated by $r_1, \dots, r_n$.\footnote{The notion of a set that is generated by a grammar is a common notion. It refers to the smallest fixpoint that obeys the rules from the grammar.}
\end{definition}

\begin{definition}[Combination of two finitely generated message types]
Let $\T$ be a finitely generated message types over $\C$ and $\N$ with a grammar $r_1, \dots, r_n$, and let $\T'$ be a 
finitely generated message types over $\C'$ and $\N'$ with a grammar $r_1', \dots, r_m'$ such that $\C \cap \C' = 
\emptyset$ and $\N \cap \N' = \emptyset$. Then, the \emph{combination $\T \sqcup \T'$ of $\T$ and $\T'$} is the set that 
is generated by the rules $r_1, \dots, r_n, r_1', \dots, r_m'$.
\end{definition}

\begin{definition}[Combination of two symbolic models]
Let a symbolic model $\Model = (\C,\NE \uplus \NP,\T,\D)$ and another symbolic model $\Model' = (\C',\NE' \uplus \NP',\T',\D')$ with $\C\cap \C' = \emptyset$, $\D \cap \D' = \emptyset$, $\NE \cap \NE' = \emptyset$, $\NP \cap \NP' = \emptyset$, $\T \cap \T' = \emptyset$ be given. Then, the \emph{combination of $\Model$ and $\Model'$} is defined as $\Model \sqcup \Model' := (\C\cup\C', \N = \NE \cup \NE' \uplus \NP \cup \NP', \T \sqcup \T', \D \cup \D')$.
\end{definition}

\begin{definition}[Restriction-defined protocol class]\label{def:restriction-defined-class}
Let a symbolic model $\Model = (\C,\N,\T,\D)$ be given. We call a set $R := \set{c_1, \dots, c_n}$ of functions from CoSP protocols to $\set{\top,\bot}$, i.e., true or false, \emph{protocol restrictions}. A protocol class $\protoClass$ is \emph{restriction-defined} by $R$ on $\Model$, written as $\protoClass := \restrictionDefined{R}{\Model}$, if $\protoClass$ is the largest set of protocols for $\Model$ (see Definition~\ref{def:cosp-protocol}) such that $\forall \proto \in \protoClass. \forall i\in\set{1,\dots,n}. c_i(\proto) = \top$.
\end{definition}

\begin{lemma}\label{lemma:extendability}
Given a symbolic model $\Model$ with a implementation $\implementation$ of $\Model$. Let $\Model_\transparent = (\C',\N',\T',\D')$ be a symbolic model with an implementation $\implementation_{\transparent}$ such that all constructors $f\in\C'$ are transparent functions. Let $R$ be protocol restrictions (see Definition~\ref{def:restriction-defined-class}).

If the implementation $\implementation$ is computationally sound for $\Model$ and the protocol class $\restrictionDefined{R}{\Model}$, then the implementation $\implementation \cup \implementation_{\transparent}$ is computationally sound for the symbolic model $\Model \sqcup \Model_{\transparent}$ and the protocol class $\restrictionDefined{R}{\Model \sqcup \Model_{\transparent}}$.
\end{lemma}

\begin{proof}
By contraposition it suffices to show that for any pair of protocols $\proto_1',\proto_2' \in \restrictionDefined{R}{\Model \sqcup \Model_{\transparent}}$, the following holds
$$\proto_1' \notCoSPIndistinguishable \proto_2' \implies \proto_1' \notCoSPEquivalent \proto_2'$$
Towards a contradiction, assume that
$$\proto_1' \notCoSPIndistinguishable \proto_2' \land \proto_1' \CoSPEquivalent \proto_2'$$
holds.
Our strategy is to show that there are protocols $\proto_1, \proto_2 \in \restrictionDefined{R}{\Model}$ such that
\begin{enumerate}
\item[] $\proto_1' \notCoSPIndistinguishable \proto_2' \implies \proto_1 \notCoSPIndistinguishable \proto_2$ (Claim 1)
\item[] $\proto_1 \notCoSPEquivalent \proto_2 \implies \proto_1' \notCoSPEquivalent \proto_2'$ (Claim 2)
\end{enumerate}
holds. These two statements imply that
$$\proto_1 \notCoSPIndistinguishable \proto_2 \land \proto_1 \CoSPEquivalent \proto_2$$
holds, which contradicts the computational soundness of $\implementation$ for $\Model$ and $\restrictionDefined{R}{\Model}$.

\begin{Claim}{1}
$\proto_1' \notCoSPIndistinguishable \proto_2' \implies \proto_1 \notCoSPIndistinguishable \proto_2$
\end{Claim}
\begin{claimproof}{1}
First, we define the \emph{encoding of an n-ary transparent constructor $f$} as follows:
$$\mathsymb{encoding-f}(\nu_1,\dots,\nu_n) := \pair(\mathstring{constructor-f},\pair(\nu_1, \pair(, \dots, \pair(\nu_{n-1},\nu_n) )))$$
Analogously, we define the \emph{encoding of the $i$-th inverse function $d_{inverse,i}$ of $f$} as the sequence of the $\fst$ and $\snd$ destructor that projects to the $i$-th argument. We denote this encoding as $\mathsymb{encoding-d-inverse-i}(\nu_1)$.

We inductively define $\proto_i$ to equal $\proto_i'$ except for occurrences of computation nodes $\nu$ that are labelled with a transparent constructor or an inverse function, which are in $\proto_i$ replaced by the sequence of computation nodes that represents the encoding of the transparent constructor or its inverse function.

The claim follows since for all terms $t_1, \dots, t_n$ the term $\mathsymb{encoding-d-inverse-i}(\mathsymb{encoding-f}(t_1,\dots,t_n))$ evaluates to $t_i$.
\TODOE{Do we have to require that no other destructor modifies $f$?}
\end{claimproof}

\begin{Claim}{2}
$\proto_1 \notCoSPEquivalent \proto_2 \implies \proto_1' \notCoSPEquivalent \proto_2'$
\end{Claim}
\begin{claimproof}{2}
Analogous to the proof of Claim 1, Claim 2 follows since for all terms $t_1, \dots, t_n$ the term $\mathsymb{encoding-d-inverse-i}(\mathsymb{encoding-f}(t_1,\dots,t_n))$ evaluates to $t_i$.
\end{claimproof}

\end{proof}

\subsection{Arbitrary bit-operations on bitstrings}
We can extend any computationally sound symbolic model that contains the basic bitstring operations $\payloadString_0/1,\payloadString_1/1,\payloadUnstring_0/1,\payloadUnstring_1/1,\payloadEmpty/0$ with
$$\forall b\in\set{0,1}. \payloadUnstring_b(\payloadString_b(x)) = x$$
to a symbolic model that contains all bitstring operations $op: \set{0,1}^* \rightarrow \set{0,1}^*$ as long as the respective destructor is solely defined on symbolic string. The proof encodes any operation on bitstring to the corresponding circuit. Hence, we will restrict ourselves to illustrating how $\land, \lor, \lnot$ gates are encoded on single bits:

\begin{myalgorithm}{$\land(\payloadString_{b_1}(\payloadEmpty),\payloadString_{b_2}(\payloadEmpty))$}
\IF{$\payloadString_{b_1}(\payloadEmpty) = \payloadString_{0}(\payloadEmpty)\land\payloadString_{b_2}(\payloadEmpty) = \payloadString_{0}(\payloadEmpty)$}
	\RETURN $\payloadString_{0}(\payloadEmpty)$
\ELSIF {$\payloadString_{b_1}(\payloadEmpty) = \payloadString_{1}(\payloadEmpty)\land\payloadString_{b_2}(\payloadEmpty) = \payloadString_{1}(\payloadEmpty)$}
	\RETURN $\payloadString_{1}(\payloadEmpty)$
\ELSIF {$\payloadString_{b_1}(\payloadEmpty) = \payloadString_{0}(\payloadEmpty)\land\payloadString_{b_2}(\payloadEmpty) = \payloadString_{1}(\payloadEmpty)$}
	\RETURN $\payloadString_{0}(\payloadEmpty)$
\ELSIF {$\payloadString_{b_1}(\payloadEmpty) = \payloadString_{1}(\payloadEmpty)\land\payloadString_{b_2}(\payloadEmpty) = \payloadString_{0}(\payloadEmpty)$}
	\RETURN $\payloadString_{0}(\payloadEmpty)$
\ENDIF
\end{myalgorithm}

\begin{myalgorithm}{$\lor(\payloadString_{b_1}(\payloadEmpty),\payloadString_{b_2}(\payloadEmpty))$}
\IF{$\payloadString_{b_1}(\payloadEmpty) = \payloadString_{0}(\payloadEmpty)\land\payloadString_{b_2}(\payloadEmpty) = \payloadString_{0}(\payloadEmpty)$}
	\RETURN $\payloadString_{0}(\payloadEmpty)$
\ELSIF {$\payloadString_{b_1}(\payloadEmpty) = \payloadString_{1}(\payloadEmpty)\land\payloadString_{b_2}(\payloadEmpty) = \payloadString_{1}(\payloadEmpty)$}
	\RETURN $\payloadString_{1}(\payloadEmpty)$
\ELSIF {$\payloadString_{b_1}(\payloadEmpty) = \payloadString_{0}(\payloadEmpty)\land\payloadString_{b_2}(\payloadEmpty) = \payloadString_{1}(\payloadEmpty)$}
	\RETURN $\payloadString_{1}(\payloadEmpty)$
\ELSIF {$\payloadString_{b_1}(\payloadEmpty) = \payloadString_{1}(\payloadEmpty)\land\payloadString_{b_2}(\payloadEmpty) = \payloadString_{0}(\payloadEmpty)$}
	\RETURN $\payloadString_{1}(\payloadEmpty)$
\ENDIF
\end{myalgorithm}

\begin{myalgorithm}{$\lnot(\payloadString_{b}(\payloadEmpty))$}
\IF{$\payloadString_{b}(\payloadEmpty) = \payloadString_{0}(\payloadEmpty)$}
	\RETURN $\payloadString_{1}(\payloadEmpty)$
\ELSIF {$\payloadString_{b}(\payloadEmpty) = \payloadString_{1}(\payloadEmpty)$}
	\RETURN $\payloadString_{0}(\payloadEmpty)$
\ENDIF
\end{myalgorithm}

These encodings can be easily extended to bitstrings and used in order to encode arbitrary circuits, hence any computable function $op:\set{0,1}^* \rightarrow \set{0,1}^*$, on bitstrings. Given a bijection $\iota$ from bitstrings to symbolic bitstrings as above, a symbolic binary operation $d_{op}/n$ is a destructor that is solely defined on symbolic bitstrings as
$$d_{op}(x_1,\dots,x_n) = \iota(op(\iota^{-1}(x_1),\dots,\iota^{-1}(x_n)))$$
As a corollary, we get

\begin{corollary}\label{corollary:bitstring-operation-extension}
Given a symbolic model $\Model$ with a implementation $\implementation$ of $\Model$. Let $\Model_{bitstringOp} = (\emptyset,\N',\T',\D')$ be a symbolic model with an implementation $\implementation_{bitstringOp}$ such that all destructors $f\in\D'$ are binary operations. Let $R$ be protocol restrictions (see Definition~\ref{def:restriction-defined-class}).

If the implementation $\implementation$ is computationally sound for $\Model$ and the protocol class $\restrictionDefined{R}{\Model}$, then the implementation $\implementation \cup \implementation_{bitstringOp}$ is computationally sound for the symbolic model $\Model \sqcup \Model_{bitstringOp}$ and the protocol class $\restrictionDefined{R}{\Model \sqcup \Model_{\transparent}}$.
\end{corollary}

\subsection{Derived destructors with symbolic operations}
Similar to transparent functions and bitstring operations on symbolic bitstrings, all destructors that can be represented as a symbolic operation can be added to a computationally sound symbolic model as destructors applied to nonces, i.e., destructor and nonce computation nodes, without loosing computational soundness. With the formal parameters as holes, a symbolic operation $O$, with say $n$ variables, can be interpreted as a context over terms. Such a context naturally defines an $n$-ary destructor that is applied to nonces $r_{k+1}$, \dots, $r_{n}$ via $\eval_O(t_1, \dots, t_k, r_{k+1}, \dots, r_n)$.\footnote{Without loss of generality, we assumed here an order on the parameters.} We call such a destructor a \emph{derived destructor}.

Such derived destructors can be encoded using constructors and destructors from the original symbolic model, we can use the same proof techniques as above and conclude computational soundness for the combined model.

\begin{corollary}\label{corollary:derived-destructors}
Given a symbolic model $\Model$ with a implementation $\implementation$ of $\Model$. Let $\Model_{derDes} = (\emptyset,\N',\T',\D')$ be a symbolic model with an implementation $\implementation_{derDes}$ such that all destructors $f\in\D'$ are derived destructors and the implementation is the evaluation of the symbolic operation, using the implementations of the constructors and destructors that occur in the symbolic operation. Let $R$ be protocol restrictions (see Definition~\ref{def:restriction-defined-class}).

If the implementation $\implementation$ is computationally sound for $\Model$ and the protocol class $\restrictionDefined{R}{\Model}$, then the implementation $\implementation \cup \implementation_{derDes}$ is computationally sound for the symbolic model $\Model \sqcup \Model_{derDes}$ and the protocol class $\restrictionDefined{R}{\Model \sqcup \Model_{\transparent}}$.
\end{corollary}

\end{document}